\documentclass[acmsmall]{acmart}

\AtBeginDocument{%
  }

\newif\ifreport\reporttrue
\usepackage{graphicx}
\usepackage{placeins}
\usepackage{verbatim}
\usepackage{color}
\usepackage{dsfont}
\usepackage{amsthm}
\newtheorem{definition}{Definition}

\newtheorem{theorem}{Theorem}
\newtheorem{lemma}{Lemma}
\newtheorem{proposition}{Proposition}
\newtheorem{corollary}{Corollary}

\usepackage[lined, boxed, linesnumbered, ruled]{algorithm2e}

\usepackage[framemethod=TikZ]{mdframed}

\received{February 2023}
\received[revised]{October 2022}
\received[accepted]{October 2023}
\begin{document}

\setcopyright{acmlicensed}
\acmJournal{POMACS}
\acmYear{2023} \acmVolume{7} \acmNumber{3} \acmArticle{60} \acmMonth{12} \acmPrice{15.00}\acmDOI{10.1145/3626791}

\title[]{Sampling for Remote Estimation of the Wiener Process over an Unreliable Channel}
\author{Jiayu Pan}
\email{pan.743@osu.edu}
\affiliation{%
  \institution{The Ohio State University}
  \city{Columbus}
  \state{OH}
  \country{USA}
  \postcode{43210}
}
\author{Yin Sun}
\email{yzs0078@auburn.edu}
\affiliation{%
  \institution{Auburn University}
  \city{Auburn}
  \state{AL}
  \country{USA}
  \postcode{36849}
}
\author{Ness B. Shroff}
\email{shroff.11@osu.edu}
\affiliation{%
  \institution{The Ohio State University}
  \city{Columbus}
  \state{OH}
  \country{USA}
  \postcode{43210}
}



\begin{abstract}
In this paper, we study a sampling problem where a source takes samples from a Wiener process and transmits them through a wireless channel to a remote estimator. Due to channel fading, interference, and potential collisions, the packet transmissions are unreliable and could take random time durations. Our objective is to devise an optimal causal sampling policy that minimizes the long-term average mean square estimation error. This optimal sampling problem is a recursive optimal stopping problem, which is generally quite difficult to solve. However, we prove that the optimal sampling strategy is, in fact, a simple threshold policy where a new sample is taken whenever the instantaneous estimation error exceeds a threshold. This threshold remains a constant value that does not vary over time. By exploring the structure properties of the recursive optimal stopping problem, a low-complexity iterative algorithm is developed to compute the optimal threshold. This work generalizes previous research by incorporating both transmission errors and random transmission times into remote estimation. Numerical simulations are provided to compare our optimal policy with the zero-wait and age-optimal policies.
\end{abstract}



\maketitle

\section{Introduction}\label{introduction}



Several applications in real-time control of systems involving sensor networks, such as autonomous driving, military networks, intelligent manufacturing, etc., involve sampling and remote estimation of information. For example, in military systems, status information about the instantaneous speed and position of the vehicles, channel conditions, and targets \textcolor{black}{changes over} time. In order to ensure that the system performs efficiently, reliably, and safely,  the controller(s) has to obtain accurate estimates of the current status of the system from nearby sensors. This involves judicious sampling of the information in order to minimize the estimation error. Designing an optimal sampling strategy is a hard problem, because some easy strategies, such as \textcolor{black}{continuous sampling, are} infeasible due to the limited energy resources and can be far from optimality due to the transmission delay, limited channel capacity, etc. 


To that end, in this paper, we aim to embark on a sampling problem of a wireless network, as is illustrated in Fig.~\ref{model}. The sampler takes the sample of a continuous-time source process and transmits the sample to a remote estimator. The continuous-time source process is modeled as the Wiener process $W_t$, which helps describe the dynamics of sensors measuring quantities like movement, providing insights into how these quantities change over time. The Wiener process, also commonly referred to as Brownian motion \cite{morters2010brownian}, is one of the best known Lévy process, that features stationary and independent increments. It finds widespread applications in various fields such as pure and applied mathematics, economics, quantitative finance, evolutionary biology, and physics.
The Wiener process $W_t$ has the following key properties: (i) $W_0=0$; (ii) $W_t$ is continuous; (iii) $W_t$ has independent increments; (iv) $W_t - W_s \sim \mathcal{N}(0,t-s)$ for $0\le s\le t$, where $\mathcal{N}$ denotes the normal distribution. 
The remote estimator, in turn, provides a minimum mean square estimation error (MMSE) estimate $\hat{W}_t$ based on the received samples. The core objective is to control the sequence of sampling times to minimize the estimation error $W_t - \hat{W}_t$, specifically, aiming at optimizing the long term average of MMSE. 

\begin{figure}[t]
\includegraphics[scale=.5]{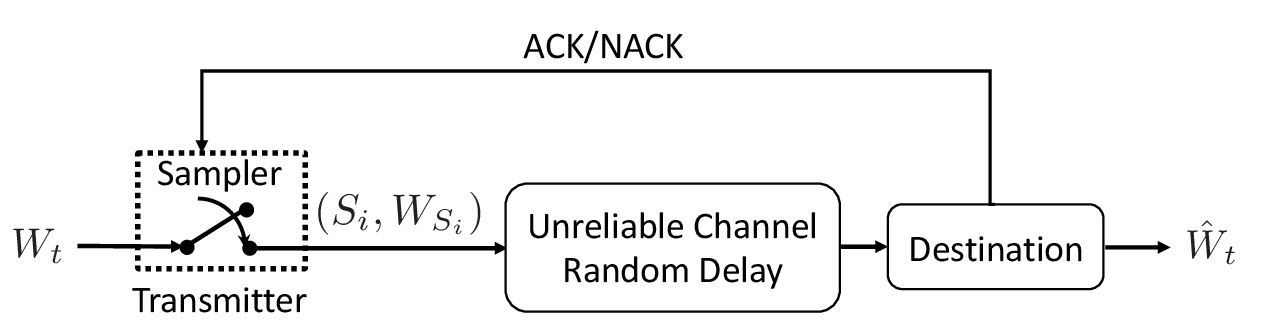}
\centering
\captionsetup{justification=justified}
\caption{System model.}
 \label{model}
\end{figure}

Organized according to the sampling strategies and optimization metrics, our review of related works encompasses three distinct perspectives. 

\subsection{Related Works}

\textbf{Signal-aware sampling with reliable transmissions.}
There have been several studies on sampling for remote estimation, e.g., in \cite{nar2014sampling,sun2020sampling,ornee2019sampling,tang2022sampling,tsai2021unifying}, where the sampling times depend on the source process (signal-aware).
A nice survey paper is included in \cite{jog2021channels}. In \cite{sun2020sampling}, the authors consider the Wiener process as the source process and provide an exact solution to minimize the estimation error. According to the optimal solution, the sampler should wait until the instantaneous estimation error exceeds a threshold, and the threshold is given explicitly. A similar result was developed in \cite{ornee2019sampling} by extending \cite{sun2020sampling} from the Wiener process to the Ornstein Uhlenbeck (OU) process. The optimal threshold retains its simplicity, remaining a root of a closed-form equation.   
Further exploration, as in \cite{tsai2021unifying}, delves into an asymmetric sensor-controller remote system. In this scenario, there are random transmission times \textcolor{black}{in} both directions. At the sensor, the sampling time is a stopping time based on the evolution of the Wiener process, and at the controller, the sampling time depends on the information sent from the sensor. The authors yield precise optimal solutions, noting the potential existence of multiple thresholds for the sensor's optimal stopping time.
Joint optimality designs on the sampling and the estimation, concerning the Wiener process or the autoregressive process, is investigated in \cite{chakravorty2020remote,guo2022optimal}. 

To summarize, except \cite{chakravorty2020remote}, these previous studies on sampling assume reliable transmissions. However, in a variety of wireless systems, channel errors may occur due to fading, and the transmission times of a packet could be random. While packet drops are considered in \cite{chakravorty2020remote}, a time-slotted system is considered, which assumes that the \textcolor{black}{total transmission time is the same as the transmission instance (one time slot)}. In contrast, in this paper, our model allows for both packet errors and random transmission times. \textcolor{black}{Moreover, we enable the selection of real-valued transmission instances.}



  
\textbf{Signal-agnostic sampling.} 
When the sampling times are independent of the Wiener process (signal-agnostic), the MMSE is equal to the age of information \cite{sun2020sampling}. More generally, the MMSE is a function of the age of information under a linear time invariant system \cite{champati2019performance,klugel2019aoi}. Thus, our study is closely related to numerous studies on age-based sampling, e.g., in \cite{arafa2020age,arafa2021timely,sun2019sampling,pan2022optimizing,pan2021minimizings,pan2022optimal,moltafet2022status,hui2022real}. 
Age of information, or simply age, is a metric to evaluate data freshness. Age at current time $t$ is defined as $\Delta(t) = t - U(t)$, where $U(t)$ is the generation time of the latest delivered sample. Age has gained much popularity in the recent decade and has contributed to various remote control systems such as sensor networks, UAV navigation, and semantic communication. A recent literature review on the age is provided in \cite{yates2021age}.

In \cite{arafa2021timely}, the paper studies sampling energy harvesting sources with a unit battery buffer under an erasure channel. In the case of a single source, it provides an optimal sampling policy without feedback. With perfect feedback, an optimal policy is offered among the policies that may wait only when the previous transmission is successful. In \cite{arafa2020age}, the paper solves explicit optimal solutions for an energy harvesting source with finite buffer sizes, where the arrived energy can fill up the whole buffer or fill up incrementally.     
In \cite{sun2019sampling}, the authors relate autocorrelation, remote estimation, and mutual information to the nonlinear age penalty functions, and provide an optimal sampling policy under sampling rate constraint. 
In \cite{pan2022optimal}, when the source process is a multidimensional Gaussian diffusion process, and the estimator is the Kalman Filter, the expected square estimation error is an increasing function of the age. For a general non-decreasing age penalty function, the optimal sampling policy has a threshold structure under unreliable channel conditions and random transmission delay. An extended sampling scenario where the sampler can transmit the sample before receiving the feedback is studied in \cite{moltafet2022status}.


 However, compared to the signal-aware sampling policies, signal-agnostic counterparts exhibit suboptimal performance in terms of minimizing the estimation error. Numerical results in \cite{sun2020sampling} validate that the optimal signal-aware sampling policy can achieve less than half of the long term average MMSE than that of the age-optimal sampling policy. This is intuitive, due to the criticality of the content of information within remote monitoring systems, such as the pedestrian intentions in vehicular networks or target locations in UAV navigations.



 \textbf{AoII-optimal scheduling}.
 Recently, researchers have studied signal-aware policies to optimize a new metric: the age of incorrect information (AoII) \cite{chen2023age,kam2020age,maatouk2020age}. AoII incorporates both the content of information (estimation error) and the freshness of information (data freshness).
AoII was first advanced in \cite{maatouk2020age}, serving as a cornerstone for subsequent research. In the context of a finite symmetric Markov source, \cite{maatouk2020age} provides the transmission strategy with a focus on minimizing the AoII, displaying low computational complexity. In \cite{kam2020age}, the authors employ dynamic programming to minimize the AoII under a binary Markovian source and exponential channel delay distribution. Meanwhile, the paper in \cite{chen2023age} extends \cite{kam2020age} to a general transmission time distribution, showing that it is optimal to always transmit whenever the channel is idle and the AoII is not zero. 

Although these studies focus on content-aware transmission strategies, they all focus on a finite state Markov source under a discrete-time system. These scenarios restrict transmission choices between transmit and idle at the beginning of each time slot. Instead, we consider an unbounded and continuous-time Markov process, enabling the selection of real-valued transmission instances.    
\subsection{Our Contributions}
In comparison to these three prevailing perspectives, in this paper, we consider a scenario of minimizing the estimation error of the Wiener process. Specifically, we (i) embrace a signal-aware sampling policy and (ii) accommodate an unreliable channel with a random transmission time. 
Our contributions expand on \cite{sun2020sampling} by considering an unreliable channel, and \cite{pan2022optimal} by allowing sampling time dependence on the content of the Wiener process. 
Our problem belongs to a semi-Markov decision problem and is difficult to solve. There have been solutions for some special cases. In the first case where the channel is reliable (e.g., \cite{sun2019sampling,ornee2019sampling,sun2020sampling,tsai2021unifying}), the original problems are reduced to a single sample problem, which can be further solved by convex optimizations or optimal stopping rules. However, these methods do not hold in our case because our new problem is decoupled to a \emph{recursive optimal stopping problems} with multiple samples\footnote{Also, our problem is significantly different from that with instantaneous transmission time, e.g., \cite{guo2022optimal}, because even if there is no sampling rate constraint, the zero-wait sampling policy is not optimal.}.
Similarly, our work is different from \cite{pan2022optimal}, because this problem is decoupled to a discounted MDP, and each action of the MDP is not a stopping time. 
Nonetheless, we are able to circumvent these challenges and solve the optimal sampling problem.
The main contributions of this paper are stated as follows:

\begin{itemize}


\item 
We provide an exact solution to our optimal sampling problem. The optimal sampling strategy has a simple structure: each sampling time is a stopping time that takes the sample when the instantaneous estimation error exceeds a threshold. The optimal threshold remains the same, independent of the Wiener process value and whether the last transmission failed or not. Moreover, the optimal threshold can be solved efficiently, e.g., by using a two layer bisection search algorithm. Our results hold for general distributions of the transmission delay and arbitrary probability of the i.i.d. transmission failure. 
To solve our \emph{recursive optimal stopping problems}, we developed new approaches. We provide an exact value function to the value iteration problem.
Specifically, we solve a sequence of optimal stopping problems, where the action value function implies taking an action at the first sample and taking the optimal stopping times at the remaining samples. The technical tools used to establish the results include (a) the strong Markov property and Martingale properties of the Wiener process, (b) Shiryaev's free boundary method for solving optimal stopping problems.  

\item When the sampling time does not depend on the Wiener process, the expected square estimation error is equal to the age \cite{sun2020sampling}, and our original problem is equivalent to an age minimization problem. We provide the exact solution as well. The sampler takes the sample when the age first exceeds a threshold. This result also improves \cite[Theorem~1]{pan2022optimal} by removing the assumption of the regenerative process. 

\item Numerical simulations are provided to validate our results. An interesting observation is that when the channel is highly unreliable, 
our optimal policy still performs much better than the age-optimal and zero-wait policies.  

\end{itemize}

\section{Model and Formulation}\label{estimation}

\begin{figure}[t]
\includegraphics[scale=.45]{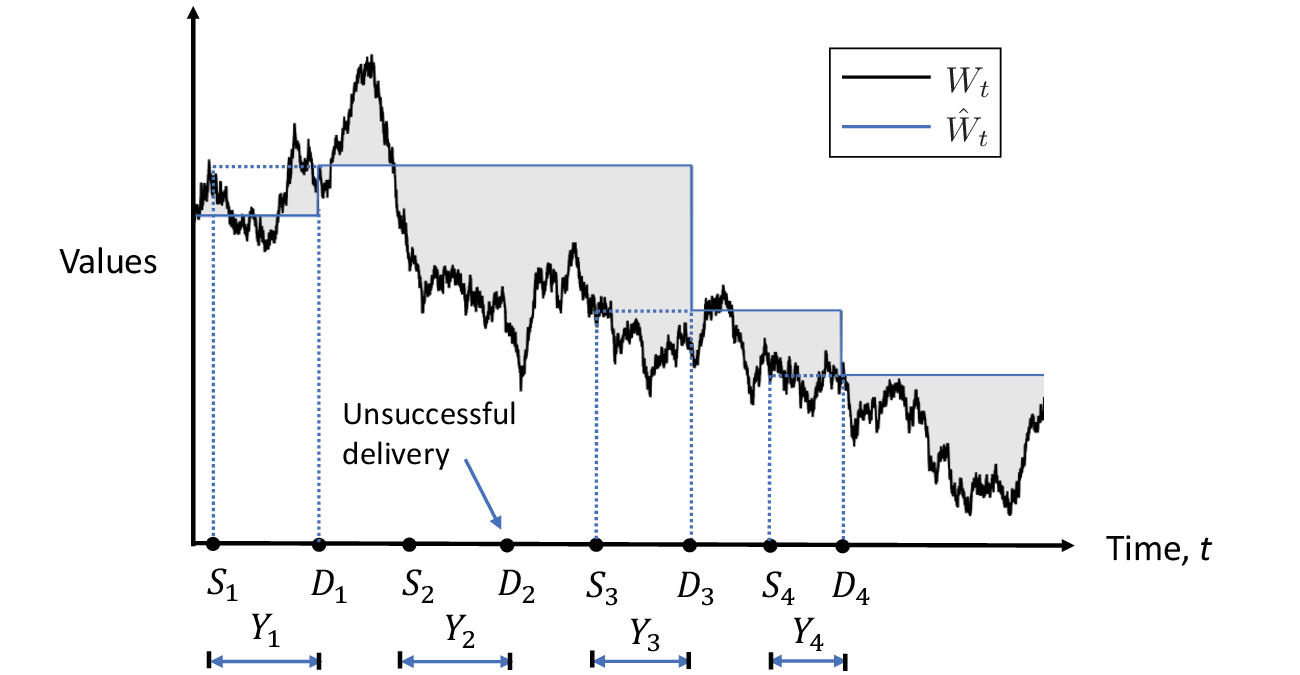}
\centering
\captionsetup{justification=justified}
\caption{A sample path of the Wiener process $W_t$ and the MMSE $\hat{W}_t$ over time $t$. At $D_1,D_3,D_4$, the sample is successfully delivered, so $\hat{W}_t$ is updated to be $W_{S_1},W_{S_3},W_{S_4}$, respectively. At $D_2$, the sample is not successfully delivered, so $\hat{W}_t$ remains unchanged. }
 \label{fig-mmse}
\end{figure}

\subsection{System Model and MMSE Estimator}\label{system-a}

We consider a continuous-time status update system as is depicted in Fig.~\ref{model}, where a sampler takes the sample from the Wiener process $W_t$ and transmits to a destination through an unreliable channel. The destination provides an estimate $\hat{W}_t$ based on the samples that have been successfully delivered. The extended setting from a reliable channel to an unreliable channel is one of the key features of our study.

We use $i \in \{1,2,\ldots\}$ to indicate the number of samples generated by the sampler. The $i$th sample is generated at time $S_i$ and is transmitted through the unreliable channel. The sample contains the sampling time $S_i$ and the sample value $W_{S_i}$. The unreliable channel has an i.i.d. transmission failure, and we denote $\alpha \in [0,1)$ as the probability of failure (i.e., the channel condition is \emph{OFF}). The channel also has an i.i.d. transmission time $Y_i$, and we have $\mathbb{E}[Y^2_{i}]<\infty$. The transmission time and the channel condition are mutually independent. In this paper, we also assume that the transmission time is lower bounded, i.e., there exists $\epsilon>0$ (which can be sufficiently small) such that $Y_{i} \ge \epsilon$. The $i$th sample is delivered to the destination at time $D_i$, where $D_i = S_i+Y_i$. At the delivery time $D_i$, the destination knows the outcome of the transmission of the $i$th sample. Only if the transmission was successful, the destination receives the sample message $(S_i, W_{S_i})$. 
In addition, at $D_{i}$, the destination then sends an acknowledgment back to the sampler, \textcolor{black}{informing} whether the transmission of the $i$th sample was successful or not. We assume that the transmission process of the acknowledgment is instantaneous and error free. 
Note that the sampler always generates a sample after it receives feedback, i.e., $S_{i+1} \ge D_{i}$. Otherwise, the generated sample will be queued for waiting to be transmitted, and the queued sample is staled compared to the fresh sample.

The estimator (destination) also provides a minimum mean square error (MMSE) estimator $\hat{W}_t$ based on the successfully received samples until time $t$.

We denote the random variable $\underline{i}$ as the index of the latest sample that is \emph{successfully} delivered to the destination by the time $D_i$. In the special case of a reliable channel, each sample is successfully delivered, so we have $\underline{i}=i$; otherwise, $\underline{i} \le i$.  The latest (and thus freshest) sample the destination has received during $t\in [D_i,D_{i+1})$ is $( S_{\underline{i}},W_{S_{\underline{i}}})$. 
Using the strong Markov property of the Wiener process \cite[Eq. (4.3.27)]{peskir2006optimal}, the MMSE estimator $\hat{W}_t$ is expressed as
\begin{align} 
\hat{W}_t = \mathbb{E}[W_t | S_{\underline{i}},W_{S_{\underline{i}}} ] = \mathbb{E}[W_t | W_{S_{\underline{i}}} ] = W_{S_{\underline{i}}}, t\in [D_i,D_{i+1}). \label{eq-msei}
\end{align}
A sample path of $W_t$, $\hat{W}_t$, and the estimation error $W_t - \hat{W}_t$ \textcolor{black}{are} depicted in Fig.~\ref{fig-mmse}. In this figure, the $2$nd sample is not successfully delivered. Thus, when $t \in [D_2, D_3)$, the estimator $\hat{W}_t$ is still $W_{S_1}$, not $W_{S_2}$. In other words, $i=2$, but $\underline{i}=1$. This is one of the key differences from the previous studies with the reliable channel case, e.g., \cite{sun2020sampling,tang2022sampling,ornee2019sampling,tsai2021unifying}.  

\subsection{Sampling Times and Problem Formulation}\label{subsection-sampling}

We will control the sequence of causal sampling times $S_{i}$'s to minimize the estimation error. We will consider two types of sampling time: (i) the sampling time depends on the Wiener process (signal-aware sampling) and (ii) the sampling time is independent of the Wiener process (signal-agnostic sampling).
\subsubsection{Signal-aware Sampling}
When the sampling time $S_{i}$ depends on the Wiener process,
$S_i$ is a \emph{stopping time}, i.e., $S_{i}$ satisfies:
\begin{align} 
  \{S_{i}<t  \}\in \mathcal{F}(t)^+,  \ \mathcal{F}(t)^+ \triangleq \cap_{s>t}  \sigma(W_r, r\in [0,s]). \label{eq-stopping2}
\end{align} 
Here, $\sigma(A_1,\ldots,A_n)$ is the $\sigma-$field generated by the random variables $A_1,\ldots,A_n$, and $\mathcal{F}(t)^+$ is a filtration, i.e., a non-decreasing and right-continuous family of $\sigma-$field available to the sampler at time $t$. Intuitively, the sampling time $S_i$ not only depends on the history information prior to $D_{i-1}$, but also depends on the evolution of the Wiener process starting from $D_{i-1}$.

Then, we define the sampling policies.
The policy space $\Pi_{\text{signal-aware}}$ is defined as the collection of causal policies $\pi = S_{1},S_2,\ldots$ such that: (i) $S_{i}$ satisfies the condition \eqref{eq-stopping2}, and $S_{i} \ge D_{i-1}$; (ii) For each $i$, the waiting time $S_i - D_{i-1}$ is bounded by a stopping time that is independent of the history information before $S_{i-1}$\footnote{\textcolor{black}{It is the upper bound stopping time that is independent of the history information before $S_{i-1}$, not all of the stopping times. For example, the upper bound stopping time is to stop where the estimation error $W_t - \hat{W}_t$ exceeds a sufficiently large value. Setting this stopping time as an upper bound is reasonable, because we want to minimize the estimation error. Then, this stopping time is independent of the history information before $S_{i-1}$.}}. In addition, this bounded stopping time $\tilde{\tau}$ satisfies $\mathbb{E} \left[ W^4_{\tilde{\tau}} \right]<\infty$\footnote{If the condition (ii) does not hold, then the term $\lim_{T\rightarrow \infty} \mathbb{E} [R_{N(T)}]/T$ may not be $0$, where $N(T)$ is the largest number $n$ such that $S_{n}<T$, and $R_n = \int_{D_{n-1}}^{D_{n}} (W_t - \hat{W}_t)^2 dt $. If $\lim_{T\rightarrow \infty} \mathbb{E} [R_{N(T)}]/T \ne 0$, $R_n$ will diverge to infinity, which is not our concern.}. 

\subsubsection{Signal-agnostic Sampling}

When the sampling time is independent of the Wiener process, we then define the collection of policies $\Pi_{\text{signal-agnostic}}$ as the collection of policies $\pi = S_{1},S_2,\ldots$ such that:  (i) $S_{i}$ satisfies $S_{i} \ge D_{i-1}$; (ii) For each $i$, $S_i - D_{i-1}$ is bounded by a finite $2$nd moment random variable that is independent of the history information before $S_{i-1}$. 

Note that for any finite $2$nd moment random variable $A$, we have 
$\mathbb{E} \left[ W^4_{A} \right]=3\mathbb{E} \left[ A^2 \right] <\infty$. Therefore, $\Pi_{\text{signal-agnostic}} \subset \Pi_{\text{signal-aware}}$.

\subsubsection{Problem Formulation}

Our objective in this paper is to optimize the long-term average mean square estimation error (MSE) for both signal-aware and signal-agnostic cases: 
\begin{align}
\text{mse}_{\text{opt}} = & \inf_{\pi \in \Pi} \limsup_{T\rightarrow \infty}  \frac{1}{T} \mathbb{E} \left[  \int_{0}^{T} (W_t - \hat{W}_t)^2 dt \right].  \label{avg-aware}  
\end{align}
We aim to find a sampling policy $\pi$ from the set $\Pi$ of all causal policies, in order to minimize the MSE. The value $\text{mse}_{\text{opt}}$ is also called the optimal objective value. Problem~\eqref{avg-aware} is typically hard to solve due to the following reasons. (i) Problem~\eqref{avg-aware} is an infinite horizon undiscounted semi-Markov decision problem with an uncountable state space. (ii) For the case of signal-aware sampling, each action (sampling time) is a stopping time.

\section{Main Results}\label{main}

\subsection{Optimal Signal-aware Sampling Policy}\label{section-aware}


We first break down the time-horizon problem~\eqref{avg-aware} into a series of optimal sampling subproblems. Each of these subproblems determines the optimal sampling times between $D_{\overline{j}}$ and $D_{\overline{j+1}}$, where$D_{\overline{j}}$ represents the time of the $j$th successful delivery.

\begin{lemma}\label{lemma-pre1}
Solving the problem~\eqref{avg-aware} is the same as solving a series of \textcolor{black}{equivalent} optimal sampling subproblems, where the $j$th subproblem is given by
\begin{align}
J(w,\beta) \triangleq \inf_{\pi \in \Pi }  \mathbb{E}\left[ \int_{D_{\overline{j}} }^{D_{\overline{j+1}}} (W_t - \hat{W}_t)^2 dt 
 - \beta ( D_{\overline{j+1}} - D_{\overline{j}} ) \Big| W_{D_{\overline{j}}} - \hat{W}_{D_{\overline{j}}} =w \right], \label{eq-jwbeta}
\end{align}
where $\beta=\text{mse}_{\text{opt}}$.
\end{lemma}
Lemma~\ref{lemma-pre1} is a restatement of Lemma~\ref{lemma-singlepoch} in Section~\ref{section-proof}.
We note that, by choosing $\beta=\text{mse}_{\text{opt}}$, the sequence of linearized optimal stopping subproblems \eqref{eq-jwbeta} have the same solution as the original problem~\eqref{avg-aware}. \textcolor{black}{Note that these subproblems are independent and thus \emph{equivalent}. In other words, we only need to solve one subproblem~\eqref{eq-jwbeta} regardless of $j$, and $J(w,\beta)$ remains the same for any given $j$. This is because each $W_{D_{\overline{j}}} - \hat{W}_{D_{\overline{j}}}$ is independent of any history information before $D_{\overline{j}}$.}
Moreover, Lemma~\ref{lemma-pre1} improves similar results in e.g.,  \cite{sun2020sampling,ornee2019sampling,tsai2021unifying,pan2022optimal}, by removing the assumption that the $S_i$'s is a regenerative process.
Overall, to solve \eqref{avg-aware}, we can firstly solve \eqref{eq-jwbeta} with any given parameter $\beta>0$.

However, problem \eqref{eq-jwbeta} is still hard to solve. 
Let $M_j$ be the total number of transmission attempts between $D_{\overline{j}}$ and $D_{\overline{j+1}}$. Then, $\overline{j+1} = \overline{j}+M_j$. Problem \eqref{eq-jwbeta} needs to determine a sequence of sampling times $S_{\overline{j}+1},S_{\overline{j}+2},\ldots,S_{\overline{j}+M_j}$ until a successful packet delivery occurs at time $S_{\overline{j+1}}$. Hence, problem \eqref{eq-jwbeta} is a \emph{repeated optimal stopping problem with continuous-time control and a continuous state space}. This is the key technical challenge of our study. To the extend of our knowledge, this type of problems has not been addressed before. One limiting case of problem \eqref{eq-jwbeta} was studied in \cite[Eq. 47]{sun2020sampling}, where there exists no transmission errors and hence $M_j=1$.


We develop a value iteration algorithm that can find the optimal stopping times for solving problem~\eqref{eq-jwbeta}. To that end, we define a sequence of optimal stopping problems:
\begin{align}
\nonumber  J_n(w,\beta)  \triangleq   \inf_{\pi \in \Pi } & \  \mathbb{E} \bigg[ \int_{D_{\overline{j+1}-\min(M_j,n)} }^{D_{\overline{j+1}}} (W_t - \hat{W}_t)^2 dt 
 - \beta ( D_{\overline{j+1}} - D_{\overline{j+1}-\min(M_j,n)} ) \\ & \ \Big| W_{D_{\overline{j+1}-\min(M_j,n)}} - \hat{W}_{D_{\overline{j+1}-\min(M_j,n)}} =w \bigg], n=1,2,\ldots. \label{eq-jwbetan}
\end{align}
Hence, $J_n(w,\beta)$ determines the optimal solution for \textcolor{black}{at most} the last $n$ transmission attempts in problem~\ref{eq-jwbeta}. 
The principle of backward induction implies that $J_n$ satisfies the \emph{value iteration algorithm}: 
\begin{align}
\nonumber J_0(w,\beta) & \triangleq 0,\\
J_{n+1}(w,\beta) &  \triangleq   \inf_{\tau} g(w; \tau)+\alpha \mathbb{E} \left[  J_n(w+W_{\tau+Y},\beta) \right], \ n=0,1,2,\ldots, \label{eq-valueiterationn-main}
\end{align}
where $w+W_{\tau+Y}$ is the estimation error after a stopping time $\tau$ and the transmission time $Y$. And
the per-stage cost function $g(w;\tau)$ is defined as the square estimation error minus $\beta$ from the last delivery time to the next delivery time with a stopping time $\tau$:
\begin{align}
g(w; \tau) &= \mathbb{E} \left[   \int_{0}^{\tau+Y} (w+W_t)^2 dt -\beta (\tau+Y) \right]. \label{eq-costgtau-main}
\end{align}

The following theorem provides an exact solution to \eqref{eq-valueiterationn-main}, which is the key contribution in this paper: 
\begin{theorem}\label{thm1-v}
The sequence of optimal stopping times $\tau_n$'s to problem~\eqref{eq-valueiterationn-main} is given as follows:
\begin{align}
  \tau_n=  \inf_{t}  \left\{  t\ge 0:  
 | w+W_t |    \ge v_n(\beta)  \right\}, \label{eq-optimalpi}
\end{align} 
$v_n(\beta)$ is the unique positive root of the free boundary \textcolor{black}{differential equation}: 
\begin{align}
\frac{\partial}{\partial w} J_n(w,\beta) \Big|_{w = v_n(\beta)^+} = \frac{\partial}{\partial w} J_n(w,\beta) \Big|_{w = v_n(\beta)^-},\label{eq-fbm+-}
\end{align}
 $J_n(w,\beta)$ is updated as: 
\begin{align}
\nonumber & J_0(w,\beta)=0,\\
& J_{n}(w,\beta) = g(w,v_n(\beta),\beta) + \alpha \mathbb{E}_{W_Y} \left[ J_{n-1} (\max\{ |w|,v_{n}(\beta) \}+W_Y,\beta) \right], \ n=1,2,\ldots, \label{eq-rootjn2}
\end{align}
the function $g(w,v,\beta)$ is equal to
\begin{align}
g(w,v,\beta) =  \frac{1}{2} \mathbb{E} \left[ Y^2 \right]+\mathbb{E} \left[ Y \right]w^2-\mathbb{E} \left[ Y \right]\beta + \frac{1}{6}\max(v^4 - w^4,0) - (\beta - \mathbb{E}[Y])\max(v^2 - w^2,0). \label{eq-gwvbeta}
\end{align}

Moreover, the sequence $\{v_n(\beta)\}_n$ is decreasing and thus convergent.
\end{theorem}

The proof of Theorem~\ref{thm1-v} is provided in Section~\ref{section-eq-solve}.
Theorem~\ref{thm1-v} implies that each optimal stopping time $\tau_n$ is a hitting time that will stop when the estimation error exceeds a threshold $v_n(\beta)$. The threshold $v_n(\beta)$ is chosen by the free boundary method \cite{peskir2006optimal}, where the optimal value function $J_n(w,\beta)$ should be continuously differentiable on $w\in \mathbb{R}$. Since $v_n(\beta)$ is decreasing and convergent, $\tau_n$ is also convergent.

\begin{figure}[t]
\includegraphics[scale=.35]{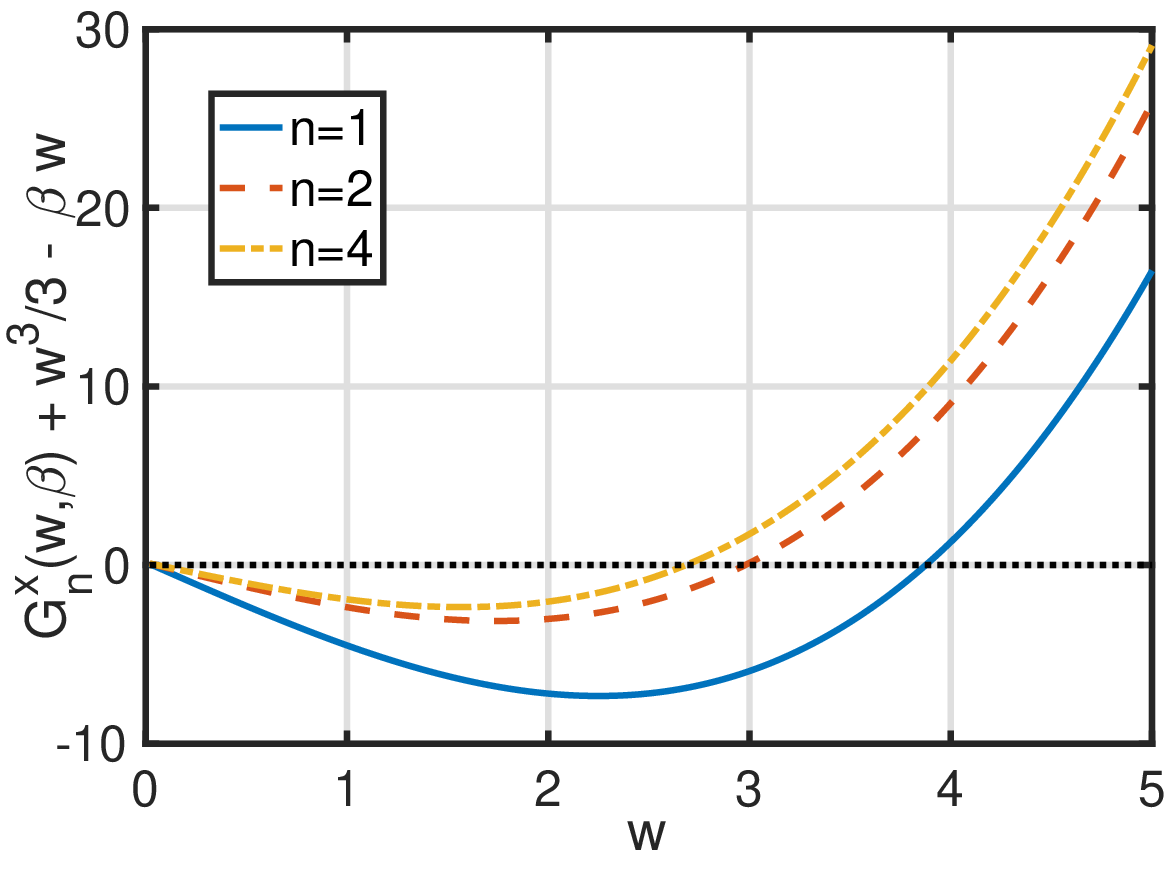}
\centering
\captionsetup{justification=justified}
\caption{\textcolor{black}{The evolution of the root function: $G^x_{n}(w,\beta) + \frac{1}{3} w^3 - \beta w$ over $w$, with $n=1,2,4$. In this example, we set $\beta=11.0,\alpha=0.3$, and a constant transmission delay $Y=6$. It is easy to see that $v_n(\beta)$, which is the positive root of $G^x_{n}(w,\beta) + \frac{1}{3} w^3 - \beta w$, is decreasing in $n$. }}
 \label{fig-gn}
\end{figure}

In addition, the optimal threshold $v_n(\beta)$ can be solved efficiently. In Theorem~\ref{thm-valueiteration} of Section~\ref{section-eq-solve}, we showed that the root of the free boundary method in \eqref{eq-fbm+-} is equivalent to: 
\begin{align}
 G^x_{n}(w,\beta) + \frac{1}{3} w^3 - \beta w =0. 
\label{eq-rootgn}
\end{align} 
Interestingly, $G^x_{n}(w,\beta)= \frac{1}{2}\frac{\partial}{\partial w} J_n(w,\beta) |_{w = v_n(\beta)^+}$, and $ -\frac{1}{3} w^3 + \beta w=\frac{1}{2}\frac{\partial}{\partial w} J_n(w,\beta) |_{w = v_n(\beta)^-}$.

$ G^x_{0}(w,\beta)=0$, and the function $ G^x_{n}(w,\beta)$ is updated as
\begin{align}
\nonumber & G^x_{n}(w,\beta) =  \mathbb{E} \left[ Y \right] w \\ & + \alpha \mathbb{E}_{W_Y} \left[ G_{n-1}^x(w+W_Y,\beta)  {\text{\large $\mathds{1}$}}_{|w+W_Y| \ge v_{n-1}(\beta) } + \left( \beta (w+W_Y) - \frac{1}{3} (w+W_Y)^3  \right)  {\text{\large $\mathds{1}$}}_{|w+W_Y| < v_{n-1}(\beta) } \right]. \label{eq-root-qn}
\end{align}
Because \eqref{eq-root-qn} contains only an expectation over $W_Y$ without derivatives, computing $G^x_{n}(w,\beta)$ is easy. We also showed that \textcolor{black}{$G^x_{n}(w,\beta) + \frac{1}{3} w^3 - \beta w$ is strongly convex for $w >0$}. Thus, we only need logarithm time complexity to solve $v_n(\beta)$ for each $n$ in \eqref{eq-rootgn}, such as bisection search or Newton's method. \textcolor{black}{Fig.~\ref{fig-gn} illustrates some intuitive properties of $v_n(\beta)$ and its root function, $G^x_{n}(w,\beta) + \frac{1}{3} w^3 - \beta w$.}

Further, $J_n$ converges linearly to $J$.
To illustrate, we first define a norm. Let us pick any value $\rho$ with $\alpha < \rho <1$, and denote a weight function $u(w) = \max(\bar{b},w^2)$, where $\bar{b}$ can take any positive value such that 
$
\mathbb{E} \left[ 1+\frac{2|W_Y|}{\sqrt{\bar{b}}}    +  \frac{W_Y^2}{\bar{b}} \right] \le \frac{\rho}{\alpha}.
$
The weight function $u(w)$ is not related to $\beta$. The sup-norm $\| \cdot \|$ of a function $f(w)$ is defined as 
$
\| u \| = \sup_{w \in \mathbb{R}} \left | \frac{f(w)}{u(w)} \right |. 
$
We have the following result:
\begin{lemma}\label{lemma-contraction-main}
 $\|J_n(\cdot,\beta) - J(\cdot,\beta) \| \le \rho \| J_{n-1}(\cdot,\beta) - J(\cdot,\beta)\|$.
\end{lemma}
Lemma~\ref{lemma-contraction-main} is restated in Lemma~\ref{lemma-contraction} at Section~\ref{section-linearconv}.
Since $\tau_n$ is also convergent, each of the optimal stopping (waiting) times in \eqref{eq-jwbeta} should also be a hitting time with the threshold $v(\beta) = \lim_{n\rightarrow \infty} v_n(\beta)$. We finally conclude the following result:
\begin{theorem}\label{thm1}
An optimal sampling solution $S_{i}$'s to the series of problem~\eqref{eq-jwbeta} is:
\begin{align}
  S_{i+1} =  \inf_{t}  \left\{  t\ge D_{i}:  
 | W_t - \hat{W}_t |    \ge v(\beta)  \right\}, \ i=0,1,2,\ldots, \label{eq-optimalpi}
\end{align}
where $v(\beta)$ is the limit of the sequence $v_n(\beta)$'s, \textcolor{black}{and $v_n(\beta)$ can be computed by solving \eqref{eq-fbm+-}, or more efficiently, by solving \eqref{eq-rootgn} and \eqref{eq-root-qn}}.
\end{theorem}
The proof of Theorem~\ref{thm1} is provided in Section~\ref{section-linearconv}.

Theorem~\ref{thm1} illustrates an important property of an optimal sampling policy for a given parameter $\beta$.
Note that $|W_t - \hat{W}_t|$ is the estimation error at the current time $t$.
Theorem~\ref{thm1} implies that the optimal sampling policy given in \eqref{eq-optimalpi} has a simple structure. The optimal policy is a \emph{threshold type}: the sampler may wait until the instantaneous estimation error $|W_t - \hat{W}_t|$ exceeds the threshold $v(\beta)$. Specifically, if the estimation error at the initial time $D_{i}$ exceeds $v(\beta)$, then it is optimal to immediately transmit the sample. The optimal threshold $v(\beta)$ is independent of the evolution of the Wiener process.  

After solving \eqref{eq-jwbeta} with a given $\beta$, we will finally determine the optimal objective value $\beta=\text{mse}_{\text{opt}} $. Note that in \eqref{eq-jwbeta}, $W_{D_{\overline{j}}} - \hat{W}_{D_{\overline{j}}}$ has the same distribution as $W_Y$, where $Y$ has the same distribution as the i.i.d. transmission delay $Y_i$'s. Then, we have the following result: 
\begin{theorem}\label{thm1-root}
$\beta = \text{mse}_{\text{opt}} $ is the root of 
\begin{align}
\mathbb{E} \left[J(W_Y,\beta)\right]=0, \label{eq-betaroot2}
\end{align}
where $\text{mse}_{\text{opt}}$ is the optimal objective value of \eqref{avg-aware}.
\end{theorem}
Theorem~\ref{thm1-root} is shown in Lemma~\ref{dinklebach-lemma} \textcolor{black}{at} Section~\ref{section-proof}.
Combining Theorem~\ref{thm1} and Theorem~\ref{thm1-root}, we finally provide the optimal solution to \eqref{avg-aware}.

Moreover, we showed that we can also use a low complexity algorithm, such as bisection search, to compute the root of $\beta$.
So in conclusion, we can efficiently solve $v( \text{mse}_{\text{opt}})$ and $\text{mse}_{\text{opt}}$ with low complexity, which is provided in Algorithm~\ref{algr1}: 
\begin{itemize}
\item
Line~\ref{alg-linebeta1}---\ref{alg-linebeta2} in Algorithm~\ref{algr1} is an \emph{inner layer} update to efficiently compute the optimal threshold $v(\beta)$ and the function $J(w,\beta)$ for a given $\beta$ \textcolor{black}{(corresponding to Theorem~\ref{thm1-v} and Theorem~\ref{thm1}). In Line~\ref{alg-linebetam}, due to Lemma~\ref{lemma-contraction-main}, we only need a logarithm number of iterations. In Line~\ref{alg-vnbeta}, since the root function in \eqref{eq-rootgn} is strongly convex, we only need a simple Newton's method to obtain $v_n(\beta)$.} 
\item
Line~\ref{alg-repeat},\ref{alg-repeat1},\ref{alg-linebeta3} serves as an \emph{outer} layer that uses a simple bisection method to determine the root of $\beta$ \textcolor{black}{(corresponding to Theorem~\ref{thm1-root})}.
\end{itemize}



 \begin{algorithm}[t]
\caption{Bisection method for solving the optimal threshold $v(\text{mse}_{\text{opt}})$ and $\text{mse}_{\text{opt}}$ }\label{algr1}
\textbf{Given} $k_1$ small, $k_2$ large, $k_1<k_2$, and tolerance $\epsilon_1,\epsilon_2$ small. 

\textbf{repeat}\label{alg-repeat} 

\qquad $\beta=\frac{1}{2}(k_1+k_2)$, \label{alg-repeat1}

\qquad Set $J_0(w,\beta) = G^x_0(w,\beta) =0$ \label{alg-linebeta1}

\qquad Set iteration number $m= \lceil - \log_{\rho} \frac{ \| J_1(\cdot,\beta) \|}{\epsilon_1} \rceil $   \label{alg-linebetam}

\qquad \textbf{for $n=1:m$} 

\qquad \qquad Update $ G^x_{n}(w,\beta)$ in \eqref{eq-root-qn} 

\qquad \qquad Solve $v_{n}(\beta)$ in \eqref{eq-rootgn} \label{alg-vnbeta}


\qquad \qquad Update $ J_{n}(w,\beta)$ in \eqref{eq-rootjn2}

\qquad \textbf{end for} \label{alg-linebeta2}

\qquad \textbf{if} $ \mathbb{E}^{w} \left[  J_{m}(W_Y,\beta)  \right] <0$: $k_2=\beta$. \textbf{else} $k_1=\beta$ \label{alg-linebeta3}

\textbf{until} $k_2-k_1<\epsilon_2$

\textbf{return} $v_{m}(\beta), \beta$

\end{algorithm}

In the special case where $\alpha=0$, it is easy to observe that $v_n(\beta) = v_1(\beta)$, and $J_n(w,\beta)=J_1(w,\beta)=g(w,v_1,\beta)$ for all $n=1,2,\ldots$. As a result, the optimal threshold $v(\beta)=v_1(\beta)$, and the optimal value function $J(w,\beta)=g(w,v_1,\beta)$. By \eqref{eq-rootgn}, $v_1(\beta)=\sqrt{3(\beta - \mathbb{E}[Y])}$. Therefore, Theorem~\ref{thm1} and~\ref{thm1-root} reduces to the following corollary:
\begin{corollary}\label{cor-1}
Suppose that $\alpha=0$, then an optimal solution $S_{i}$'s to problem~\eqref{avg-aware} satisfies:
\begin{align}
  S_{i+1} =  \inf_{t}  \left\{  t\ge D_{i}:  
 | W_t - \hat{W}_t |    \ge \sqrt{3(\beta - \mathbb{E}[Y])}  \right\},
\end{align}
where $\beta$ is the root of 
\begin{align}
\mathbb{E} \left[  g(W_Y, \sqrt{3(\beta - \mathbb{E}[Y])} , \beta )  \right]=0.
\end{align}
Moreover, $\beta = \text{mse}_{\text{opt}}$ is the optimal objective value of \eqref{avg-aware}.
\end{corollary} 
The optimal policy provided in Corollary~\ref{cor-1} is the same as that of \cite[Theorem~1]{sun2020sampling}. In addition, we have also improved \cite[Theorem~1]{sun2020sampling} by removing the assumption of the regenerative process.
The optimal sampling policy provided in Corollary~\ref{cor-1} is a threshold type on the instantaneous estimation error, and the optimal threshold is given in closed-form. 

There are several variations of Corollary~\ref{cor-1} with a reliable channel case $\alpha=0$.
In \cite{ornee2019sampling}, the paper changes the source process to be the Ornstein-Uhlenbeck process and shows that the optimal threshold is a root of the closed-form equation. The model where the source can reset the Wiener process is described in \cite{tsai2021unifying}. 
Theorem~\ref{thm1} and Theorem~\ref{thm1-root} are different from these studies by generalizing to an i.i.d. unreliable channel scenario ($\alpha \ge 0$). 
Note that the last transmission may be successful or failed for each sample. However, in Theorem~\ref{thm1} and Theorem~\ref{thm1-root}, each sampling time follows the same threshold type with the same threshold $v(\beta)$, regardless of whether the last transmission failed or not. 

The expression \eqref{eq-optimalpi} in Theorem~\ref{thm1} implies that our optimal policy relies on the value of the Wiener process at the sampling time of the successfully delivered sample, $S_{\underline{i}}$, but may not on $S_i$. 
This is also a key difference from the case of a reliable channel ($\alpha=0$), e.g., \cite{sun2020sampling,ornee2019sampling,tsai2021unifying} and Corollary~\ref{cor-1}.

\subsection{Optimal Signal-agnostic Sampling Policy with Sampling Rate Constraint}\label{sec-agnostic}

Finally, we turn to the signal-agnostic case and provide the exact solution to Problem~\eqref{avg-aware}.  
Using \cite{sun2020sampling}, for any signal-agnostic policy, we have 
\begin{align}
\mathbb{E} \left[ (W_t - \hat{W}_t )^2 \right] = \Delta_t = t - S_{\underline{i}}, \ t\in [D_{i},D_{i+1}).\label{eq-age-def}
\end{align}
In other words, when the sampling time does not depend on the Wiener process, the expected square estimation error MMSE is equal to the \emph{age of information}. So our MSE-optimal sampling problem (Problem~\eqref{avg-aware}) is equivalent to the age-optimal sampling problem. Problem~\eqref{avg-aware} is equivalent to 
\begin{align}
 \text{age}_{\text{opt}}  = & \inf_{\pi \in \Pi} \limsup_{T\rightarrow \infty}  \frac{1}{T} \mathbb{E} \left[  \int_{0}^{T} \Delta_t dt \right]. \label{eq-avg-agnostic}
\end{align}
Age of information $\Delta_t$, or simply the age, is a metric for evaluating the data freshness. As is mentioned in \eqref{eq-age-def}, the age $\Delta_t$ is defined as the time elapsed since the freshest delivered sample is generated \cite{sun2019age}. If a fresh sample is successfully delivered to the estimator, the age decreases to the system time of the sample. Otherwise, the age increases linearly in time. A sample path of the age $\Delta_t$ is depicted in Fig~\ref{fig-age}.

\begin{figure}[t]
\includegraphics[scale=.35]{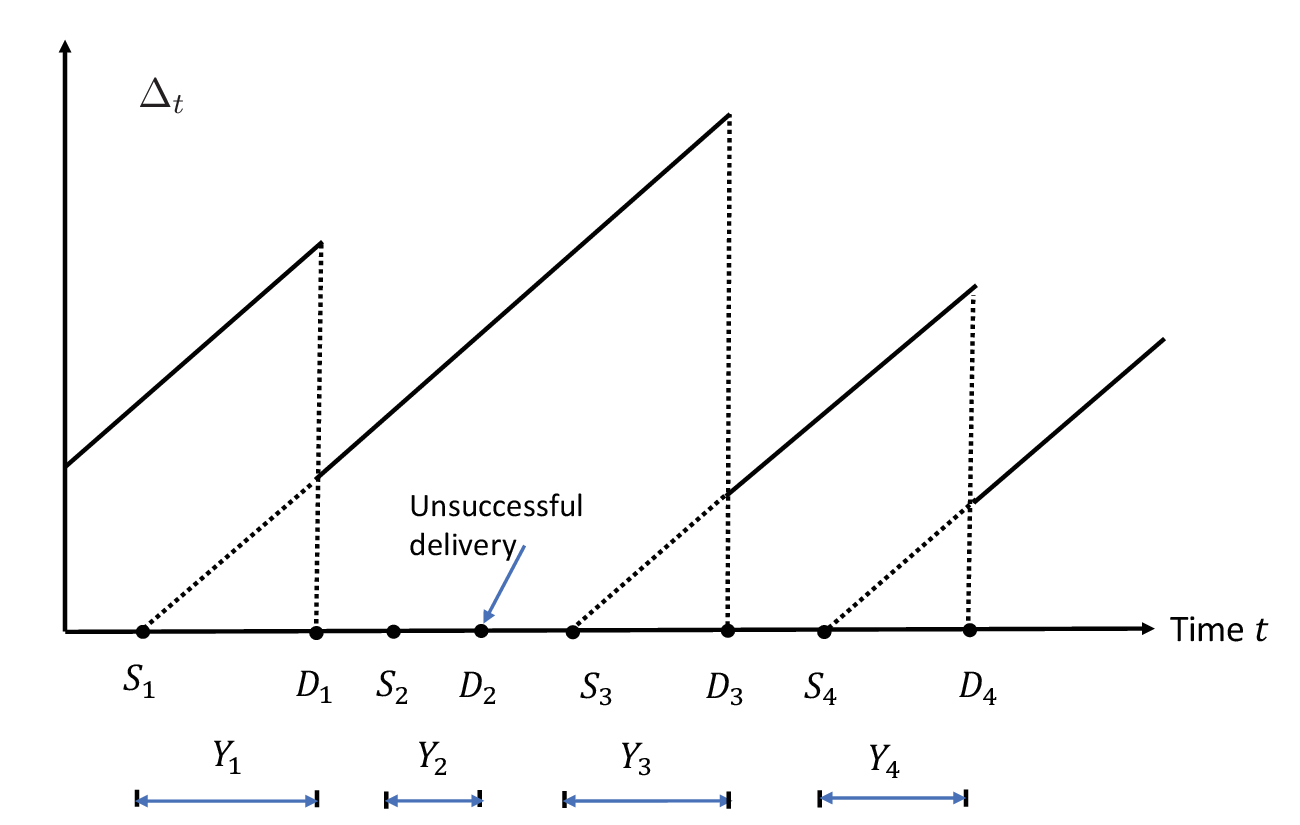}
\centering
\captionsetup{justification=justified}
\caption{Evolution of the age $\Delta_t$ over time $t$. }
 \label{fig-age}
\end{figure}

We then have the following result: 
\begin{theorem}\label{thm2} 
An optimal solution $S_i$'s to the problem~\eqref{eq-avg-agnostic} is provided as:
\begin{align}
 & S_{i+1} =  \inf_{t}  \left\{  t \ge D_i:  
 \Delta_t    \ge \beta - \frac{\mathbb{E} \left[ Y \right]}{1-\alpha}  \right\}.  
\label{thm2-beta2}
\end{align}
 $\beta$ is the root of 
\begin{align}
\mathbb{E} \left[ \int^{\max(\beta -  \frac{\mathbb{E} \left[ Y \right]}{1-\alpha} , Y) + Y'}_{Y} t dt \right] - \beta \mathbb{E} \left[\max(\beta  - Y,  \frac{\mathbb{E} \left[ Y \right]}{1-\alpha}) \right]=0, \label{eq-ageroot}
\end{align} 
where $Y' = \sum_{k=1}^{M}Y_{j,k}$, $Y$ and $Y_{j,1},Y_{j,2},\ldots$ are i.i.d. and have the same distribution as the transmission delay $Y_i$'s. 
\end{theorem}
Theorem~\ref{thm2} provides the same sampling policy as that of \cite[Theorem~1]{pan2022optimal}. But we slightly improve \cite[Theorem~1]{pan2022optimal} by removing its assumption of the regenerative process. The proof of this improvement is provided in Appendix~\ref{appprop-thm2}.

Different from Theorem~\ref{thm1}, the optimal sampling policy is a threshold policy on the age, or equivalently, the MMSE, instead of the instantaneous estimation error. Note that the age keeps increasing over time if there is no successful delivery. 
As a result, if the previous transmission failed, the age is always larger than the optimal threshold $  \beta - \frac{\mathbb{E} \left[ Y \right]}{1-\alpha}$.
Therefore, 
Theorem~\ref{thm2} tells that if the previous transmission is successful, the sampler may wait for some time until the current age exceeds the threshold $  \beta - \frac{\mathbb{E} \left[ Y \right]}{1-\alpha}$. If the previous transmission failed, the sampler chooses zero-wait. This is another key difference from the optimal signal-aware sampling policy in Theorem~\ref{thm1}. In Theorem~\ref{thm1}, due to the randomness of the Wiener process, each sampler may need to wait, regardless of the outcome of the previous transmission. In addition, since there is only one waiting time between two successful deliveries, the optimal objective value $\beta$ is the root of the closed form expression \eqref{eq-ageroot}. But the root function of $\beta$ for the signal-aware case in \eqref{eq-betaroot2} is not closed-form. Instead, as is illustrated in Theorem~\ref{thm1} and Algorithm~\ref{algr1}, we need to construct a sequence of functions $J_n$'s to approach the root function.  

\begin{figure}[t]
\includegraphics[scale=.30]{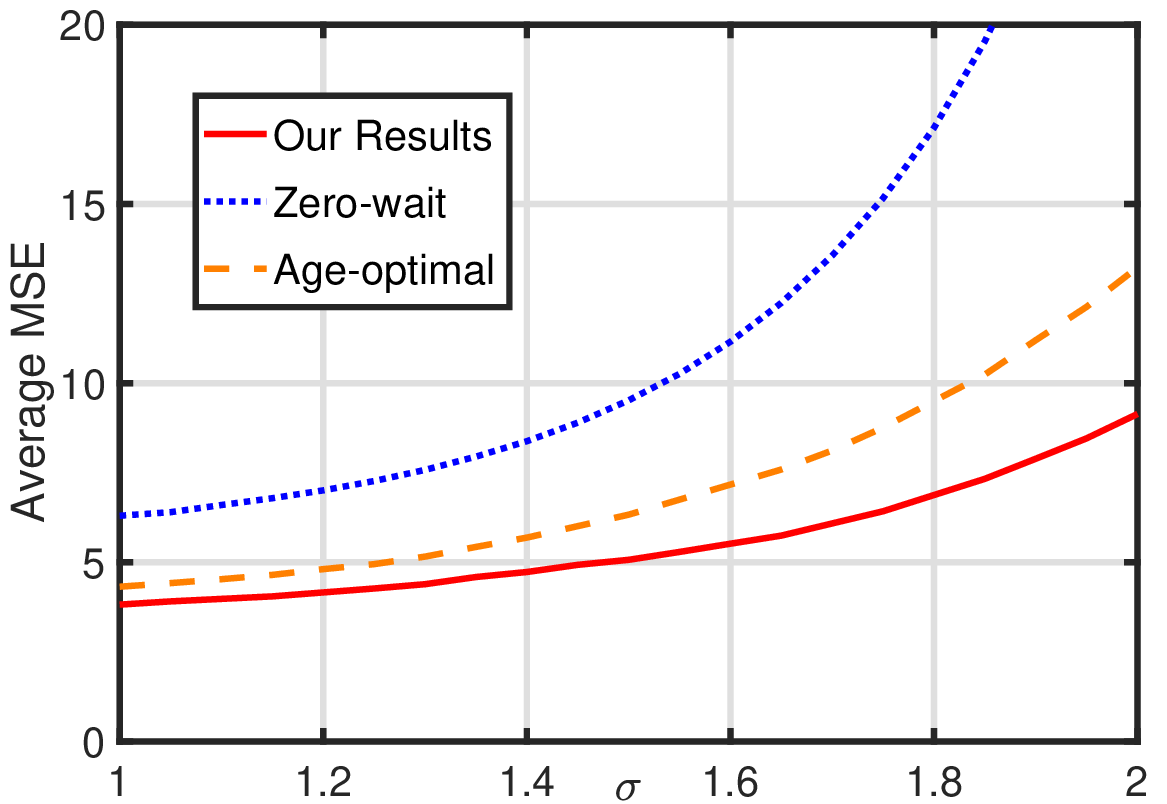}
\centering
\captionsetup{justification=justified}
\caption{Average MSE versus $\sigma$, where the channel delay is lognormal distributed with the parameter $\sigma$. As $\sigma$ increases, the channel delay distribution is more heavy-tailed. The probability of i.i.d. transmission failure $\alpha=0.65$.}
 \label{fig-sigma}
\end{figure}
\begin{figure}[t]
\includegraphics[scale=.30]{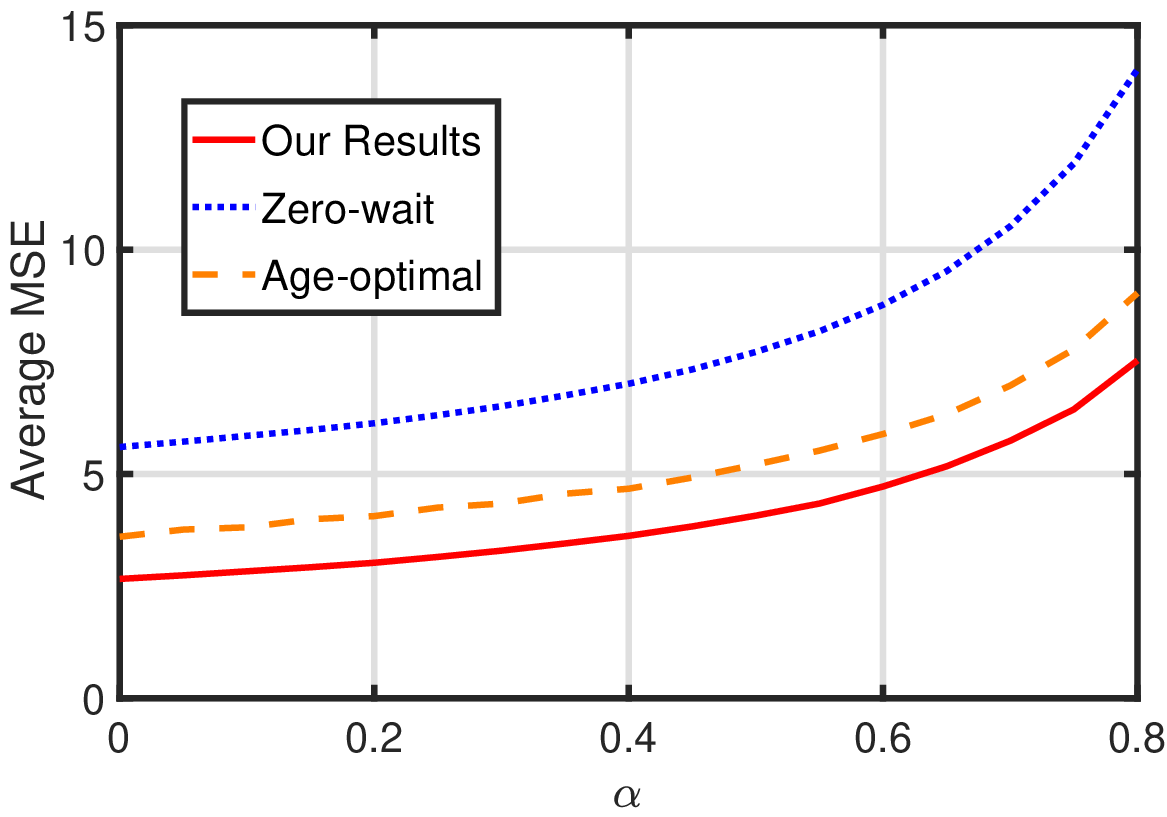}
\centering
\captionsetup{justification=justified}
\caption{Average MSE versus the probability of i.i.d. transmission failure $\alpha$, where the channel delay is lognormal distributed with the parameter $\sigma = 1.5$.}
 \label{fig-alpha-logn}
\end{figure}


\section{Simulation}\label{numerical}

In this section, we will compute the long term average MMSE (average MSE) of the following three sampling policies: 

$1$. Our Results: our optimal sampling policy, which is the solution to problem~\eqref{avg-aware}, provided in Theorem~\ref{thm1-v}---\ref{thm1-root}. The average MSE is then computed in Algorithm~\ref{algr1}. It waits until the estimation error exceeds a threshold.

$2$. Zero-wait: The source transmits a sample once it receives the feedback, i.e., $S_{i+1}=D_i$. This simple policy can achieve the maximum throughput and the minimum delay. However, even in the case of a reliable channel, it may not optimize the age of information \cite{yates2015lazy} or optimize the estimation error \cite{sun2020sampling}. In our study with an unreliable channel, Theorem~\ref{thm2} implies that the zero-wait policy does not optimize the age. Moreover, Theorem~\ref{thm1-v}---\ref{thm1-root} imply that the zero-wait policy does not optimize the estimation error.       

$3$. Age-optimal: This policy is provided in Theorem~\ref{thm2}, restated in \cite[Theorem~1]{pan2022optimal}, and the average MSE is computed by \cite[Algorithm~1]{pan2022optimal}. Age-optimal policy achieves the optimal average age. It waits until the age (i.e., MMSE $\mathbb{E}[(W_t - \hat{W}_t)^2]$) exceeds a threshold.  

We will follow the same network system as is illustrated in Section~\ref{estimation} and Fig.~\ref{model}.
We consider two scenarios about the delay distribution of the unreliable channel: heavy-tailed distribution (e.g., lognormal distribution) and short-tailed distribution (e.g., constant). 


\begin{figure}[t]
\includegraphics[scale=.32]{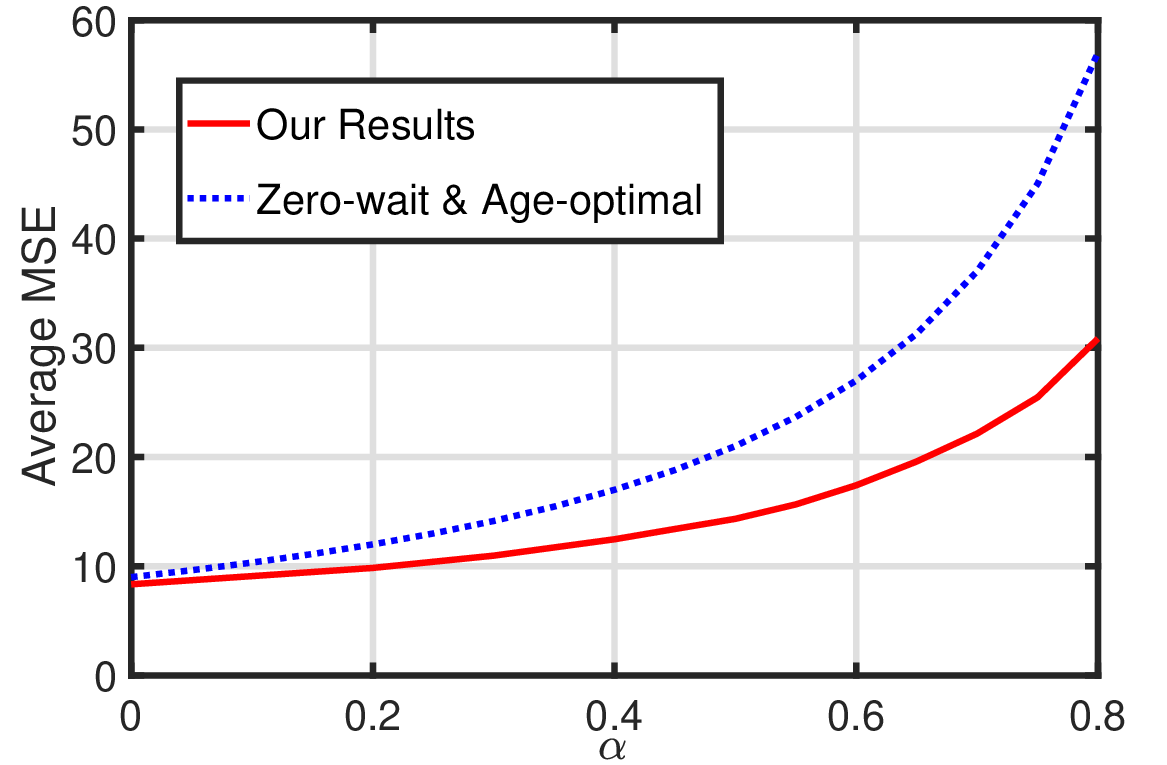}
\centering
\captionsetup{justification=justified}
\caption{Average MSE versus $\alpha$, where the channel delay is a constant with the delay $Y=6$.}
 \label{fig-alpha-const}
\end{figure}

In the first scenario, we assume that the channel delay follows a lognormal distribution. The lognormal random variable with scale parameter $\sigma$ is expressed as $e^{\sigma A}/ \mathbb{E} [e^{\sigma A}]$, where $A$ is the standard normal random variable. Fig.~\ref{fig-sigma} illustrates the relationship between the average MSE of the four sampling policies with parameter $\sigma$ of lognormal channel delay, given a discount factor $\alpha$ (probability of failure of the channel). The numerical results validate that our proposed policy always achieves the lowest average MSE. Note that as $\sigma$ increases, the lognormal distribution of the channel becomes more heavy-tailed. We observe that the zero-wait policy is far from optimality, and the age-optimal policy also grows much quicker than our optimal policy. Therefore, our optimal policy substantially outperforms the age-optimal and zero-wait policies when the channel delay becomes heavy-tailed.    
Fig.~\ref{fig-alpha-logn} plots the evolutions of the average MSE with the change of $\alpha$ given that the parameter $\sigma = 1.5$. From our observation, the zero-wait policy is always far from our optimal policy. 

In the second scenario, we assume that the channel delay is a constant. Fig.~\ref{fig-alpha-const} depicts the evolution of the average MSE of different policies with the change of $\alpha$. Note that the age-optimal policy is equivalent to the zero-wait policy when the delay is a constant, as is shown in \cite[Corollary~3]{pan2022optimal}. We observe that when the channel connectivity is more reliable ($\alpha$ very small), then the zero-wait policy is only slightly inferior to the optimal policy. However, as $\alpha$ increases, the zero-wait policy becomes far from optimality. \textcolor{black}{The intuitive reason is as follows:} since the Wiener process oscillates, with a nontrivial probability, our optimal policy waits at each sample, \textcolor{black}{no matter} whether the last transmission failed or not. Compared to the zero-wait policy, such a quite different sampling strategy leads to much improvement for the average MSE. This is the newly observed phenomenon that has not been found in the previous studies, e.g., \cite{sun2020sampling,sun2019sampling,pan2022optimal,ornee2019sampling}.

In summary, our optimal policy can perform much better than the zero-wait and the age-optimal policy when either (i) the transmission time is heavy-tailed, or (ii) the transmission time is light-tailed, and the channel is highly unreliable.


\section{Proof of Main Results}\label{section-proof}


In this section, we provide the proof for \textcolor{black}{efficiently} solving the optimal signal-aware policy for \eqref{avg-aware}. In Section~\ref{proof-preliminary}, we first show that there exists an optimal policy such that the inter-sampling time of the successfully delivered packet is i.i.d. Thus, the long term average MMSE in \eqref{avg-aware} is equal to the average MMSE only between the two successful delivery times. In Section~\ref{section-refor}, after linearizing, the reduced problem is equivalent to optimizing a discrete time discounted problem with multiple stopping times \eqref{eq-singlepoch-form}. This new problem a strict generalization to a discrete time discounted MDP, \textcolor{black}{where} each action is extended to be a stopping time. To solve \eqref{eq-singlepoch-form}, in Section~\ref{section-eq-solve}, we first speculate that the optimal policy and its optimal value function satisfy the Bellman equation. Then, we use a value iteration algorithm to approach the optimal value function, where each iteration is an optimal stopping problem. Interestingly, we analytically solve the optimal stopping time for each iteration, which is a key technical contribution in this paper. Finally, in Section~\ref{section-linearconv}, we use the contraction mapping property to show that the optimal value function of the value iteration algorithm convergences linearly to that of the Bellman equation. Thus, we exactly solve \eqref{eq-singlepoch-form}. This ends the proof.

\subsection{Reducing to a Single-epoch Problem}\label{proof-preliminary}

\subsubsection{Replacing the subscript $i$ by $(j,k)$}
The proof relies on the number of successfully delivered samples and the number of samples attempted for a successful delivery. These messages cannot be easily described in $\{ S_i,Y_i,D_i\}$'s by using only one subscript $i$. Therefore, for notational simplicity, throughout Section~\ref{section-proof}, we will replace $S_i,Y_i,D_i$ by $S_{j,k},Y_{j,k},D_{j,k}$, respectively. Here, 
we denote the $j$th \emph{epoch} to be the time \textcolor{black}{interval} between the $(j-1)$th and the $j$th successful deliveries. Let $M_j$ represent the total number of transmissions attempted during the $j$th epoch. Then, $M_j$ has a geometric distribution with parameter $1-\alpha$. Note that if the channel is reliable, then $M_j=1$. 
In addition, $k \in \{1,2,\ldots, M_j \}$ represents the index of transmission for the $j$th epoch, where the case $k=1$ implies that the last transmission \textcolor{black}{was} successful.  
Note that the mapping from $i$ to $(j,k)$ is one-to-one. For example, in Fig~\ref{fig-mmse}, $S_1 = S_{1,1}$ with $M_1=1$, $S_2 = S_{2,1},S_{3} = S_{2,2}$ with $M_2 = 2$, and $S_4 = S_{3,1}$ with $M_3=1$. 

By \eqref{eq-msei}, the MMSE estimator $\hat{W}_t$ is expressed as 
\begin{align} 
\hat{W}_t  = W_{S_{j-1,M_{j-1}}}, t\in [D_{j-1,M_{j-1}},D_{j,M_{j}}). 
\end{align}


\subsubsection{Reducing to a Single-epoch Problem}

We aim to show that solving the original problem \eqref{avg-aware} can be reduced to solving the optimal sampling times $S_{j,1},S_{j,2},\ldots$ within an epoch $j$ over a subset of the policy space $\Pi_{\text{signal-aware}}$. 
We denote such the subset $\Pi_j$ as a collection of sampling times $S_{j,1},S_{j,2},\ldots$ within epoch $j$ such that each inter-sampling time $\{S_{j,k}-S_{j-1,M_{j-1}}, k=1,2,\ldots\}$ is independent of the history information before $S_{j-1,M_{j-1}}$. 
The following result shows that our average cost problem \eqref{avg-aware} reduces to a single epoch problem (with arbitrary index $j$) that contains possibly multiple samples from one successful delivery time until the next successful delivery time.

%
\begin{proposition}\label{prop-aware}
There exists an optimal policy for the problem \eqref{avg-aware} such that $\{ S_{j,M_{j}} - S_{j-1,M_{j-1}} \}_j$ are i.i.d. Moreover, problem~\eqref{avg-aware} is equivalent to
\begin{align}
\text{mse}_{\text{opt}} = & \inf_{(S_{j,1},S_{j,2},\ldots) \in \Pi_j} \frac{ \mathbb{E} \left[  \int_{D_{j-1,M_{j-1}}}^{D_{j,M_j}} (W_t - W_{S_{j-1,M_{j-1}}})^2 dt \right]}{ \mathbb{E} \left[   D_{j,M_j} - D_{j-1,M_{j-1}} \right] }. \label{eq-singlepoch-mse}
\end{align} 
\end{proposition}
 \begin{proof}
See Appendix~\ref{appprop-aware}.
\end{proof}
Proposition~\ref{prop-aware} implies that to solve the \textcolor{black}{long term} average MMSE problem \eqref{avg-aware}, we can solve a problem with only a single epoch. Each sampling decision in this epoch is independent of the history information prior to the final sampling time of the previous epoch. Proposition~\ref{prop-aware} is motivated by \cite{sun2019sampling,sun2020sampling} under a reliable channel. In these studies, the original problem is reduced to an average MMSE problem between two delivery times (a single sample problem). One of the key reasons is that at each delivery time, the estimation error is updated and is independent of the history information before the last sampling time. But in our unreliable case, at a failed delivery time, the estimation error is not updated and is still correlated to that history information. Thus, our single epoch problem cannot be further reduced to a single sample problem. In addition, we also improve \cite{sun2019sampling,sun2020sampling} by removing the assumption of the regenerative process. A similar result to Proposition~\ref{prop-aware} is presented in \cite{arafa2021timely} with an unreliable channel and signal-agnostic sampling, without the assumption of the regenerative process. We also generalize \cite{arafa2021timely} since our sampling time depends on the Wiener process.    

Although we have reformulated the long term average MMSE problem~\eqref{avg-aware} into an average MMSE problem within a single epoch \eqref{eq-singlepoch-mse}, problem \eqref{eq-singlepoch-mse} is still hard to solve. This is because it contains a fraction and thus is a repeated semi-MDP. 


\subsection{Reformulating as a Multiple Stopping Times Problem: an Extension to a Discounted MDP}\label{section-refor}

In this section, we will linearize problem~\eqref{eq-singlepoch-mse} and reformulate it as a discounted cost and repeated Markov decision process (MDP), where each action is a stopping time.


Let us define a minimization problem with a parameter $\beta \in \mathbb{R}$: 
\begin{align}
h(\beta) = \inf_{\pi\in \Pi_j} \mathbb{E} \left[   \int_{D_{j-1,M_{j-1}}}^{D_{j,M_j}} (W_t - W_{S_{j-1,M_{j-1}}})^2 dt - \beta (D_{j,M_j} - D_{j-1,M_{j-1}})  \right]. \label{eq-ha}
\end{align} Here, $\pi = (S_{j,1},S_{j,2},\ldots)$.
By Dinkelbach's method \cite{dinkelbach1967nonlinear}, we have 
\begin{lemma}\label{dinklebach-lemma}

(i) $h(\beta) \lesseqqgtr 0$ if and only if $\text{mse}_{\text{opt}} \lesseqqgtr \beta $.

(ii) When $\beta = \text{mse}_{\text{opt}}$, the solution to \eqref{eq-singlepoch-mse} and \eqref{eq-ha} are equivalent.
\end{lemma}
Therefore, to solve \eqref{eq-singlepoch-mse}, we will solve $h(\text{mse}_{\text{opt}})=0$. 

We denote $Z_{j,k}$ as the waiting time for the $k$th sample in epoch $j$. Then,
\begin{align}
D_{j,M_j} - D_{j-1,M_{j-1}} & = \sum_{k=1}^{M_j} Z_{j,k}+Y_{j,k}. \label{eq-length}
\end{align}
Then, combined with \eqref{eq-length} and the strong Markov property of the Wiener process, given that $W_{D_{j-1,M_{j-1}}}-W_{S_{j-1,M_{j-1}}}=w$, $w\in \mathbb{R}$, we have 
\begin{align}
 \int_{D_{j-1,M_{j-1}}}^{D_{j,M_j}} (W_t - W_{S_{j-1,M_{j-1}}})^2 dt & = \int_{0}^{ \sum_{k=1}^{M_j} Z_{j,k}+Y_{j,k}} (W_t+w)^2 dt. \label{eq-int-length}
\end{align}
As a result, \eqref{eq-length} and \eqref{eq-int-length} give: 
\begin{lemma}\label{lemma-singlepoch}
An optimal solution to \eqref{eq-singlepoch-mse} given that $W_{D_{j-1,M_{j-1}}}-W_{S_{j-1,M_{j-1}}}=w$, $w\in \mathbb{R}$ satisfies 
\begin{align}
 J(w) & \triangleq \inf_{\pi\in \Pi_j} J_\pi(w),  \label{eq-singlepoch-form}
 \\
 J_\pi(w) & \triangleq \mathbb{E} \left[    \int_{0}^{ \sum_{k=1}^{M_j} Z_{j,k}+Y_{j,k}} (W_t+w)^2 dt - \text{mse}_{\text{opt}} ( \sum_{k=1}^{M_j} Z_{j,k}+Y_{j,k})  \right]. \label{eq-singlepoch}
\end{align} 
\end{lemma}
Here, $J(w)$ is the total cost of the optimal policy, which is also called the \emph{optimal value function}. And $J_\pi(w)$ is the total cost of a policy, which is also called the \emph{action value function} with a policy $\pi$. 

For any policy $\pi$, the action value function $J_\pi(w)$ in \eqref{eq-singlepoch} is further written as
\begin{align}
J_\pi(w) & =   \mathbb{E} \left[  \sum_{k=1}^{M_j} g( \tilde{W}_{k}; Z_{j,k}) |  \tilde{W}_{1} = w \right], \label{eq-totalsingle2}
 \end{align}
where the state values $\tilde{W}_k, \ k=1,2\ldots$ satisfy 
\begin{align}
 &  \tilde{W}_{k+1}  =  \tilde{W}_{k}+W_{Z_{j,k}+Y_{j,k}}, k=1,2,\ldots, \label{eq-evolution2}
\end{align}
 $g(w; \tau)$, also called a \emph{per stage cost function}, is the expected integration of square estimation error minus $ \text{mse}_{\text{opt}}$ from the last delivery time to the next delivery time,\footnote{For comparison, $J_\pi(w)$ is the expected integration of square estimation error minus $ \text{mse}_{\text{opt}}$ from the last delivery time to the next \emph{succsssful} delivery time.} where the initial estimation error is $w$, and the sampler's waiting time is $\tau$. $g(w; \tau)$ is defined as:
\begin{align}
g(w; \tau) &= \mathbb{E} \left[   \int_{0}^{\tau+Y} (w+W_t)^2 dt - \text{mse}_{\text{opt}} (\tau+Y) \right], \label{eq-costgtau}
\end{align}
where $Y$ has the same distribution as the channel delay. The equation \eqref{eq-totalsingle2} holds because of the strong Markov property of the Wiener process. 

Note that $J_\pi(w)$ represents the expected cost of \textcolor{black}{square estimation error} minus a constant $\text{mse}_{\text{opt}}$ within an epoch. In an epoch, if the transmission is successful with probability $1-\alpha$, then the system will stop. Thus, the system state will enter a ``stopping'' set with $0$ cost; If the transmission fails with probability $\alpha$, the system state will enter the next transmission with a per-stage cost $g$. Therefore, 
\begin{align}
J_\pi(w) = \sum_{k=1}^{\infty} \alpha^{k-1} \mathbb{E} \left[  g(\tilde{W}_{k}; Z_{j,k}) | \tilde{W}_{1} = w \right], \label{eq-totalsingle} 
\end{align} which is proven in \cite[Appendix~F]{pan2022optimal}. The $k$th stage state $\tilde{W}_{k}$ implies that all the previous $k-1$ transmissions failed, and the coefficient $\alpha^{k-1}$ is the probability of $k-1$ consecutive failures.

Equations \eqref{eq-singlepoch-form}---\eqref{eq-totalsingle} imply that
problem \eqref{eq-singlepoch-form} belongs to a discounted cost problem with multiple stopping times, or in other words, a \emph{repeated MDP}, because there are multiple waiting times $Z_{j,1},Z_{j,2},\ldots$, and each waiting time is a stopping time. Suppose that each waiting time is not a stopping time, i.e., the waiting time policy \textcolor{black}{chooses} a real value that is independent of the Wiener process. Then, problem \eqref{eq-singlepoch-form} is reduced to a discrete time discounted cost MDP \cite{bertsekas1995dynamic1}. This is because: (i) the state at each stage $k$ is the estimation error at the $k-1$th delivery time, $\tilde{W}_{k}$ (when $k=1$, $\tilde{W}_{1}=w$ \eqref{eq-totalsingle}). (ii) The action at each stage $k$ is the waiting time for the $k$th sample, $Z_{j,k}$. (iii) The state transition is provided in \eqref{eq-evolution2}. (iv) The cost function is defined in \eqref{eq-costgtau}.

%
Note that the waiting times $Z_{j,1},Z_{j,2},\ldots$ are correlated. Thus, despite that we have linearized the problem~\eqref{eq-singlepoch-mse} into a multiple stopping time problem \eqref{eq-singlepoch-form}, problem~\eqref{eq-singlepoch-form} still faces the curse of dimensionality.

\subsection{Analytical Solution to the Value Iteration \eqref{eq-valueiterationn} for the Multiple Stopping Times Problem \eqref{eq-singlepoch-form}}\label{section-eq-solve}

In the special case where each waiting time $Z_{j,1},Z_{j,2},\ldots$ is not a stopping time, the optimal policy and the optimal value function to the discounted MDP satisfies the Bellman equation \cite[Chapter~9]{bertsekas2004stochastic}. The advantage of the Bellman equation is that it turns the MDP with correlated waiting times into an optimization problem over a single waiting time and thus helps reduce the complexity of the MDP. Suppose that we can propose a waiting time decision $z(w),w\in \mathbb{R}$ and the action value function of the stationary policy $z,z,\ldots$ that is the unique solution to the Bellman equation. Then, the policy $z,z,\ldots$ is an optimal policy. 

Similar to the previous MDP case, we believe that the optimal policy and the optimal value function of our repeated MDP \eqref{eq-singlepoch-form} still satisfies the Bellman equation\footnote{This statement is technically true if we can show that our action space is a Borel space (We call $B$ as a Borel space if there exists a complete separable metric space $R$ and a Borel subset $\tilde{B}\in \mathcal{B}_R$ such that $B$ is homeomorphic to $\tilde{B}$) \cite[Chapter~9]{bertsekas2004stochastic}. Examples of a Borel space are $\mathbb{R}$ and any real-valued intervals. For showing that our action space is a Borel space, we leave to our future studies.    }. Because except that each waiting time is extended to be a stopping time, our repeated MDP \eqref{eq-singlepoch-form} has the same components as that of a discounted MDP. The Bellman equation for our repeated MDP \eqref{eq-singlepoch-form} is defined as follows:
\begin{align}
J(w) =  T J(w) \triangleq \inf_{\tau \in  \mathfrak{M}} g(w; \tau)+\alpha \mathbb{E} \left[  J(w+W_{\tau+Y}) \right], \label{eq-bellman1}
\end{align} where $\mathfrak{M}$ is the set of stopping times on the Wiener process $W_t$ such that
 \begin{align}
 \mathfrak{M} = \left\{ \tau:  \{ \tau <t  \}\in \mathcal{F}(t)^+,   \mathbb{E} \left[ \tau^2 \right]<\infty   \right\},
 \end{align}
 where $ \mathcal{F}(t)^+ = \cap_{s>t}  \sigma(W_r, r\in [0,s])$. In \eqref{eq-bellman1}, $w+W_{\tau+Y}$ is the next state of estimation error, after a stopping time $\tau$ and a channel delay $Y$.
 
However, problem \eqref{eq-bellman1} is not an optimal stopping problem because the function $J$ exists in both sides. 
To overcome this issue and exactly solve \eqref{eq-bellman1}, our method in this paper is to use the value iteration algorithm \cite{bertsekas1995dynamic2} to convert \eqref{eq-bellman1} into multiple standard optimal stopping problems that are solvable. Specifically, we will construct a sequence of optimal stopping problems to approach the problem \eqref{eq-bellman1}, where in each optimal stopping problem, the action value functions are well-defined. 


We define the value iteration algorithm regarding to the problem \eqref{eq-bellman1} as follows:      
\begin{align}
\nonumber J_0(w) & \triangleq 0,\\
J_{n+1}(w) &  \triangleq T J_n(w) =  \inf_{\tau \in  \mathfrak{M}} g(w; \tau)+\alpha \mathbb{E} \left[  J_n(w+W_{\tau+Y}) \right], \ n=0,1,2,\ldots \label{eq-valueiterationn}
\end{align}
We also denote $\tau_1,\tau_2,\ldots$ as the optimal stopping time of the problem \eqref{eq-valueiterationn} when $n=1,2,\ldots,$ respectively. 
 \textcolor{black}{Then, $J_n(w)=T^n 0(w)$ is the discounted integrated cost from the first delivery time (the last transmission was successful) until at most the $n$th delivery time, where the $n$th transmission} implies that previous $n-1$ transmissions have failed. In addition, the waiting times for the $n$ transmissions are $\tau_1,\tau_2,\ldots, \tau_n$, respectively. Note that $J(w)$ is the discounted cost about infinite number of transmissions. Thus, our objective is to exactly solve \eqref{eq-valueiterationn} by figuring out $\tau_1,\tau_2,\ldots$ and show that $T^n 0(w) \rightarrow J(w)$ as $n\rightarrow \infty$.   

\subsubsection{Candidate Solutions to \eqref{eq-valueiterationn}}\label{section-candidate}


We speculate that each optimal stopping time $\tau_1,\tau_2,\ldots$ for \eqref{eq-valueiterationn} is a \emph{hitting time}, or in other words, \emph{threshold type}, defined as follows:
\begin{align}
\tau_n  = \inf_{t\ge 0}\{t:  |w+W_t|\ge v_n  \}, \ n=1,2,\ldots, \label{eq-taunstar}
\end{align}
where $w$, called the initial state, is the estimation error at the $n-1$th delivery time $D_{j,n-1}$ ($n=1$ implies that the last transmission was successful, and the delivery time is $D_{j-1,M_{j-1}}$). Next, we \textcolor{black}{aim to} find out the sequence of the optimal thresholds $v_1,v_2,\ldots$.  

Let us define a function $G_{n}(w)$ as follows:
\begin{align} 
G_{n}(w) = g(w;0) + \alpha \mathbb{E} \left[ J_{n-1} (w+W_Y) \right].\label{eq-root-qn2}
\end{align}
Intuitively, $G_n(w)$ is the action value function that chooses $0$ waiting time at the first stage, \textcolor{black}{incurs} the cost $g(w;0)$, and chooses the optimal waiting times at the remaining $n-1$ stages. Since the speculated optimal waiting time \eqref{eq-taunstar} is a hitting time, $J_n(w) = G_n(w)$ if $|w| \ge v_n$.  
In addition, we provide an alternative expression of $g(w;\tau)$:
\begin{lemma}\label{cost1stage}
\begin{align}
 g(w; \tau) = \mathbb{E} \left[ \int_{0}^{\tau}  (w+W_t)^2 - \text{mse}_{\text{opt}}  dt +\mathbb{E} \left[ Y \right] (w+W_\tau)^2  \right] +\frac{1}{2} \mathbb{E} \left[ Y \right]^2 - \mathbb{E} \left[ Y \right] \text{mse}_{\text{opt}}. \label{eq-cost1divide0}
 \end{align}
Moreover, if $\tau$ is a hitting time with a threshold $v$ given the initial value $w$. i.e., $\tau  = \inf_{t\ge 0}\{t:  |w+W_t|\ge v  \}$, then we have 
\begin{align}
 & g(w;\tau) =  g(w,v,\text{mse}_{\text{opt}}),
\end{align}
where $g(w,v, \text{mse}_{\text{opt}})$ is defined in \eqref{eq-gwvbeta}.
\end{lemma}
\begin{proof}
See Appendix~\ref{cost1stageapp}
\end{proof}

Then, our problem \eqref{eq-valueiterationn} is augmented as the sequence of standard optimal stopping problem \cite[Chapter 1]{peskir2006optimal}:
\begin{align}
\tilde{J}_{n}(w,q) = \inf_{\tau\in\mathfrak{M}} \mathbb{E} \left[  \tilde{G}_{n}(w+W_\tau, q+Q_\tau) \right], \ \text{for all } w,q\in \mathbb{R}, \label{eq-bellmansolve}
\end{align}
where 
\begin{align}
\tilde{G}_{n} (w+W_t,q+Q_t) & \triangleq \tilde{g}(w+W_t,q+Q_t) + \alpha \mathbb{E} \left[ J_{n-1}(w+W_t+W_{Y}) \right],\\
\tilde{g}(w+W_t,q+Q_t) & \triangleq  q+Q_t +\mathbb{E} \left[ Y \right] (w+W_t)^2   + \frac{1}{2} \mathbb{E} \left[ Y \right] ^2 - \mathbb{E} \left[ Y \right] \text{mse}_{\text{opt}},\\
 Q_t & \triangleq   \int_{0}^{t} (w+W_r)^2 - \text{mse}_{\text{opt}}  dr  .
\end{align}
By Lemma~\ref{cost1stage}, for any $\tau$, we have $g(w; \tau) = \mathbb{E} \left[  \tilde{g}(w+W_t,Q_t)  \right]  = \mathbb{E} \left[  \tilde{g}(w+W_t,q+Q_t)  \right] -q$. 

According to \cite[Chapter~10]{oksendal2013stochastic} and \cite[Section~8]{peskir2006optimal}, the free boundary method implies that the optimal objective function $\tilde{J}_{n}(w,q)$ should satisfy 
\begin{align}
& \frac{1}{2} \frac{\partial^2}{\partial w^2} \tilde{J}_{n}(w,q) + w^2 - \text{mse}_{\text{opt}}  = 0,  w\in (-v_{n},v_{n}), \label{eq-fbm1} \\
& \tilde{J}_{n}(w,q)  = \tilde{G}_{n}(w,q),  w\in (-\infty,-v_{n}]\cup [v_{n}, \infty), \label{eq-fbm2} \\
&  \frac{\partial}{\partial w}  \tilde{J}_{n}(w,q) \Big{|}_{w=\pm v_{n}}  =  \frac{\partial}{\partial w}   \tilde{G}_{n}(w,q)  \Big{|}_{w=\pm v_{n}}. \label{eq-fbm3}
\end{align}
The first equation \eqref{eq-fbm1} tells that in the continuation set $(-v_{n},v_{n})$, the infinitesimal operator of $ \tilde{J}_{n}(w,q)$ is zero. In the second equation \eqref{eq-fbm2}, at the stopping set $(-\infty,-v_{n}]\cup [v_{n}, \infty)$, the stopping time $\tau_{n}$ is zero. The third equation \eqref{eq-fbm3} implies that $\tilde{J}_{n}(w,q)$ should be continuously differentiable at the boundary points $w = \pm v_{n}$.
These three equations are then simplified to:
\begin{align}
& \frac{1}{2}  J_{n}''(w) + w^2 - \text{mse}_{\text{opt}}  = 0,  w\in (-v_{n},v_{n}),  \label{eq-fbm1simple}  \\
& J_{n}(w)  = G_{n}(w),  w\in (-\infty,-v_{n}]\cup [v_{n}, \infty),  \label{eq-fbm2simple}  \\
&    J_{n}'(w) \Big{|}_{w=\pm v_{n}}  =  G_{n}'(w)  \Big{|}_{w=\pm v_{n}}.  \label{eq-fbm3simple} 
\end{align}

By \eqref{eq-fbm1simple}---\eqref{eq-fbm3simple}, $v_n$ is the positive solution to $J'_n(v_n^-)=G_n'(v_n)$. Combined with Lemma~\ref{cost1stage}, we provide the following results for deriving the sequence $v_1,v_2,\ldots$:
\begin{lemma}\label{derivative1}
For all $n=1,2,\ldots$ we have that:

(a)  
If $|w|<v_{n}$, then
\begin{align}
J'_{n}(w)  = \frac{\partial}{\partial w} g(w,v_n, \text{mse}_{\text{opt}}) = -\frac{2}{3}w^3+2 \text{mse}_{\text{opt}} w. \label{eq-stage1wsmallv}
\end{align}
If $|w| > v_{n}$, then
\begin{align}
 J'_{n}(w) & =  G'_{n}(w) = \frac{\partial}{\partial w} g(w,0, \text{mse}_{\text{opt}}) + \alpha  \mathbb{E} \left[ J'_n(w+W_Y) \right] = 2 \mathbb{E} \left[ Y \right] w + \alpha  \mathbb{E} \left[ J'_{n-1}(w+W_Y) \right]. \label{eq-stage1vsmallw} 
\end{align}
The optimal threshold $v_n$ is the positive solution to 
\begin{align}
 G'_{n}(w) + \frac{2}{3}w^3 - 2 \text{mse}_{\text{opt}} w=0. \label{eq-gn}
\end{align}
Moreover, $G_{n}''(w)$ and $G_{n}'''(w)$ are continuous.

(b) $G'_n(0)=0,$ and $  G''_{n}(w)  +2w^2 - 2 \text{mse}_{\text{opt}} \ge 0$ for all $w\in [v_{n},\infty)$.

(c) $G'''_{n}(w) \ge 0$, and $  G'''_{n}(w) +4w \ge 0$ for all $w \ge 0$.


(d) The sequence of thresholds $v_1,v_2,\ldots$ is bounded with $v_{n}\le \sqrt{3 \text{mse}_{\text{opt}}}$ and is decreasing, thus converges. 
\end{lemma}
\begin{proof}
See Appendix~\ref{appderivative1}.
\end{proof}




\subsubsection{Optimality of the Candidate Solution to \eqref{eq-valueiterationn}}

We \textcolor{black}{finally validate} that the hitting time \eqref{eq-taunstar} is the optimal solution. Combined with Lemma~\ref{derivative1}, we have the following result: 

\begin{theorem}\label{thm-valueiteration}
(a) An optimal sequence of waiting times $\tau_1,\tau_2,\ldots$ for \eqref{eq-valueiterationn} satisfies \eqref{eq-taunstar}, and each threshold $v_n$ is the positive root of \eqref{eq-gn}, where $G'_0(w)=0$, $G'_n(w)$ is updated by \eqref{eq-stage1vsmallw}, $J'_0(w)=0$, and $J'_n(w)$ is updated by \eqref{eq-stage1wsmallv}\eqref{eq-stage1vsmallw}. 

(b) The function $G'_{n}(w)  + \frac{2}{3}w^3-2 \text{mse}_{\text{opt}} w$ in \eqref{eq-gn} is convex for $w \ge 0$ and strongly convex for $w > 0$. Therefore, the positive root of $v_{n}$ is unique. In addition, $v_n$ decreases and thus converges.
\end{theorem}
Theorem~\ref{thm-valueiteration} (b) is directly shown by Lemma~\ref{derivative1}.
It remains to show that the exact solution provided in Theorem~\ref{thm-valueiteration} (a) is optimal to the value iteration problem \eqref{eq-valueiterationn}. 
\begin{proof}[Proof of Theorem~\ref{thm-valueiteration} (a)]
we obtain the two following results:

\begin{lemma}\label{jsmallg}
We have $ \tilde{J}_{n}(w,q) \le \tilde{G}_{n}(w,q)$ for any $(w,q)\in \mathbb{R}^2$ and the iteration number $n=1,2,\ldots$.
\end{lemma}
\begin{proof}
See Appendix~\ref{appjsmallg}.
\end{proof}

\begin{definition}
A function f(w,q) is excessive if $\mathbb{E}[\tilde{f}(w+W_t,q+Q_t)]\le \tilde{f}(w,q)$ for all $t\ge0$ and $(w,q)\in \mathbb{R}^2$.
\end{definition}

\begin{lemma}\label{j-is-excessive}
The negative value function $-\tilde{J}_n(w,q)$ is excessive for any $(w,q)\in \mathbb{R}^2$ and the iteration number $n=1,2,\ldots$. 
\end{lemma}

\begin{proof}
See Appendix~\ref{appj-is-excessive}.
\end{proof}

By Lemma~\ref{jsmallg} and Lemma~\ref{j-is-excessive},
using \textcolor{black}{Corollary to Theorem 1 in \cite[Section 3.3.1]{shiryaev1978optimal}}, we have that the stopping times $\tau_{1},\tau_2,\ldots$ in \eqref{eq-taunstar} are optimal to \eqref{eq-bellmansolve}, thus are optimal to \eqref{eq-valueiterationn}. This completes the proof of Theorem~\ref{thm-valueiteration} (a). 
\end{proof}

\subsection{Linear Convergence of Value Iteration to the Repeated MDP \eqref{eq-singlepoch-form}}\label{section-linearconv}

In this section, we will show that the optimal value functions $J_n$ of the value iteration algorithm \eqref{eq-valueiterationn} converge linearly to the optimal value function $J$ of our problem in \eqref{eq-singlepoch-form}. 
We have the following result:

%


\begin{lemma}\label{lemma-contraction}
(a) Suppose that the continuation set of $\tau$ is is bounded by $\bar{b}$, i.e., if $w^2 \ge \bar{b}$, then $\tau=0$. Then, 
$
\frac{\mathbb{E} \left[ u(w+W_{\tau+Y}) \right]}{u(w)} \le \frac{\rho}{\alpha}. 
$

(b) The function $J_n(w) = T^n 0(w)$ satisfies the contraction mapping property, i.e., $\| T^n 0 \| <\infty$, $\| T^\infty 0 \| <\infty$, and $\| T^{n+1}0 - T^{n} 0 \| \le \rho^n \| T0 \|$.  The Bellman operator $T$ is defined in \eqref{eq-bellman1}. 

(c) $J(w) = T^\infty 0 (w)$ is the unique solution to the Bellman equation $J = T J $ \eqref{eq-bellman1} (with $\| J \| <\infty$).
Further, $\|T^n 0 - J \| \le \rho \| T^{n-1} 0 - J\|$.

\end{lemma}
\begin{proof}
See Appendix~\ref{app-lemma-contraction}.
\end{proof}

By Lemma~\ref{lemma-contraction}(c), $J$ is the unique solution to the Bellman equation $J=TJ $. Therefore, $J$ is the optimal value function for the problem~\eqref{eq-singlepoch-form}. Due to the linear convergence of $J_n$ to $J$, Lemma~\ref{derivative1}(d) implies that the optimal stopping time for \eqref{eq-singlepoch-form} is also a hitting time, where the optimal threshold is $v = \lim_{n\rightarrow \infty} v_n$. This completes the proof of Theorem~\ref{thm1}. 
In addition, Lemma~\ref{dinklebach-lemma} implies that $\text{mse}_{\text{opt}}$ is the solution to $\mathbb{E}[\lim_{n\rightarrow \infty} J(W_Y)]=0$.
These statements combined with Theorem~\ref{thm-valueiteration} completes the solution to the problem~\eqref{avg-aware}.

\subsection{Discussion}
In this section, we compare our proof and technical contributions with some related works and discuss some interesting future directions. 

\subsubsection{Special Case $1$: Reliable Channel \cite{sun2020sampling}}
In the special case of a reliable channel ($\alpha=0$), $M_j=1$. The problem \eqref{eq-singlepoch-form} is then reduced to:
\begin{align}
J(w) & \triangleq \inf_{\tau \in  \mathfrak{M} } g(w;\tau).\label{eq-sing}
\end{align}
The problem \eqref{eq-singlepoch-form} for general $\alpha\ge 0$ is a repeated MDP, because we need to determine multiple correlated waiting times in an epoch, and each waiting time is a stopping time. However, when $\alpha=0$, the problem \eqref{eq-sing} reduces to an MDP, or in other words, an optimal stopping problem with a single waiting time. Note that solving \eqref{eq-sing} is still nontrivial. We speculate that the optimal waiting time $\tau$ is a hitting time. Using Lemma~\ref{derivative1}~(a), the optimal threshold $v$ is the positive root of 
\begin{align}
 \frac{2}{3} v^3 -  2(\text{mse}_{\text{opt}} - \mathbb{E}[Y] )v = 0,
\end{align}   
which is $v = \sqrt{3(\text{mse}_{\text{opt}} - \mathbb{E}[Y])}$. By Theorem~\ref{thm-valueiteration}, the speculated waiting time is optimal. This implies the final result Corollary~\ref{cor-1} (\cite[Theorem~1]{sun2020sampling}).

Similar studies with a reliable channel are also indicated, e.g., in \cite{sun2019sampling,ornee2019sampling,tsai2021unifying}. The key insight is to solve an optimal stopping time like \eqref{eq-sing}. Our study with an unreliable channel is different from these studies, because we need to solve a problem with multiple correlated stopping times \eqref{eq-singlepoch-form}. To solve this, we need to analytically solve a value iteration algorithm \eqref{eq-valueiterationn} that includes a sequence of optimal stopping problems. Compared to \eqref{eq-sing}, for each iteration $n$, our optimal stopping problem is more challenging to solve, because the optimal value function is a more complicated expression that contains a summation of $n$ correlated samples.    


\subsubsection{Special Case $2$: Signal-agnostic Sampling \cite{pan2022optimal}}

When the sampling time is independent of the Wiener process, each waiting time takes a nonnegative value based on the timing history information, but not the evolution of the Wiener process. The previous problem \eqref{eq-singlepoch-form} is reduced from a discounted and repeated MDP to a discounted MDP. The study in \cite{pan2022optimal} has shown that the optimal policy is a threshold policy on the age (i.e., MMSE). 

Since the optimal signal-aware sampling policy is different from the optimal signal-aware sampling policy, the proof of solving our problem \eqref{eq-singlepoch-form} is different from that of the discounted MDP in \cite{pan2022optimal}. The authors in \cite{pan2022optimal} solve their problems as follows: (i) they first propose a threshold based waiting decision $\mu(\delta) = \max(\text{age}_{\text{opt}} -\delta - \mathbb{E}[Y]/(1-\alpha),0)$, where $\delta$ is the age state, and $\text{age}_{\text{opt}}$ is the optimal average age; (ii) then they show that $\mu$ and its value function are the unique solution to the Bellman equation: $J_{\text{agnostic}}(\delta) = \inf_{z \ge 0} g_{\text{agnostic}}(\delta,z)+\mathbb{E}[J_{\text{agnostic}}(\delta+z+Y)]$, where $g_{\text{agnostic}}(\delta,z) \triangleq \mathbb{E}[\int_{\delta}^{\delta+z+Y} (t - \text{age}_{\text{opt}})dt ]$. 

However, such the proof ideas cannot be applied to our case, due to the following challenges that do not appear in  \cite{pan2022optimal}: (i) Since each waiting time is a stopping time, solving \eqref{eq-singlepoch-form} faces the curse of dimensionality. For example, when $\alpha=0$, \eqref{eq-singlepoch-form} reduces to \eqref{eq-sing}, but \eqref{eq-sing} is still an optimal stopping problem. In the signal-agnostic case, \eqref{eq-sing} is reduced to a convex optimization problem \cite[Lemma~7]{sun2019sampling}, thus is much easier to solve; (ii) In \cite{pan2022optimal}, the Bellman equation is solvable. Since $\mu$ is threshold type on the age, it is optimal to wait $(\mu>0)$ only when the last transmission was successful. Thus, the optimal value function is a closed-form expression: 
$J_{\text{agnostic}}(\delta) = \mathbb{E} \Big[    \int_{\delta}^{\delta+\mu(\delta)+Y'} \delta dt - \text{age}_{\text{opt}} ( \mu(\delta)+ Y')  \Big]$, where $Y'$ is given in Theorem~\ref{thm2}. Since the Bellman equation is a minimization over nonnegative values, solving the Bellman equation is the same as comparing a few closed-form expressions. In our case, however, it is hard to compare, because the optimal value function $J(w)$ is not closed-form. This is due to the randomness of the Wiener process, and we may wait for each sample. 


\subsubsection{Future Direction $1$: Non i.i.d. Channel Failure}
When the channel failure is extended from i.i.d. to Markovian, we still believe that the statements in Section~\ref{proof-preliminary} and Section~\eqref{section-refor} are correct. 
However, there is a key difference in Section~\ref{section-eq-solve}: the problem \eqref{eq-singlepoch-form} \eqref{eq-totalsingle} is changed to be
\begin{align}
J(w) = \inf_{\pi = Z_{j,1},Z_{j,2},\ldots}    g(w; Z_{j,1})  + (1-\alpha') \mathbb{E} \left[  g(\tilde{W}_{2}; Z_{j,2}) | \tilde{W}_{1} = w \right] + \sum_{k=3}^{\infty} \alpha^{k-1} \mathbb{E} \left[  g(\tilde{W}_{k}; Z_{j,k}) | \tilde{W}_{1} = w \right], \label{eq-totalsinglemarkov} 
\end{align} where $\alpha$ is the self transition probability from $OFF$ state to $OFF$ state, and $\alpha'$ is the self transition probability from $ON$ state to $ON$ state. In the non i.i.d. case where $1-\alpha' \ne \alpha$, Problem \eqref{eq-totalsinglemarkov} has a changing discount factor. Thus, the Bellman equation and the value iteration algorithm are not well-defined, making this new problem challenging to solve.

\subsubsection{Future Direction $2$:  Non i.i.d. Transmission Delay}
Suppose that we consider a Markovian transmission delay. Then, the waiting time not only should depend on the evolution of the Wiener process, but also should depend on the last transmission delay. This is because the last transmission delay effects the next transmission delay. Therefore, the value iteration \eqref{eq-valueiterationn} should be extended as:       
\begin{align}
J_{n+1,\text{markov}}(w,y) =  \inf_{\tau \in  \mathfrak{M}} g(w; \tau)+\alpha \mathbb{E} \left[  J_{n,\text{markov}}(w+W_{\tau+Y},Y) \right], \ n=0,1,2,\ldots,
\end{align} where $y$ is the last transmission delay, and the distribution of $Y$ is affected by $y$. Due to the space limitation, we will consider this extended problem in the future directions. 


\vspace{-0.1cm}
\section{Conclusion}\label{conclusion}
In this paper, we provide a sampling policy to minimize the mean square estimation error, where the sampler generates the sample at the source and transmits it to the remote estimator over a time-varying channel. We show that the optimal sampling policy is a threshold policy on the instantaneous estimation error, and the threshold is computed efficiently. The curse of dimensionality that originates from the randomness of the Wiener process, channel conditions, and the channel delay is circumvented. We believe that the proof of our main results provides an insight about how to solve a problem with discounted and multiple stopping times.

\section*{Acknowledgment}


This work has been supported in part by NSF grants NSF AI Institute (AI-EDGE) CNS-2112471, CNS-2106933, CNS-2106932, CNS-2312836, CNS-1955535, CNS-1901057, and CNS-2239677, by Army Research Office under Grant W911NF-21-1-0244, and was sponsored by the Army Research Laboratory and was accomplished under Cooperative Agreement Number W911NF-23-2-0225. The views and conclusions contained in this document are those of the authors and should not be interpreted as representing the official policies, either expressed or implied, of the Army Research Laboratory or the U.S. Government. The U.S. Government is authorized to reproduce and distribute reprints for Government purposes notwithstanding any copyright notation herein.

We thank Tasmeen Zaman Ornee, Md Kamran Chowdhury Shisher, and Yining Li for their valuable suggestions for this paper.



\bibliographystyle{ACM-Reference-Format}
\bibliography{wienersample}


\begin{thebibliography}{35}


\ifx \showCODEN    \undefined \def \showCODEN     #1{\unskip}     \fi
\ifx \showDOI      \undefined \def \showDOI       #1{#1}\fi
\ifx \showISBNx    \undefined \def \showISBNx     #1{\unskip}     \fi
\ifx \showISBNxiii \undefined \def \showISBNxiii  #1{\unskip}     \fi
\ifx \showISSN     \undefined \def \showISSN      #1{\unskip}     \fi
\ifx \showLCCN     \undefined \def \showLCCN      #1{\unskip}     \fi
\ifx \shownote     \undefined \def \shownote      #1{#1}          \fi
\ifx \showarticletitle \undefined \def \showarticletitle #1{#1}   \fi
\ifx \showURL      \undefined \def \showURL       {\relax}        \fi
\providecommand\bibfield[2]{#2}
\providecommand\bibinfo[2]{#2}
\providecommand\natexlab[1]{#1}
\providecommand\showeprint[2][]{arXiv:#2}

\bibitem[Arafa et~al\mbox{.}(2020)]%
        {arafa2020age}
\bibfield{author}{\bibinfo{person}{Ahmed Arafa}, \bibinfo{person}{Jing Yang},
  \bibinfo{person}{Sennur Ulukus}, {and} \bibinfo{person}{H~Vincent Poor}.}
  \bibinfo{year}{2020}\natexlab{}.
\newblock \showarticletitle{Age-minimal transmission for energy harvesting
  sensors with finite batteries: Online policies}.
\newblock \bibinfo{journal}{\emph{IEEE Transactions on Information Theory}}
  \bibinfo{volume}{66}, \bibinfo{number}{1} (\bibinfo{year}{2020}),
  \bibinfo{pages}{534--556}.
\newblock


\bibitem[Arafa et~al\mbox{.}(2022)]%
        {arafa2021timely}
\bibfield{author}{\bibinfo{person}{Ahmed Arafa}, \bibinfo{person}{Jing Yang},
  \bibinfo{person}{Sennur Ulukus}, {and} \bibinfo{person}{H~Vincent Poor}.}
  \bibinfo{year}{2022}\natexlab{}.
\newblock \showarticletitle{Timely status updating over erasure channels using
  an energy harvesting sensor: Single and multiple sources}.
\newblock \bibinfo{journal}{\emph{IEEE Transactions on Green Communications and
  Networking}} \bibinfo{volume}{6}, \bibinfo{number}{1} (\bibinfo{year}{2022}),
  \bibinfo{pages}{6--19}.
\newblock


\bibitem[Bertsekas(2012a)]%
        {bertsekas1995dynamic1}
\bibfield{author}{\bibinfo{person}{Dimitri~P Bertsekas}.}
  \bibinfo{year}{2012}\natexlab{a}.
\newblock \bibinfo{booktitle}{\emph{Dynamic programming and optimal control}}.
  Vol.~\bibinfo{volume}{1}.
\newblock \bibinfo{publisher}{Athena scientific Belmont, MA}.
\newblock


\bibitem[Bertsekas(2012b)]%
        {bertsekas1995dynamic2}
\bibfield{author}{\bibinfo{person}{Dimitri~P Bertsekas}.}
  \bibinfo{year}{2012}\natexlab{b}.
\newblock \bibinfo{booktitle}{\emph{Dynamic programming and optimal control}}.
  Vol.~\bibinfo{volume}{2}.
\newblock \bibinfo{publisher}{Athena scientific Belmont, MA}.
\newblock


\bibitem[Bertsekas and Shreve(2004)]%
        {bertsekas2004stochastic}
\bibfield{author}{\bibinfo{person}{Dimitir~P Bertsekas} {and}
  \bibinfo{person}{Steven Shreve}.} \bibinfo{year}{2004}\natexlab{}.
\newblock \bibinfo{booktitle}{\emph{Stochastic optimal control: the
  discrete-time case}}.
\newblock
\urldef\tempurl%
\url{http://web.mit.edu/dimitrib/www/soc.html}
\showURL{%
\tempurl}


\bibitem[Chakravorty and Mahajan(2020)]%
        {chakravorty2020remote}
\bibfield{author}{\bibinfo{person}{Jhelum Chakravorty} {and}
  \bibinfo{person}{Aditya Mahajan}.} \bibinfo{year}{2020}\natexlab{}.
\newblock \showarticletitle{Remote estimation over a packet-drop channel with
  Markovian state}.
\newblock \bibinfo{journal}{\emph{IEEE Trans. Automat. Control}}
  \bibinfo{volume}{65}, \bibinfo{number}{5} (\bibinfo{year}{2020}),
  \bibinfo{pages}{2016--2031}.
\newblock


\bibitem[Champati et~al\mbox{.}(2019)]%
        {champati2019performance}
\bibfield{author}{\bibinfo{person}{Jaya~Prakash Champati},
  \bibinfo{person}{Mohammad~H Mamduhi}, \bibinfo{person}{Karl~H Johansson},
  {and} \bibinfo{person}{James Gross}.} \bibinfo{year}{2019}\natexlab{}.
\newblock \showarticletitle{Performance characterization using {AoI} in a
  single-loop networked control system}. In \bibinfo{booktitle}{\emph{IEEE
  INFOCOM 2019-IEEE Conference on Computer Communications Workshops (INFOCOM
  WKSHPS)}}. IEEE, \bibinfo{pages}{197--203}.
\newblock


\bibitem[Chen and Ephremides(2023)]%
        {chen2023age}
\bibfield{author}{\bibinfo{person}{Yutao Chen} {and} \bibinfo{person}{Anthony
  Ephremides}.} \bibinfo{year}{2023}\natexlab{}.
\newblock \showarticletitle{Minimizing Age of Incorrect Information over a
  Channel with Random Delay}.
\newblock \bibinfo{journal}{\emph{arXiv preprint arXiv:2301.06150}}
  (\bibinfo{year}{2023}).
\newblock


\bibitem[Dinkelbach(1967)]%
        {dinkelbach1967nonlinear}
\bibfield{author}{\bibinfo{person}{Werner Dinkelbach}.}
  \bibinfo{year}{1967}\natexlab{}.
\newblock \showarticletitle{On nonlinear fractional programming}.
\newblock \bibinfo{journal}{\emph{Management science}} \bibinfo{volume}{13},
  \bibinfo{number}{7} (\bibinfo{year}{1967}), \bibinfo{pages}{492--498}.
\newblock


\bibitem[Durrett(2010)]%
        {durrett2019probability}
\bibfield{author}{\bibinfo{person}{Rick Durrett}.}
  \bibinfo{year}{2010}\natexlab{}.
\newblock \bibinfo{booktitle}{\emph{Probability: theory and examples}}.
\newblock \bibinfo{publisher}{Cambridge university press}.
\newblock


\bibitem[Guo and Kostina(2022)]%
        {guo2022optimal}
\bibfield{author}{\bibinfo{person}{Nian Guo} {and} \bibinfo{person}{Victoria
  Kostina}.} \bibinfo{year}{2022}\natexlab{}.
\newblock \showarticletitle{Optimal causal rate-constrained sampling of the
  Wiener process}.
\newblock \bibinfo{journal}{\emph{IEEE Trans. Automat. Control}}
  \bibinfo{volume}{67}, \bibinfo{number}{4} (\bibinfo{year}{2022}),
  \bibinfo{pages}{1776--1791}.
\newblock


\bibitem[Hui et~al\mbox{.}(2022)]%
        {hui2022real}
\bibfield{author}{\bibinfo{person}{Haiming Hui}, \bibinfo{person}{Shaoling Hu},
  {and} \bibinfo{person}{Wei Chen}.} \bibinfo{year}{2022}\natexlab{}.
\newblock \showarticletitle{Real Time Monitoring of Brownian Motions}.
\newblock \bibinfo{journal}{\emph{IEEE Transactions on Communications}}
  \bibinfo{volume}{70}, \bibinfo{number}{9} (\bibinfo{year}{2022}),
  \bibinfo{pages}{5867--5881}.
\newblock


\bibitem[Jog et~al\mbox{.}(2021)]%
        {jog2021channels}
\bibfield{author}{\bibinfo{person}{Varun Jog}, \bibinfo{person}{Richard~J La},
  \bibinfo{person}{Michael Lin}, {and} \bibinfo{person}{Nuno~C Martins}.}
  \bibinfo{year}{2021}\natexlab{}.
\newblock \showarticletitle{Channels, remote estimation and queueing systems
  with a utilization-dependent component: A unifying survey of recent results}.
\newblock \bibinfo{journal}{\emph{arXiv preprint arXiv:1905.04362}}
  (\bibinfo{year}{2021}).
\newblock


\bibitem[Kam et~al\mbox{.}(2020)]%
        {kam2020age}
\bibfield{author}{\bibinfo{person}{Clement Kam}, \bibinfo{person}{Sastry
  Kompella}, {and} \bibinfo{person}{Anthony Ephremides}.}
  \bibinfo{year}{2020}\natexlab{}.
\newblock \showarticletitle{Age of incorrect information for remote estimation
  of a binary markov source}. In \bibinfo{booktitle}{\emph{IEEE INFOCOM
  2020-IEEE Conference on Computer Communications Workshops (INFOCOM WKSHPS)}}.
  \bibinfo{pages}{1--6}.
\newblock


\bibitem[Kl{\"u}gel et~al\mbox{.}(2019)]%
        {klugel2019aoi}
\bibfield{author}{\bibinfo{person}{Markus Kl{\"u}gel},
  \bibinfo{person}{Mohammad~H Mamduhi}, \bibinfo{person}{Sandra Hirche}, {and}
  \bibinfo{person}{Wolfgang Kellerer}.} \bibinfo{year}{2019}\natexlab{}.
\newblock \showarticletitle{{AoI}-penalty minimization for networked control
  systems with packet loss}. In \bibinfo{booktitle}{\emph{IEEE INFOCOM
  WKSHPS}}. \bibinfo{pages}{189--196}.
\newblock


\bibitem[Maatouk et~al\mbox{.}(2020)]%
        {maatouk2020age}
\bibfield{author}{\bibinfo{person}{Ali Maatouk}, \bibinfo{person}{Saad
  Kriouile}, \bibinfo{person}{Mohamad Assaad}, {and} \bibinfo{person}{Anthony
  Ephremides}.} \bibinfo{year}{2020}\natexlab{}.
\newblock \showarticletitle{The age of incorrect information: A new performance
  metric for status updates}.
\newblock \bibinfo{journal}{\emph{IEEE/ACM Transactions on Networking}}
  \bibinfo{volume}{28}, \bibinfo{number}{5} (\bibinfo{year}{2020}),
  \bibinfo{pages}{2215--2228}.
\newblock


\bibitem[Moltafet et~al\mbox{.}(2022)]%
        {moltafet2022status}
\bibfield{author}{\bibinfo{person}{Mohammad Moltafet}, \bibinfo{person}{Markus
  Leinonen}, \bibinfo{person}{Marian Codreanu}, {and} \bibinfo{person}{Roy~D
  Yates}.} \bibinfo{year}{2022}\natexlab{}.
\newblock \showarticletitle{Status Update Control and Analysis under Two-Way
  Delay}.
\newblock \bibinfo{journal}{\emph{arXiv preprint arXiv:2208.06177}}
  (\bibinfo{year}{2022}).
\newblock


\bibitem[M{\"o}rters and Peres(2010)]%
        {morters2010brownian}
\bibfield{author}{\bibinfo{person}{Peter M{\"o}rters} {and}
  \bibinfo{person}{Yuval Peres}.} \bibinfo{year}{2010}\natexlab{}.
\newblock \bibinfo{booktitle}{\emph{Brownian motion}}.
  Vol.~\bibinfo{volume}{30}.
\newblock \bibinfo{publisher}{Cambridge University Press}.
\newblock


\bibitem[Nar and Ba{\c{s}}ar(2014)]%
        {nar2014sampling}
\bibfield{author}{\bibinfo{person}{Kamil Nar} {and} \bibinfo{person}{Tamer
  Ba{\c{s}}ar}.} \bibinfo{year}{2014}\natexlab{}.
\newblock \showarticletitle{Sampling multidimensional Wiener processes}. In
  \bibinfo{booktitle}{\emph{53rd IEEE Conference on Decision and Control}}.
  \bibinfo{pages}{3426--3431}.
\newblock


\bibitem[Oksendal(2013)]%
        {oksendal2013stochastic}
\bibfield{author}{\bibinfo{person}{Bernt Oksendal}.}
  \bibinfo{year}{2013}\natexlab{}.
\newblock \bibinfo{booktitle}{\emph{Stochastic differential equations: an
  introduction with applications}}.
\newblock \bibinfo{publisher}{Springer Science \& Business Media}.
\newblock


\bibitem[Ornee and Sun(2021)]%
        {ornee2019sampling}
\bibfield{author}{\bibinfo{person}{Tasmeen~Zaman Ornee} {and}
  \bibinfo{person}{Yin Sun}.} \bibinfo{year}{2021}\natexlab{}.
\newblock \showarticletitle{Sampling and remote estimation for the
  Ornstein-Uhlenbeck process through queues: Age of information and beyond}.
\newblock \bibinfo{journal}{\emph{IEEE/ACM Transactions on Networking}}
  \bibinfo{volume}{29}, \bibinfo{number}{5} (\bibinfo{year}{2021}),
  \bibinfo{pages}{1962--1975}.
\newblock


\bibitem[Pan et~al\mbox{.}(2022a)]%
        {pan2021minimizings}
\bibfield{author}{\bibinfo{person}{Jiayu Pan}, \bibinfo{person}{Ahmed~M
  Bedewy}, \bibinfo{person}{Yin Sun}, {and} \bibinfo{person}{Ness~B Shroff}.}
  \bibinfo{year}{2022}\natexlab{a}.
\newblock \showarticletitle{Age-optimal scheduling over hybrid channels}.
\newblock \bibinfo{journal}{\emph{IEEE Transactions on Mobile Computing}}
  (\bibinfo{year}{2022}).
\newblock


\bibitem[Pan et~al\mbox{.}(2022b)]%
        {pan2022optimizing}
\bibfield{author}{\bibinfo{person}{Jiayu Pan}, \bibinfo{person}{Ahmed~M
  Bedewy}, \bibinfo{person}{Yin Sun}, {and} \bibinfo{person}{Ness~B Shroff}.}
  \bibinfo{year}{2022}\natexlab{b}.
\newblock \showarticletitle{Optimizing sampling for data freshness: Unreliable
  transmissions with random two-way delay}. In \bibinfo{booktitle}{\emph{IEEE
  INFOCOM 2022-IEEE Conference on Computer Communications}}. IEEE,
  \bibinfo{pages}{1389--1398}.
\newblock


\bibitem[Pan et~al\mbox{.}(2023)]%
        {pan2022optimal}
\bibfield{author}{\bibinfo{person}{Jiayu Pan}, \bibinfo{person}{Ahmed~M
  Bedewy}, \bibinfo{person}{Yin Sun}, {and} \bibinfo{person}{Ness~B Shroff}.}
  \bibinfo{year}{2023}\natexlab{}.
\newblock \showarticletitle{Optimal sampling for data freshness: Unreliable
  transmissions with random two-way delay}.
\newblock \bibinfo{journal}{\emph{IEEE/ACM Transactions on Networking}}
  \bibinfo{volume}{31}, \bibinfo{number}{1} (\bibinfo{year}{2023}),
  \bibinfo{pages}{408--420}.
\newblock


\bibitem[Peskir and Shiryaev(2006)]%
        {peskir2006optimal}
\bibfield{author}{\bibinfo{person}{Goran Peskir} {and} \bibinfo{person}{Albert
  Shiryaev}.} \bibinfo{year}{2006}\natexlab{}.
\newblock \bibinfo{booktitle}{\emph{Optimal stopping and free-boundary
  problems}}.
\newblock \bibinfo{publisher}{Springer}.
\newblock


\bibitem[Resnick(2019)]%
        {resnick2019probability}
\bibfield{author}{\bibinfo{person}{Sidney Resnick}.}
  \bibinfo{year}{2019}\natexlab{}.
\newblock \bibinfo{booktitle}{\emph{A probability path}}.
\newblock \bibinfo{publisher}{Springer}.
\newblock


\bibitem[Shiryaev(1978)]%
        {shiryaev1978optimal}
\bibfield{author}{\bibinfo{person}{Albert~N Shiryaev}.}
  \bibinfo{year}{1978}\natexlab{}.
\newblock \bibinfo{booktitle}{\emph{Optimal stopping rules}}.
\newblock \bibinfo{publisher}{Springer Science \& Business Media}.
\newblock


\bibitem[Sun and Cyr(2019)]%
        {sun2019sampling}
\bibfield{author}{\bibinfo{person}{Yin Sun} {and} \bibinfo{person}{Benjamin
  Cyr}.} \bibinfo{year}{2019}\natexlab{}.
\newblock \showarticletitle{Sampling for data freshness optimization:
  Non-linear age functions}.
\newblock \bibinfo{journal}{\emph{Journal of Communications and Networks}}
  \bibinfo{volume}{21}, \bibinfo{number}{3} (\bibinfo{year}{2019}),
  \bibinfo{pages}{204--219}.
\newblock


\bibitem[Sun et~al\mbox{.}(2019)]%
        {sun2019age}
\bibfield{author}{\bibinfo{person}{Yin Sun}, \bibinfo{person}{Igor Kadota},
  \bibinfo{person}{Rajat Talak}, {and} \bibinfo{person}{Eytan Modiano}.}
  \bibinfo{year}{2019}\natexlab{}.
\newblock \showarticletitle{Age of information: A new metric for information
  freshness}.
\newblock \bibinfo{journal}{\emph{Synthesis Lectures on Communication
  Networks}} \bibinfo{volume}{12}, \bibinfo{number}{2} (\bibinfo{year}{2019}),
  \bibinfo{pages}{1--224}.
\newblock


\bibitem[Sun et~al\mbox{.}(2020)]%
        {sun2020sampling}
\bibfield{author}{\bibinfo{person}{Yin Sun}, \bibinfo{person}{Yury Polyanskiy},
  {and} \bibinfo{person}{Elif Uysal}.} \bibinfo{year}{2020}\natexlab{}.
\newblock \showarticletitle{Sampling of the Wiener process for remote
  estimation over a channel with random delay}.
\newblock \bibinfo{journal}{\emph{IEEE Transactions on Information Theory}}
  \bibinfo{volume}{66}, \bibinfo{number}{2} (\bibinfo{year}{2020}),
  \bibinfo{pages}{1118--1135}.
\newblock


\bibitem[Tang et~al\mbox{.}(2022)]%
        {tang2022sampling}
\bibfield{author}{\bibinfo{person}{Haoyue Tang}, \bibinfo{person}{Yin Sun},
  {and} \bibinfo{person}{Leandros Tassiulas}.} \bibinfo{year}{2022}\natexlab{}.
\newblock \showarticletitle{Sampling of the wiener process for remote
  estimation over a channel with unknown delay statistics}. In
  \bibinfo{booktitle}{\emph{ACM MobiHoc}}. \bibinfo{pages}{51--60}.
\newblock


\bibitem[Tsai and Wang(2021)]%
        {tsai2021unifying}
\bibfield{author}{\bibinfo{person}{Cho-Hsin Tsai} {and}
  \bibinfo{person}{Chih-Chun Wang}.} \bibinfo{year}{2021}\natexlab{}.
\newblock \showarticletitle{Unifying {AoI} minimization and remote
  estimation—Optimal sensor/controller coordination with random two-way
  delay}.
\newblock \bibinfo{journal}{\emph{IEEE/ACM Transactions on Networking}}
  \bibinfo{volume}{30}, \bibinfo{number}{1} (\bibinfo{year}{2021}),
  \bibinfo{pages}{229--242}.
\newblock


\bibitem[Wu et~al\mbox{.}(2017)]%
        {wu2017optimal}
\bibfield{author}{\bibinfo{person}{Xianwen Wu}, \bibinfo{person}{Jing Yang},
  {and} \bibinfo{person}{Jingxian Wu}.} \bibinfo{year}{2017}\natexlab{}.
\newblock \showarticletitle{Optimal status update for age of information
  minimization with an energy harvesting source}.
\newblock \bibinfo{journal}{\emph{IEEE Transactions on Green Communications and
  Networking}} \bibinfo{volume}{2}, \bibinfo{number}{1} (\bibinfo{year}{2017}),
  \bibinfo{pages}{193--204}.
\newblock


\bibitem[Yates(2015)]%
        {yates2015lazy}
\bibfield{author}{\bibinfo{person}{Roy~D Yates}.}
  \bibinfo{year}{2015}\natexlab{}.
\newblock \showarticletitle{Lazy is timely: Status updates by an energy
  harvesting source}. In \bibinfo{booktitle}{\emph{2015 IEEE International
  Symposium on Information Theory (ISIT)}}. \bibinfo{pages}{3008--3012}.
\newblock


\bibitem[Yates et~al\mbox{.}(2021)]%
        {yates2021age}
\bibfield{author}{\bibinfo{person}{Roy~D Yates}, \bibinfo{person}{Yin Sun},
  \bibinfo{person}{D~Richard Brown}, \bibinfo{person}{Sanjit~K Kaul},
  \bibinfo{person}{Eytan Modiano}, {and} \bibinfo{person}{Sennur Ulukus}.}
  \bibinfo{year}{2021}\natexlab{}.
\newblock \showarticletitle{Age of information: An introduction and survey}.
\newblock \bibinfo{journal}{\emph{IEEE Journal on Selected Areas in
  Communications}} \bibinfo{volume}{39}, \bibinfo{number}{5}
  (\bibinfo{year}{2021}), \bibinfo{pages}{1183--1210}.
\newblock


\end{thebibliography}


\appendix




\section{Proof of Lemma~\ref{cost1stage}}\label{cost1stageapp}
We denote $X_t \triangleq w+W_t$ as the Wiener process starting from the initial state $X_0=w$.
Using the definition of $g(w;\tau)$ in \eqref{eq-costgtau}, we have
\begin{align}
\nonumber g(w;\tau) & \triangleq \mathbb{E} \left[   \int_{0}^{\tau+Y}X_t^2 dt - \text{mse}_{\text{opt}} (\tau+Y) \  \right] \\
& = \mathbb{E} \left[   \int_{0}^{\tau}X_t^2 dt +  \int_{\tau}^{\tau+Y}X_t^2 dt - \text{mse}_{\text{opt}} (\tau+Y) \  \right]. \label{eq-cost1divide}
\end{align}
Using the strong Markov property of the wiener process, $\{ X_{\tau+t},t\ge0 \}$ has the same distribution as $\{ X_\tau+W_t,t\ge0 \}$. The second term of \eqref{eq-cost1divide} turns to:
\begin{align}
\nonumber &  \mathbb{E} \left[   \int_{\tau}^{\tau+Y}X_t^2 dt  \  \right]  \\ \nonumber = & \mathbb{E} \left[   \int_{\tau}^{\tau+Y}(X_\tau+W_t)^2 dt  \  \right]\\  \nonumber = & \mathbb{E} \left[   Y X_\tau^2+2X_\tau \int_0^Y W_t dt+\int_0^Y W_t^2 dt \right] \\  = & \mathbb{E} \left[ Y \right] \mathbb{E} \left[   X_\tau^2 \right] +2 \mathbb{E} \left[ X_\tau \right] \mathbb{E} \left[ \int_0^Y W_t dt \right]+ \mathbb{E} \left[\int_0^Y W_t^2 dt \right],\label{eq-cost1divide2}
\end{align}
the last equality holds because the delay $Y$ is independent of $X_\tau$.
By \cite[Theorem~2.5.1]{morters2010brownian}, $1/3 W_t^3 -  \int_0^t W_r dr$ and $1/6 W_t^4 -  \int_0^t W_r^2 dr $ are martingales, respectively. So
\begin{align}
\nonumber  \mathbb{E} \left[   \int_{\tau}^{\tau+Y}X_t^2 dt  \  \right] &  = \mathbb{E} \left[ Y \right]  \mathbb{E} \left[   X_\tau^2 \right] +\frac{2}{3} \mathbb{E} \left[ X_\tau \right] \mathbb{E} \left[   \mathbb{E} \left[ W_Y^3 \ | Y \right] \right]+ \frac{1}{6} \mathbb{E} \left[   \mathbb{E} \left[ W_Y^4 \ | Y \right] \right] \\  \nonumber &
=   \mathbb{E} \left[ Y \right]  \mathbb{E} \left[   X_\tau^2 \right] +\frac{2}{3} \mathbb{E} \left[ X_\tau \right]   \mathbb{E} \left[ W_Y^3  \right]+ \frac{1}{6}   \mathbb{E} \left[ W_Y^4  \right] \\ & 
   =  \mathbb{E} \left[ Y \right] \mathbb{E} \left[   X_\tau^2 \right] + \frac{1}{2}\mathbb{E} \left[ Y^2 \right]. \label{eq-cost1divide3}
\end{align}
Combined with \eqref{eq-cost1divide2} and \eqref{eq-cost1divide3}, we finally get \eqref{eq-cost1divide0}.

Before showing that $g(w;\tau) =  g(w,v,\text{mse}_{\text{opt}})$, we need Lemma~\ref{cubicwiener}:
\begin{lemma}\label{cubicwiener}
If a finite stopping time $\tau$ satisfies that $\{ W_t, 0\le t \le \tau  \}$ is bounded, then 
\begin{align}
\mathbb{E} \left[  \int_{0}^{\tau} W_t dt \right] = \frac{1}{3} \mathbb{E} \left[   W_\tau^3 \right].
\end{align} 
\end{lemma}
\begin{proof}
By \cite[Theorem~2.5.1]{morters2010brownian} and \cite[Theorem~8.5.1]{durrett2019probability}, $1/3 W_t^3 -  \int_0^t W_s ds$ is a martingale for any given positive value $t$. Note that for any $n=1,2,\ldots$, $\tau \wedge n$ is obviously bounded. Then, we have  
\begin{align}
\mathbb{E} \left[  \int_{0}^{\tau \wedge n } W_t dt \right] = \frac{1}{3} \mathbb{E} \left[   W_{\tau \wedge n}^3 \right].
\end{align} 
Since $\tau$ is finite, $\tau\wedge n \rightarrow \tau, W_{\tau \wedge n} \rightarrow W_{\tau}$ almost surely.
Since $\{ W_t, 0\le t \le \tau  \}$ is bounded, using Dominated convergence theorem \cite[Theorem~5.3.3]{resnick2019probability}, 
\begin{align}
\lim_{n\rightarrow \infty} \mathbb{E} \left[   W_{\tau \wedge n}^3 \right] =  \mathbb{E} \left[   W_\tau^3 \right].
\end{align} Using Monotone Convergence Theorem \cite[Theorem~5.3.1]{resnick2019probability}, $\mathbb{E} \left[ \tau \wedge n \right]\rightarrow \mathbb{E} \left[ \tau \right]$. This leads to
\begin{align}
\lim_{n\rightarrow \infty} \mathbb{E} \left[    \int_{0}^{\tau } W_t dt - \int_{0}^{\tau \wedge n } W_t dt \right] \le \lim_{n\rightarrow \infty} \mathbb{E} \left[   (\tau - \tau\wedge n) \right]\times b = 0,
\end{align} where $b$ is an upper bound of $\{ W_t, 0\le t \le \tau  \}$. So we have
\begin{align}
\mathbb{E} \left[    \int_{0}^{\tau } W_t dt \right] = \lim_{n\rightarrow \infty} \mathbb{E} \left[  \int_{0}^{\tau \wedge n } W_t dt \right] =  \frac{1}{3}  \lim_{n\rightarrow \infty} \mathbb{E} \left[  W^3_{\tau\wedge n}  \right] =  \frac{1}{3}  \mathbb{E} \left[  W^3_{\tau}  \right] . \
\end{align} 
\end{proof}
Now we start to prove $g(w;\tau) =  g(w,v,\text{mse}_{\text{opt}})$. 
If $|w|>v$, then $\tau=0$, and $\mathbb{E} \left[ Y \right] X^2_\tau = \mathbb{E} \left[ Y \right] w^2$. Therefore, 
\begin{align}
g(w;\tau)  = \mathbb{E} \left[ Y \right] w^2 + \frac{1}{2} \mathbb{E} \left[ Y^2 \right]  - \mathbb{E} \left[ Y \right] \text{mse}_{\text{opt}}= g(w,v,\text{mse}_{\text{opt}}).
\end{align}
If $|w|\le v$, then $\mathbb{E} \left[ Y \right] X_\tau^2 = \mathbb{E} \left[ Y \right] v^2$. By \cite[Theorem~8.5.5]{durrett2019probability}, $\mathbb{E} \left[ \tau \right] = v^2 - w^2$. So we have
\begin{align}
g(w;\tau) = \mathbb{E}^w \left[ \int_{0}^{\tau}  X^2_{t} dt  \right]  - \text{mse}_{\text{opt}} (v^2 - w^2) + \mathbb{E} \left[ Y \right]v^2 +\frac{1}{2} \mathbb{E} \left[ Y^2 \right] - \mathbb{E} \left[ Y \right] \text{mse}_{\text{opt}}.
\end{align}
Since $X_0=w$ and the Wiener process has strong Markov property, the first term becomes
\begin{align}
\nonumber \mathbb{E}^w \left[ \int_{0}^{\tau}  X^2_{t} dt  \right] & =  \mathbb{E} \left[ \int_{0}^{\tau}  (w+W_t)^2 dt  \right] \\
\nonumber & =  \mathbb{E} \left[ \int_{0}^{\tau}  w^2+2wW_t+W_t^2 dt  \right]\\
 & = w^2(v^2-w^2) + \frac{2}{3}w  \mathbb{E} \left[ W_\tau^3  \right] + \frac{1}{6}  \mathbb{E} \left[ W_\tau^4  \right].
\end{align} The last equality holds due to Lemma~\ref{cubicwiener} and \cite[Lemma~3]{sun2020sampling}. From \cite[Theorem~2.49]{morters2010brownian}, we have
\begin{align}
W_\tau = \left\{
\begin{array}{lll}
v-w & \ \text{with probability } \frac{v+w}{2v}, \vspace{1mm} \\
-v-w & \ \text{with probability } \frac{v-w}{2v}.
\end{array}
\right.
\end{align}
Then, we have
\begin{align}
\mathbb{E} \left[ W_\tau^3  \right] & = \frac{v+w}{2v}(v-w)^3 - \frac{v-w}{2v}(v+w)^3  = -2w(v^2 - w^2),\\
\mathbb{E} \left[ W_\tau^4  \right] & = \frac{v+w}{2v}(v-w)^4 + \frac{v-w}{2v}(v+w)^4 = \frac{v^2 - w^2}{2v}(2v^3+6vw^2) = (v^2-w^2) (v^2+3w^2).
\end{align}
This gives
\begin{align}
\nonumber \mathbb{E}^w \left[ \int_{0}^{\tau}  X^2_{t} dt  \right] & = w^2(v^2-w^2)-\frac{4}{3} w^2(v^2 - w^2) + \frac{1}{6} (v^2-w^2) (v^2+3w^2) \\
 & =  \frac{1}{6}(v^4 - w^4).
\end{align}
Therefore, if $|w|<v$,
\begin{align}
 g(w;\tau) & = \frac{1}{6}(v^4 - w^4)  - \text{mse}_{\text{opt}} (v^2 - w^2) + \mathbb{E} \left[ Y \right]v^2 +\frac{1}{2} \mathbb{E} \left[ Y^2 \right] - \mathbb{E} \left[ Y \right]\text{mse}_{\text{opt}} = g(w,v,\text{mse}_{\text{opt}}).
\end{align} This ends our proof.


\section{Proof of Proposition~\ref{prop-aware}}\label{appprop-aware}
The proof is modified from \cite{arafa2021timely}, but we strictly extends \cite{arafa2021timely} in two-folds: (i) we consider the square estimation error $(W_t - \hat{W}_t)^2$, which is a more complicated metric than the age $\Delta_t$ considered in \cite{arafa2021timely}. Note that in the special case where the sampling time is independent of the Wiener process, $ \mathbb{E} [ (W_t - \hat{W}_t)^2] = \Delta_t$; (ii) The process $(W_t - \hat{W}_t)^2$ of two consecutive epochs are correlated, while in \cite{arafa2021timely}, the process $\Delta_t$ of that are independent.

We denote $l^e_j \triangleq S_{j,M_j} - S_{j-1,M_{j-1}}$ as the inter sampling time of the $j$th epoch. 
We also denote $\mathcal{H}_{j}$ as the history information of sampling times, transmission times and the Wiener process until $S_{j,M_j}$. Then, by the definition of $\Pi_{\text{signal-aware}}$ in Section~\ref{subsection-sampling}, $l^e_j$ is bounded by a stopping time, denoted by $\tilde{l}^e_j$, and we have $\mathbb{E} \left[ W^4_{\tilde{l}^e_j} \right]<\infty$. 


For simplicity, let us denote $r(T) = \int_{0}^{T} (W_t - \hat{W}_t)^2 dt $.
We denote $R_j = \int_{D_{j-1,M_{j-1}}}^{D_{j,M_j}} (W_t - \hat{W}_t)^2 dt =  \int_{D_{j-1,M_{j-1}}}^{D_{j,M_j}} (W_t - W_{S_{j-1,M_{j-1}}})^2 dt $, and $N(T)$ as the largest epoch number $j$ such that $S_{j,M_j}<T$, i.e., the number of successful samples attempted until $T$. Then, we have 
\begin{align}
& \sum_{j=1}^{\infty} R_j \mathds{1}_{j\le n-1} \le \frac{r(T)}{T} \le \sum_{j=1}^{\infty} R_j \mathds{1}_{j \le n},  \ \text{if } T\in [S_{n,M_n},D_{n,M_n}],\\
& \sum_{j=1}^{\infty} R_j \mathds{1}_{j\le n} \le \frac{r(T)}{T} \le \sum_{j=1}^{\infty} R_j \mathds{1}_{j\le n+1},  \ \text{if } T\in [D_{n,M_n},S_{n+1,M_{n+1}}].
\end{align}
This tells that 
\begin{align}
\frac{\sum_{j=1}^{\infty} R_j \mathds{1}_{j \le N(T)-1}}{T} \le \frac{r(T)}{T} \le \frac{ \sum_{j=1}^{\infty} R_j \mathds{1}_{j \le N(T)+1}}{T}. 
\end{align}
Then, we have the following lemma:  
\begin{lemma}\label{lemma-tail0}
\begin{align}
\lim_{T\rightarrow \infty} \frac{\mathbb{E} \left[ R_{N(T)}+R_{N(T)+1} \right]}{T} = 0.\label{eq-tail0}
\end{align}
\end{lemma}
\begin{proof} See Appendix~\ref{apptail0}. \end{proof}
Lemma~\ref{lemma-tail0} tells that the "residual terms" $R_{N(T)}/T$ and $R_{N(T)+1}/T$ vanishes as time $T$ goes to infinity. Therefore, instead of $r(T)$, we can analyze $\{ R_j\}_j$. 
We have 
\begin{align}
\limsup_{T\rightarrow \infty} \frac{\mathbb{E} \left[ r(T) \right]}{T} = \limsup_{T\rightarrow \infty}\frac{1}{T} \mathbb{E} \left[ \sum_{j=1}^{\infty} R_j \mathds{1}_{j\le N_T} \right]. 
\end{align} Here we denote $N_T = N(T)+1$ for simplicity.
We denote $R_{j,m} $, $l_{j,m}$ as the integral of $(W_t - \hat{W}_t)^2$ and $t-D_{j-1}$ between $D_{j-1}$ and $m$th delivery time of $j$th epoch, respectively, given that there are $m$ transmissions at $j$th epoch. Then, 
\begin{align}
\nonumber & \mathbb{E} \left[ R_{j,m} \mathds{1}_{j\le N_T}  \right] \\ \nonumber
= &  \mathbb{E}_{\mathcal{H}_{j-1}} \left[   \mathbb{E} \left[ R_{j,m} \mathds{1}_{j\le N_T}  \right] \Big{|} \mathcal{H}_{j-1} \right]\\
\overset{(i)}{=} &  \mathbb{E}_{\mathcal{H}_{j-1}} \left[   \mathbb{E} \left[ R_{j,m}  \big{|} \mathcal{H}_{j-1} \right] \mathds{1}_{j\le N_T}  \Big{|} \mathcal{H}_{j-1} \right]. \label{eq-rim}
\end{align}
Condition $(i)$ is because $j\le N_T$ (i.e., $j-1 \le N(T)$) is fixed given $\mathcal{H}_{j-1}$. Similarly, 
\begin{align}
\nonumber & \mathbb{E} \left[ l_{j,m} \mathds{1}_{j\le N_T}  \right] \\
= &  \mathbb{E}_{\mathcal{H}_{j-1}} \left[   \mathbb{E} \left[ l_{j,m}  | \mathcal{H}_{j-1} \right] \mathds{1}_{j\le N_T}  \big{|} \mathcal{H}_{j-1} \right].
\end{align}
We then find out the lower bound of $ \sum_{j=1}^{\infty} R_j \mathds{1}_{j\le N_T}$, equals to $T R_{\text{min}}$, in the following equations \eqref{eq-lowbound}. Here, condition (i) is due to monotone convergence theorem, and condition (ii) is due to \eqref{eq-rim}. The value $R^*(\mathcal{H}_{j-1})$ is the minimum of the fraction $\frac{ \sum_{m=1}^{\infty} \alpha^{m-1} (1-\alpha)   \mathbb{E} \left[   R_{j,m} \big{|} \mathcal{H}_{j-1} \right]}{ \sum_{j=m}^{\infty} \alpha^{m-1} (1-\alpha)   \mathbb{E} \left[   l_{j,m} \big{|} \mathcal{H}_{j-1} \right]}$, and $R_{min}$ is the minimum of $R^*(\mathcal{H}_{j-1})$ over all $\mathcal{H}_{j-1}$. Note that any policy that achieves $R^*(\mathcal{H}_{j-1})$ is not related to $\mathcal{H}_{j-1}$. Thus, the inequalities hold if we can find out such a policy that is not related to $\mathcal{H}_{j-1}$. In addition, $R^*(\mathcal{H}_{j-1}) = R_{\text{min}}$.
\begin{align}
\nonumber & \mathbb{E} \left[  \sum_{j=1}^{\infty} R_j \mathds{1}_{j\le N_T} \right]\\
\nonumber = & \mathbb{E} \left[   \sum_{j=1}^{\infty}   \sum_{m=1}^{\infty} R_{j,m} \Pi_{k=1}^{m-1} \mathds{1}_{e_{j,k}} \mathds{1}_{e^c_{j,m}} \mathds{1}_{j\le N_T} \right]\\
\nonumber  \overset{(j)}{=} &  \sum_{j=1}^{\infty}   \sum_{m=1}^{\infty} \mathbb{E} \left[   R_{j,m} \Pi_{k=1}^{m-1} \mathds{1}_{e_{j,k}} \mathds{1}_{e^c_{j,m}} \mathds{1}_{j\le N_T} \right]\\
\nonumber = &   \sum_{j=1}^{\infty}   \sum_{m=1}^{\infty} \alpha^{m-1} (1-\alpha) \mathbb{E} \left[   R_{j,m}  \mathds{1}_{j\le N_T} \right] \\
\nonumber \overset{(ii)}{=} &  \sum_{j=1}^{\infty}   \sum_{m=1}^{\infty} \alpha^{m-1} (1-\alpha)   \mathbb{E}_{\mathcal{H}_{j-1}} \left[   \mathbb{E} \left[   R_{j,m} \big{|} \mathcal{H}_{j-1} \right]  \mathds{1}_{j\le N_T}  \Big{|} \mathcal{H}_{j-1} \right]\\
\nonumber = &  \sum_{j=1}^{\infty}    \mathbb{E}_{\mathcal{H}_{j-1}} \left[  \sum_{m=1}^{\infty} \alpha^{m-1} (1-\alpha)   \mathbb{E} \left[   R_{j,m} \big{|} \mathcal{H}_{j-1} \right]  \mathds{1}_{i\le N_T}  \Big{|} \mathcal{H}_{j-1} \right]\\
\nonumber = &   \sum_{j=1}^{\infty}    \mathbb{E}_{\mathcal{H}_{j-1}} \left[  \sum_{j=m}^{\infty} \alpha^{m-1} (1-\alpha)   \mathbb{E} \left[   l_{j,m} \big{|} \mathcal{H}_{j-1} \right]  \mathds{1}_{j\le N_T}  \frac{ \sum_{m=1}^{\infty} \alpha^{m-1} (1-\alpha)   \mathbb{E} \left[   R_{j,m} \big{|} \mathcal{H}_{j-1} \right]}{ \sum_{j=m}^{\infty} \alpha^{m-1} (1-\alpha)   \mathbb{E} \left[   l_{j,m} \big{|} \mathcal{H}_{j-1} \right]} \Big{|} \mathcal{H}_{j-1} \right]\\
\nonumber \ge &   \sum_{j=1}^{\infty}    \mathbb{E}_{\mathcal{H}_{j-1}} \left[  \sum_{m=1}^{\infty} \alpha^{m-1} (1-\alpha)   \mathbb{E} \left[   l_{j,m} \big{|} \mathcal{H}_{j-1} \right]  \mathds{1}_{j\le N_T} \cdot R^*(\mathcal{H}_{j-1}) \right] \\
\nonumber  \ge &   \sum_{j=1}^{\infty}    \mathbb{E}_{\mathcal{H}_{j-1}} \left[  \sum_{m=1}^{\infty} \alpha^{m-1} (1-\alpha)   \mathbb{E} \left[   l_{j,m} \big{|} \mathcal{H}_{j-1} \right]  \mathds{1}_{j\le N_T}  \right] R_{\text{min}}\\
\nonumber  =&  \mathbb{E} \left[  \sum_{j=1}^{\infty} l^e_j  \mathds{1}_{j\le N_T} \right]R_{\text{min}}\\
  \ge & T R_{\text{min}}. \label{eq-lowbound}
\end{align}
Divide $T$ on both sides and take the limit of $T$, then we can get
$
\lim_{T\rightarrow \infty}\frac{1}{T} \mathbb{E} \left[  \sum_{j=1}^{\infty} R_j \mathds{1}_{j\le N_T} \right] = R_{\text{min}},
$ i.e., all of the inequalities will hold if $l^e_j$ is independent of $\mathcal{H}_{j-1}$, and we can find out an optimal policy that solves \eqref{eq-singlepoch-mse}. This ends our proof of Proposition~\ref{prop-aware}.

\section{Proof of Lemma \ref{lemma-tail0}}\label{apptail0} 

Using \cite[Lemma~3]{sun2020sampling}, for any finite stopping time $\tau$, we have 
\begin{align}
\mathbb{E} \left[ \int_0^{\tau} W_t^2 dt \right] = \frac{1}{6} \mathbb{E} \left[ W_{\tau}^4 \right]. \label{eq-four}
\end{align}

Denote $\tilde{R}_n =  \int_0^{\tilde{l}^e_n+Y_{n,M_n}} W_t^2 dt $, where $\tilde{l}^e_n$ is a stopping time upper bound denoted in Appendix~\ref{appprop-aware}, independent of $Y_{n,M_n}$. Then, 
\begin{align}
\nonumber \mathbb{E} \left[ \tilde{R}_n \right] = & \mathbb{E} \left[  \mathbb{E} \left[  \int_0^{\tilde{l}^e_n+Y_{n,M_n}} W_t^2 dt \ \bigg{|} M_n \right] \right]\\
 \nonumber = & \frac{1}{6}  \mathbb{E} \left[  \mathbb{E} \left[ W_{\tilde{l}^e_n+Y_{n,M_n}}^4   \ \bigg{|} M_n  \right] \right] \\
 \nonumber  = & \frac{1}{6}  \mathbb{E} \left[  W_{\tilde{l}^e_n+Y_{n,M_n}}^4   \right]\\
 \nonumber  = & \frac{1}{6}  \mathbb{E} \left[  \left(W_{\tilde{l}^e_n+Y_{n,M_n}} - W_{\tilde{l}_n} + W_{\tilde{l}^e_n}\right)^4   \right] \\
 \nonumber  = &  \frac{1}{6}  \mathbb{E} \left[  W_{\tilde{l}^e_n}^4   \right] + \frac{1}{6}  \mathbb{E} \left[  \left(W_{\tilde{l}^e_n+Y_{n,M_n}} - W_{\tilde{l}^e_n}\right) ^4   \right]+
\mathbb{E} \left[  W_{\tilde{l}^e_n}^2     \right]\mathbb{E} \left[  \left(W_{\tilde{l}^e_n+Y_{n,M_n}} - W_{\tilde{l}^e_n}\right) ^2 \right] \\  & +
 \frac{2}{3}\mathbb{E} \left[  W_{\tilde{l}^e_n}^3     \right]\mathbb{E} \left[  W_{\tilde{l}^e_n+Y_{n,M_n}} - W_{\tilde{l}^e_n} \right] + \frac{2}{3} \mathbb{E} \left[  W_{\tilde{l}^e_n}     \right]\mathbb{E} \left[  \left(W_{\tilde{l}^e_n+Y_{n,M_n}} - W_{\tilde{l}^e_n}\right) ^3 \right]. 
\end{align}
According to strong Markov property of the Wiener process, $W_{\tilde{l}^e_n+Y_{n,M_n}} - W_{\tilde{l}^e_n}$ is independent of $W_{\tilde{l}^e_n}$. Using \cite[Theorem 2.44 and Theorem 2.48]{morters2010brownian}, for any finite stopping time $\tau$, $\mathbb{E} \left[ W^2_\tau \right] = \mathbb{E} \left[ \tau \right]$ and $\mathbb{E} \left[ W_\tau \right]=0$. 
Both $\tilde{l}^e_n$ and $Y_{n,M_n}$ are finite and $ \mathbb{E} \left[  W_{\tilde{l}^e_n}^4   \right]<\infty$. So 
\begin{align}
\mathbb{E} \left[ \tilde{R}_n \right] & =  \frac{1}{6}  \mathbb{E} \left[  W_{\tilde{l}^e_n}^4   \right] + \frac{1}{6}  \mathbb{E} \left[  W_{Y_{n,M_n}}  ^4   \right]+
\mathbb{E} \left[  \tilde{l}^e_n     \right]\mathbb{E} \left[  Y_{n,M_n} \right]  <\infty. \label{eq-finite}
\end{align}
Also, since mse $(W_t - \hat{W}_t)^2$ is nonnegative and $l^e_n\le \tilde{l}^e_n$,
\begin{align}
R_n =  \int_{D_{n-1,M_{n-1}}}^{D_{n,M_n}} (W_t - W_{S_{n-1,M_{n-1}}})^2 dt =  \int_{Y_{n-1,M_{n-1}}}^{l^e_n+Y_{n,M_{n}}} W_t^2 dt \le \tilde{R}_n. \label{eq-rnsmall}
\end{align}
Using \eqref{eq-rnsmall}, we have that 
\begin{align}
\nonumber & \mathbb{E} \left[ (R_n+R_{n+1}) \mathds{1}_{l^e_n+l^e_{n+1}>T-t} \Big{|} S_{n-1,M_{n-1}}=t \right]\\
\nonumber \le & \mathbb{E} \left[ (\tilde{R}_n+\tilde{R}_{n+1}) \mathds{1}_{\tilde{l}^e_n+\tilde{l}^e_{n+1}>T-t} \Big{|} S_{n-1,M_{n-1}}=t \right]\\
\nonumber = &  \mathbb{E} \left[ (\tilde{R}_n+\tilde{R}_{n+1}) \mathds{1}_{\tilde{l}^e_n+\tilde{l}^e_{n+1}>T-t} \right]\\
\triangleq & F(T-t).
\end{align}
The first equality holds because $\tilde{R}_n,\tilde{R}_{n+1},\tilde{l}^e_n,\tilde{l}^e_{n+1}$ are independent of $\mathcal{H}_{n-1}$. 
By \eqref{eq-finite}, $F(0)<\infty$, $F(\delta)$ is monotone decreasing, and $F(\delta)\rightarrow 0$ as $\delta\rightarrow \infty$. We trivially set $S_{-1,M_{-1}}=S_{0,M_0}=0$, so we have $N(T)\ge 0$ and $R_0=0$. We have 
\begin{align}
\nonumber R_{N(T)}+R_{N(T)+1} &= \sum_{n=0}^{\infty} (R_n+R_{n+1}) \mathds{1}_{N(T)=n} =  \sum_{n=0}^{\infty} (R_n+R_{n+1}) \mathds{1}_{S_{n,M_n}\le T, S_{n+1,M_{n+1}} >T} \\
& \le \sum_{n=0}^{\infty} (R_n+R_{n+1}) \mathds{1}_{S_{n-1,M_{n-1}}\le T, S_{n+1,M_{n+1}} >T}.
\end{align}
Therefore, 
\begin{align}
\nonumber \mathbb{E} \left[ R_{N(T)}+R_{N(T)+1}  \right] & \le  \sum_{n=0}^{\infty} \mathbb{E} \left[ (R_n+R_{n+1}) \mathds{1}_{S_{n-1,M_{n-1}}\le T, S_{n+1,M_{n+1}} >T} \right] \\
\nonumber & =  \sum_{n=0}^{\infty} \int_0^T \mathbb{E} \left[ (R_{n}+R_{n+1}) \mathds{1}_{S_{n+1,M_{n+1}}>T} \Big{|}S_{n-1,M_{n-1}}=t \right] dP_{S_{n-1,M_{n-1}}}(t)\\
\nonumber & \le  \sum_{n=0}^{\infty} \int_0^T F(T-t) dP_{S_{n-1,M_{n-1}}}(t) \\
& =  \sum_{n=2}^{\infty} \int_0^T F(T-t) dP_{S_{n-1,M_{n-1}}}(t) + 2F(T),
\end{align}
where $P_{S_{n-1,M_{n-1}}}(t) = P(S_{n-1,M_{n-1}}\le t)$. Note that $S_{n-1,M_{n-1}} \le t$ is equivalent to $N(T)\ge n-1$. So 
\begin{align}
\mathbb{E} \left[ N(t)+1 \right] = \sum_{n=1}^{\infty} P(N(t) +1 \ge n) =  \sum_{n=2}^{\infty} P(N(t) \ge n-1)+1 =  \sum_{n=2}^{\infty} P_{S_{n-1,M_{n-1}}}(t)+1.
\end{align}
So we have $\mathbb{E} \left[ N(t) \right]= \sum_{n=2}^{\infty} P_{S_{n-1,M_{n-1}}}(t) $ and 
\begin{align}
 \mathbb{E} \left[ R_{N(T)}+R_{N(T)+1}  \right]  
\le  \int_0^T F(T-t) d\mathbb{E} \left[ N(t) \right] + 2F(T).
\end{align}
Note that $F(T)$ vanishes to $0$ as $T\rightarrow \infty$. Following the same steps as \cite[Appendix C1]{wu2017optimal}, we have $ \int_0^T F(T-t) d\mathbb{E} \left[ N(t) \right]/T\rightarrow 0$ as $T\rightarrow \infty$. This ends our proof.


\section{Proof of Lemma~\ref{derivative1}}\label{appderivative1}
In this appendix, for simplicity, we will replace the per stage cost $g(w,v,\text{mse}_{\text{opt}})$ by $g(w,v)$.
\subsection{Preliminary}

\begin{definition}\label{def-5}
Let $Y_1,\cdots,Y_n$ as an i.i.d. sequence with the same distribution as the channel delay $Y$, and $a_1\cdots a_n$ as any nonnegative sequence. For any real value $w$, we denote the event $c_n(w) = \{ |w+W_{Y_1}|\ge a_1, |w+W_{Y_1}+W_{Y_2}| \ge a_2,\cdots, |w+W_{Y_1}+\cdots + W_{Y_n}|\ge a_n \}$. If $n=0$, we denote $c_0(w)$ as simply the whole set. 
Denote $f_{w+W_{Y_1} + \cdots + W_{Y_{n+1}} } (x)$ as the conditional probability density function (pdf) of $w+W_{Y_1} + \cdots + W_{Y_{n+1}}$ with the condition $c1_n(w)$, multiplied by a probability $\mathbb{P}(c_n(w))$. In other words, 
\begin{align}
f_{w+W_{Y_1} + \cdots + W_{Y_{n+1}} } (x) =  \frac{d}{dx} \mathbb{P} \left( \{ w+W_{Y_1} + \cdots + W_{Y_{n+1}} \le x \} \cap c_n(w) \right). \label{eq-jointpdf}
\end{align}
\end{definition}
Note that $f_{w+W_Y}(x)$ is equal to the pdf of $w+W_Y$ at $x$, since $c_0(w)$ is the whole set.

\begin{lemma} \label{lemma-prerequisite}
Suppose that $\mathbb{E} \left[ Y \right]<\infty$, and there exists $\epsilon>0$ (which can be arbitrary small), such that $Y\ge \epsilon$. Then, the following conditions hold.

(a) For any $x\in \mathbb{R}$, $f_{W_Y}(x)$ is continuously differentiable in $x$. In addition, $f_{W_Y}(x)$, $f'_{W_Y}(x)$, $f''_{W_Y}(x)$ are bounded, thus $f_{W_Y}(x)$, $f'_{W_Y}(x)$ are both uniformly continuous.


(b) Almost surely, ${\text{\large $\mathds{1}$}}_{c_n(w),c_n(w+\Delta w)} \rightarrow {\text{\large $\mathds{1}$}}_{c_n(w)}$.  

(c) 
For all $w,x \ge 0$, we have $f_{w+W_{Y_1} + \cdots + W_{Y_{n+1}} } (x) \ge f_{w+W_{Y_1} + \cdots + W_{Y_{n+1}} } (-x)$ for all $n=0,1,\ldots$. In addition, $\frac{d}{dw}f_{w+W_{Y_1} + \cdots + W_{Y_{n+1}} } (x)$ is continuous and bounded in $w$.

\end{lemma}
\begin{proof}
We first show Lemma~\ref{lemma-prerequisite}(a). Note that $Y\ge \epsilon$, and $W_y$ is normally distributed with variance $y$. Thus, $f_{W_y} (x) = 1/\sqrt{y}  e^{-0.5 x^2/y}$, and $f_{W_y} (x)$ is bounded. Also, $f'_{W_y} (x) = -x/y^{1.5} e^{-0.5 x^2/y}, f''_{W_y} (x) = -1/y^{1.5} e^{-0.5 x^2/y} + x^2/y^{2.5} e^{-0.5 x^2/y} $ are still bounded and continuous. So
\begin{align}
& f_{W_Y}(x) = \frac{dP(W_Y \le x)}{dx} = \lim_{\Delta x \rightarrow 0} \frac{ \mathbb{E}_Y \mathbb{E} \left[ {\text{\large $\mathds{1}$}}_{x< W_{Y} \le x+\Delta x} \right]  }{\Delta x} =  \lim_{\Delta x \rightarrow 0} \frac{ \mathbb{E}_Y f_{W_Y} (x) \Delta x   }{\Delta x} =  \mathbb{E}_Y f_{W_Y} (x), \\
& f'_{W_Y}(x)  =   \lim_{\Delta x \rightarrow 0} \frac{ \mathbb{E}_Y [ f_{W_Y} (x+\Delta x) - f_{W_Y} (x) ] } {\Delta x} =  \mathbb{E}_Y f'_{W_Y} (x), \  \text{similarly}, \  f''_{W_Y}(x) = \mathbb{E}_Y f''_{W_Y} (x). 
\end{align}
Thus, both $f'_{W_Y}(x)$ and $f_{W_Y}(x)$ are bounded and continuous, and $f''_{W_Y}(x)$ is bounded. 

Then we show Lemma~\ref{lemma-prerequisite}(b). It suffices to show that $c_n(w) \cap c_n^c(w+\Delta w)\rightarrow \phi$ and $c_n(w+\Delta w) \cap c_n^c(w)\rightarrow \phi$ almost surely. Due to symmetry, without loss of generality, we will assume $\Delta w \ge 0$ and show that $c_n(w) \cap c_n^c(w+\Delta w)\rightarrow \phi$. Note that ${Y_1}+\cdots + {Y_k} \ge \epsilon$ and is finite as well, so by Lemma~\ref{lemma-prerequisite}(a), the pdf of $W_{Y_1+\cdots+Y_k} (= W_{Y_1}+\cdots + W_{Y_k})$ is bounded. So we have
\begin{align}
\nonumber & \{ c_n(w) \cap c_n^c(w+\Delta w) \}  \subseteq \bigcup\limits_{k=1}^n \left\{  |w+W_{Y_1}+\cdots + W_{Y_k}| \ge v \cap  |w+\Delta w+W_{Y_1}+\cdots + W_{Y_k}| < v \right\} \\
& = \bigcup_{k=1}^n \left( W_{Y_1}+\cdots + W_{Y_k} \in (-v - w - \Delta w,-v-w] \right) \rightarrow \phi, \ \text{almost surely}.
\end{align} 

We finally show Lemma~\ref{lemma-prerequisite}(c) by induction. Note that the initial condition holds because $f_{W_Y} (x - w) \ge f_{W_Y} (-x-w)$. Suppose that the hypothesis holds. Then, by \eqref{eq-jointpdf}, for any $x,w \ge 0$,
\begin{align}
\nonumber & f_{w+W_{Y_1} + \cdots + W_{Y_{n}} + W_{Y_{n+1}} } (x) = \int_{|a| \ge a_{k}} f_{w+W_{Y_1} + \cdots + W_{Y_{n}}  } (x) \times f_{W_{Y_{n+1}}}(x - a) da, \\ \nonumber &
 f_{w+W_{Y_1} + \cdots + W_{Y_{n}} + W_{Y_{n+1}} } (x) - f_{w+W_{Y_1} + \cdots + W_{Y_{n}} + W_{Y_{n+1}} } (-x) \\ \nonumber =&  \int_{|a| \ge a_{k}} f_{w+W_{Y_1} + \cdots + W_{Y_{n}}  } (a) \times ( f_{W_{Y_{n+1}}}(x - a) - f_{W_{Y_{n+1}}}(-x - a)) da \\ 
 = & \int_{a \ge a_{k}} ( f_{w+W_{Y_1} + \cdots + W_{Y_{n}}  } (a) -  f_{w+W_{Y_1} + \cdots + W_{Y_{n}}  } (-a) ) \times ( f_{W_{Y_{n+1}}}(x - a) - f_{W_{Y_{n+1}}}(-x - a)) da \ge 0, 
\end{align} which ends the proof of the first claim. The last equation holds because
\begin{align}
\nonumber & \int_{a \le - a_{k}} f_{w+W_{Y_1} + \cdots + W_{Y_{n}}  } (a) \times ( f_{W_{Y_{n+1}}}(x - a) - f_{W_{Y_{n+1}}}(-x - a)) da \\ \nonumber
 = &  \int_{a \ge a_{k}} f_{w+W_{Y_1} + \cdots + W_{Y_{n}}  } (-a) \times ( f_{W_{Y_{n+1}}}(x + a) - f_{W_{Y_{n+1}}}(-x + a)) da \\
 = &  \int_{a \ge a_{k}}  f_{w+W_{Y_1} + \cdots + W_{Y_{n}}  } (-a) \times ( f_{W_{Y_{n+1}}}(-x - a) - f_{W_{Y_{n+1}}}(x - a)) da.
\end{align}

To show that $\frac{d}{dw}f_{w+W_{Y_1} + \cdots W_{Y_k}}(x)$ is continuous and bounded in $w$, note that      
\begin{align}
\nonumber & f_{w+W_{Y_1} + \cdots W_{Y_k}}( x ) = \int_{s_1: |x - s_1| \ge v_1} f_{W_{Y_k}} (s_1) \int_{s_2: |x - s_1 - s_2 |\ge a_2} f_{W_{Y_{k-1}}}(s_2) \\ & \cdots \int_{s_{k-1}: |x - s_1 - \cdots - s_{k-1}| \ge v_{k-1}} f_{W_{Y_2}} (s_{k-1}) f_{W_{Y_1}} (x - w - s_1 - \cdots - s_{k-1})   d s_{k-1} \cdots ds_2 ds_1.
\end{align}
In the above expression, only the final term $f_{W_{Y_1}} (x - w - s_1 - \cdots - s_{k-1})$ is related to $w$ and this term is continuously differentiable in $w$. Also, $f'_{W_Y}(w)$ is bounded, and the above expression is bounded. Therefore,
\begin{align}
\nonumber & \frac{d}{dw}f_{w+W_{Y_1} + \cdots W_{Y_k}}( x ) = - \int_{s_1: |x - s_1| \ge a_1} f_{W_{Y_k}} (s_1) \int_{s_2: |x - s_1 - s_2 |\ge a_2} f_{Y_{k-1}}(s_2) \\ & \cdots \int_{s_{k-1}: |x - s_1 - \cdots - s_{k-1}| \ge a_{k-1}} f_{W_{Y_2}} (s_{k-1}) f'_{W_{Y_1}} (x - w - s_1 - \cdots - s_{k-1})   d s_{k-1} \cdots ds_2 ds_1.
\end{align} Since $f''_{W_Y}(w)$ is bounded, $\frac{d}{dw}f_{w+W_{Y_1} + \cdots W_{Y_k}}( x)$ is bounded and continuous in $w$.
This ends the proof of Lemma~\ref{lemma-prerequisite}.
\end{proof}
For the property of $J'_{n}(w)$, we need the following lemma: 

\begin{lemma}\label{lemma-jw}
(a) $J'_n(w)$ is continuous for all $w$.

(b) The functions $J_{n}(w), J'_{n}(w)$ are upper bounded by some functions $\bar{J}_n (w), \bar{J'}_n(w)$, respectively, such that for any given $w$, we have $\mathbb{E}[\bar{J}_{n}(w+W_Y)]<\infty$ and $\mathbb{E}[\bar{J'}_{n}(w+W_Y)]<\infty$.

(c) If $|w| < v_{n+1}$ $J'_{n+1}(w) = \partial_x g(w,v_{n+1})$. If $|w| > v_{n+1}$, we have\footnote{Note that the event $|w+W_{Y_1}+\ldots W_{Y_n}|=a$ has zero probability for all index $n$ and real value $a$.} 
\begin{align}
\nonumber J'_{n+1}(w) 
 = & \partial_x g(w,v_{n+1}) + \alpha \mathbb{E} \left[  J'_n(w+W_{Y_1})   \right] \\ 
\nonumber   = & \partial_x g(w,v_{n+1})  + \alpha \mathbb{E} \left[ \partial_x g(w+W_{Y_1},v_{n})  \right]   + \alpha^2 \mathbb{E} \left[ J'(w+W_{Y_1}+W_{Y_2}) {\text{\large $\mathds{1}$}}_{|w+W_{Y_1}| \ge v_n }   \right]  \\
\nonumber & \cdots \\
 = &  \partial_x g(w,v_{n+1})  +
\sum_{k=1}^{n} \alpha^k \mathbb{E} \left[ \partial_x g(w+W_{Y_1+\cdots+Y_k},v_{n+1-k}) {\text{\large $\mathds{1}$}}_{c_{n,k}(w) }   \right]  \label{eq-eachterm}
\end{align} 
the event $c_{n,1}(w)$ is the whole set, and the events $c_{n,k}(w)$ for $k=2,\cdots,n$ are defined as
\begin{align}
c_{n,k}(w) = \{  |w+W_{Y_1}| \ge v_{n}, \cdots, |w+W_{Y_1}+\cdots+W_{Y_{k-1}}| \ge v_{n+2-k} \}, \ n=1,2,3,\cdots.
\end{align}
\end{lemma}
\begin{proof}
Note that 
\begin{align}
\nonumber & g(w,v_{n+1}) \\ 
\nonumber & = \left\{ 
\begin{array}{lll}
 \frac{1}{2} \mathbb{E} \left[ Y^2 \right]+\mathbb{E} \left[ Y \right]w^2-\mathbb{E} \left[ Y \right] \text{mse}_{\text{opt}} & |w| \ge v_{n+1}, \\
  \frac{1}{6}(v^4 - w^4) +  \frac{1}{2} \mathbb{E} \left[ Y^2 \right]+\mathbb{E} \left[ Y \right]w^2 - \mathbb{E} \left[ Y \right]\text{mse}_{\text{opt}} - (\text{mse}_{\text{opt}} - \mathbb{E} \left[ Y \right])(v^2 - w^2) & |w| < v_{n+1},
\end{array}\right. \\
& \partial_x g(w,v_{n+1}) = \left\{ 
\begin{array}{lll}
2 \mathbb{E} \left[ Y \right] w & |w| > v_{n+1}, \\
-2/3  w^3 +2 \text{mse}_{\text{opt}} w & |w| < v_{n+1}.
\end{array}\right.
\end{align}
When $n=1$, the free boundary method \eqref{eq-fbm1simple}---\eqref{eq-fbm3simple} implies that $v_1$ is the positive root of $-2/3  w^3 +2 \text{mse}_{\text{opt}} w=2 \mathbb{E} \left[ Y \right] w$, which is $\sqrt{3(\text{mse}_{\text{opt}} - \mathbb{E} \left[ Y \right])}$. Then, $J_1(w) = g(w,v_1)$. By \eqref{eq-fbm3simple}, $J_1'(w)$ is continuous at $w= \pm v_1$, thus continuous at $w\in \mathbb{R}$. For any given $w$, $J_1(w)$ is bounded by $ \mathbb{E} \left[ Y \right] w^2$ plus a constant, and $J_1'(w)$ is bounded by $2 \mathbb{E} \left[ Y \right] |w|$ plus a constant. By this statement and $\mathbb{E}[W^2_Y]<\infty$, condition~(b) holds. Condition (c) trivially holds because we have already set $J_0(w)=0$.

Now we suppose that the hypothesis holds at $n$. We will show condition~(a)---(c) for the case $n+1$. 
Since function $J_{n}$ is even, and $W_{Y}$ has a symmetric pdf, we have 
\begin{align}
J_{n+1}(w) = \left\{ \begin{array}{lll}
 g(w,v_{n+1}) + \alpha \mathbb{E} \left[  J_n(w+W_{Y}) \right] & |w| \ge v_{n+1},\\
 g(w,v_{n+1}) + \alpha \mathbb{E} \left[  J_n(v_{n+1}+W_{Y}) \right] & |w| < v_{n+1}.
 \end{array}\right.
\end{align}
Utilizing the hypothesis that $J_n(w)$ is continuous, condition (b), $g(w,v_{n+1})$ is continuous, we have that $J_{n+1}(w)$ is continuous. 
When $|w|<v_{n+1}$, it is easy to find that 
\begin{align}
J'_{n+1}(w) = \partial_x g(w,v_{n+1}) = -2/3  w^3 +2 \text{mse}_{\text{opt}} w.\label{eq-prooflemma4a2}
\end{align}
Further, when $|w| > v_{n+1}$, by the definition in \eqref{eq-root-qn2}, $J_{n+1}(w) = G_{n+1}(w)$ and
\begin{align}
\nonumber J'_{n+1}(w) = G'_{n+1}(w) & \triangleq  \partial_x g(w,v_{n+1}) + \alpha \frac{d}{d w}  \mathbb{E} \left[  J_n(w+W_{Y}) \right]  \\
\nonumber & =  \partial_x g(w,v_{n+1}) + \alpha \lim_{\Delta w \rightarrow 0} \frac{1}{\Delta w}  \mathbb{E} \left[  J_n(w+W_{Y}+\Delta w)  - J_n(w+W_{Y}) \right]  \\
\nonumber & =  \partial_x g(w,v_{n+1}) + \alpha \lim_{\Delta w \rightarrow 0} \mathbb{E} \left[  J'_n(w+W_Y+\epsilon) \right] \\
& =  \partial_x g(w,v_{n+1}) + \alpha  \mathbb{E} \left[  J'_n(w+W_Y) \right]. \label{eq-prooflemma4a1}
\end{align}
Here, $\epsilon$ is a number that is between $0$ and $\Delta w$. The third equation holds because $J'_n(w)$ is well-defined. The last equation holds due to dominated convergence theorem and the hypothesis conditions~(a),(b). When $|w|<v_{n+1}$, we have $J'_{n+1}(w)=  \partial_x g(w,v_{n+1}) $. Thus, we directly get condition~(c). By the free boundary method \eqref{eq-fbm3simple}, $J'_{n+1}(w)$ is continuous at $|w|=v_{n+1}$, thus condition~(a) holds. In addition, note that
\begin{align}
\nonumber  J'_{n+1}(w) & \le 
| \partial_x g(w,v_{n+1}) | +
\sum_{k=1}^{n} \alpha^k \mathbb{E} \left[ | \partial_x g(w+W_{Y_1+\cdots+Y_k},v_{n+1-k}) | \right] \\
\nonumber & \le  2 \mathbb{E} \left[ Y \right] w +
\sum_{k=1}^{n} \alpha^k \mathbb{E} \left[  2 \mathbb{E} \left[ Y \right] (w+|W_{Y_1}|+\cdots+|W_{Y_k}|)    \right] +b_1,\\
\nonumber  J_{n+1}(w) & \le 
   \mathbb{E} \left[ Y \right] w^2 +
\sum_{k=1}^{n} \alpha^k \mathbb{E} \left[   \mathbb{E} \left[ Y \right] \left(w+|W_{Y_1}|+\cdots+|W_{Y_k}| + \Sigma_{i=1}^{k}v_{n+1-i} \right)^2    \right] +b_2,
\end{align} where $b_1,b_2$ are bounded values irrelevant to $w$. Thus, combined with $\mathbb{E}[W^2_Y]<\infty$, condition~(b) holds.
This ends the proof of lemma~\ref{lemma-jw}.
\end{proof}
Lemma~\ref{lemma-jw} implies that $G'_{n+1}(w),J'_{n+1}(w)$ are well-defined. Also, \eqref{eq-prooflemma4a1} implies that we can interchange the derivative and expectation of $J_n(w+W_Y)$, i.e.,
\begin{align}
 \frac{d}{dw} \mathbb{E} \left[ J_n(w+W_Y) \right] =  \mathbb{E} \left[ J'_n(w+W_Y) \right]. \label{eq-prooflemma4a}
\end{align}

\subsection{Proof of Lemma~\ref{derivative1}(a)}\label{app-lemma4a}

Equations \eqref{eq-stage1wsmallv} and \eqref{eq-stage1vsmallw} in Lemma~\ref{derivative1}(a) are easily shown by \eqref{eq-prooflemma4a2} and \eqref{eq-prooflemma4a1} in the proof of Lemma~\ref{lemma-jw}. 

According to \eqref{eq-prooflemma4a1}, to show that $G''_{n+1}(w),G'''_{n+1}(w)$ are continuous, it is equivalent to show that $ \frac{d^2}{dw^2} \mathbb{E} \left[   J_n(w+W_Y) \right]$ and $\frac{d^3}{dw^3} \mathbb{E} \left[   J_n(w+W_Y) \right]$ are continuous. To show this, we should analyze the derivative of each term $ \mathbb{E} \left[ \partial_x g(w+W_{Y_1+\cdots +Y_j},v_{n+1-k}) {\text{\large $\mathds{1}$}}_{ c_{n,k}(w)  } \right]$ in \eqref{eq-eachterm}. 
We look at any odd polynomial function $f (w+W_{Y_1} + \cdots W_{Y_n})$ with $f(w) = O(w^3)$. 

We are interested in analyzing the derivative
\begin{align}
\lim_{\Delta w \rightarrow 0} \frac{1}{\Delta w}  \mathbb{E} \left[    f(w+\Delta w+W_{Y_1} + \cdots W_{Y_n})  {\text{\large $\mathds{1}$}}_{ c_{n-1}(w+\Delta w)} -  f(w+W_{Y_1} + \cdots W_{Y_n})  {\text{\large $\mathds{1}$}}_{ c_{n-1}(w)} \right]. \label{eq-100}
\end{align}
For simplicity, we utilize the event $c_{n-1}(w)$ from Definition~\ref{def-5}.
We partition the whole set to $n$ sets:
\begin{align}
\nonumber & c^1(w) \triangleq |w+W_{Y_1}| < a_1,\\
\nonumber & c^k(w) \triangleq |w+W_{Y_1}| \ge a_1, \ldots, |w+W_{Y_1}+\ldots+W_{Y_{k-1}}| \ge a_{k-1},  |w+W_{Y_1}+\ldots+W_{Y_{k}}| < a_{k}, k=2,3,\ldots.
\end{align}
First, Lemma~\ref{lemma-prerequisite}(b), Lemma~\ref{lemma-jw} and dominated convergence theorem give\footnote{we use $ \Delta f (w+W_{Y_1} + \cdots W_{Y_n}) $ to replace $  f(w+\Delta w+W_{Y_1} + \cdots W_{Y_n})  {\text{\large $\mathds{1}$}}_{ c_{n-1}(w+\Delta w)} -  f(w+W_{Y_1} + \cdots W_{Y_n})  {\text{\large $\mathds{1}$}}_{ c_{n-1}(w)}$ for simplicity. $a_1,\ldots,a_k$ are arbitrary finite numbers.}  
\begin{align}
&\lim_{\Delta w \rightarrow 0}\frac{1}{\Delta w} \mathbb{E} \left[   \Delta f (w+W_{Y_1} + \cdots W_{Y_n})  {\text{\large $\mathds{1}$}}_{ c_{n-1}(w),c_{n-1}(w+\Delta w)} \right] =   \mathbb{E} \left[    f' (w+W_{Y_1} + \cdots W_{Y_n})  {\text{\large $\mathds{1}$}}_{ c_{n-1}(w)} \right], \label{eq-der2forgmain} \\
& \lim_{\Delta w \rightarrow 0}\frac{1}{\Delta w} \mathbb{E} \left[   \Delta f (w+W_{Y_1} + \cdots W_{Y_n})  {\text{\large $\mathds{1}$}}_{ c^{k1}(w),c^{k2}(w+\Delta w)} \right] = 0, \ k1,k2 \in \{1,2,\cdots,n-1\}. \label{eq-der2forgmain2}
\end{align}
Similarly, for any $k\in \{1,2,\ldots,n-1\}$, 
\begin{align}
\nonumber  & \lim_{\Delta w \rightarrow 0} \frac{1}{\Delta w} \mathbb{E} \left[   \Delta f (w+W_{Y_1} + \cdots W_{Y_n})  {\text{\large $\mathds{1}$}}_{ c_{n-1}(w),c^k(w+\Delta w)} \right] \\
\nonumber  = &  \lim_{\Delta w \rightarrow 0}  \frac{1}{\Delta w} \mathbb{E} \left[ - f (w+W_{Y_1} + \cdots W_{Y_n})  {\text{\large $\mathds{1}$}}_{ c_{n-1}(w),c_{k-1}(w+\Delta w), |w+\Delta w +W_{Y_1}+\cdots+W_{Y_k}| <a_k}  \right] \\
\nonumber &  \Delta w \rightarrow 0+: \\ \nonumber = &  \lim_{\Delta w \rightarrow 0}  \frac{1}{\Delta w}\mathbb{E} \Big[  - f (-a_k+W_{Y_{k+1}} + \cdots W_{Y_n}) {\text{\large $\mathds{1}$}}_{ |-a_k+W_{Y_{k+1}}| \ge a_{k+1}, \cdots,|-a_k+W_{Y_{k+1}} + \cdots W_{Y_{n-1}} |\ge a_{n-1} } \\ \nonumber & \times  {\text{\large $\mathds{1}$}}_{c_{k-1}(w), c_{k-1}(w+\Delta w), -a_k - \Delta w < w +W_{Y_1}+\cdots+W_{Y_k} < -a_k  } \Big]  \\
\nonumber = &  \mathbb{E} \left[  - f (-a_k+W_{Y_{k+1}} + \cdots W_{Y_{n}}) {\text{\large $\mathds{1}$}}_{ |-a_k+W_{Y_{k+1}}| \ge a_{k+1}, \cdots,|-a_k+W_{Y_{k+1}} + \cdots W_{Y_{n-1}} |\ge a_{n-1} } \right] \\  & \times  \lim_{\Delta w \rightarrow 0} \frac{1}{\Delta w} \mathbb{E} \left[  {\text{\large $\mathds{1}$}}_{c_{k-1}(w), -a_k - \Delta w < w +W_{Y_1}+\cdots+W_{Y_k} < -a_k  } \right]; \\
 \nonumber &   \Delta w \rightarrow 0-: \\ = & \mathbb{E} \left[   - f (a_k+W_{Y_{k+1}} + \cdots W_{Y_{n}}) {\text{\large $\mathds{1}$}}_{ \cdots,|a_k+W_{Y_{k+1}} + \cdots W_{Y_{n-1}} |\ge a_{n-1} }  \right]   \lim_{\Delta w \rightarrow 0}  \frac{1}{\Delta w}  \mathbb{E} \left[  {\text{\large $\mathds{1}$}}_{ c_{k-1}(w), a_k  < w +W_{Y_1}+\cdots+W_{Y_k} < a_k - \Delta w } \right].   \end{align} 
 In other case, 
 \begin{align} 
 \nonumber  & \lim_{\Delta w \rightarrow 0} \frac{1}{\Delta w} \mathbb{E} \left[   \Delta f (w+W_{Y_1} + \cdots W_{Y_n})  {\text{\large $\mathds{1}$}}_{ c_{n-1}(w+\Delta w),c^k(w)} \right] \\
\nonumber  = & \lim_{\Delta w \rightarrow 0}  \frac{1}{\Delta w} \mathbb{E} \left[  f (w+W_{Y_1} + \cdots W_{Y_n})  {\text{\large $\mathds{1}$}}_{ c_{n-1}(w+\Delta w), c^{k-1}(w), |w +W_{Y_1}+\cdots+W_{Y_k}| <a_k}  \right] \\
\nonumber & \Delta w \rightarrow 0+: \\ = &  \mathbb{E} \left[  f (a_k+W_{Y_{k+1}} + \cdots W_{Y_n}) {\text{\large $\mathds{1}$}}_{ \cdots,|a_k+W_{Y_{k+1}} + \cdots W_{Y_{n-1}} |\ge a_{n-1} } \right]  \lim_{\Delta w \rightarrow 0}  \frac{1}{\Delta w} \mathbb{E} \left[  {\text{\large $\mathds{1}$}}_{ c_{k-1}(w), a_k - \Delta w < w +W_{Y_1}+\cdots+W_{Y_k} < a_k  } \right]; \\
\nonumber &  \Delta w \rightarrow 0-: \\ \nonumber = & \mathbb{E} \left[  f (-a_k+W_{Y_{k+1}} + \cdots W_{Y_n}) {\text{\large $\mathds{1}$}}_{  |-a_k+W_{Y_{k+1}}| \ge a_{k+1}, \cdots,|-a_k+W_{Y_{k+1}} + \cdots W_{Y_{n-1}} |\ge a_{n-1} } \right] \\ & \times  \lim_{\Delta w \rightarrow 0}  \frac{1}{\Delta w} \mathbb{E} \left[  {\text{\large $\mathds{1}$}}_{c_{k-1}(w), -a_k  < w +W_{Y_1}+\cdots+W_{Y_k} < -a_k - \Delta w }  \right].
\end{align}
Therefore,
\begin{align}
\nonumber & \lim_{\Delta w \rightarrow 0}  \frac{1}{\Delta w} \mathbb{E} \left[   \Delta f (w+W_{Y_1} + \cdots W_{Y_n})  {\text{\large $\mathds{1}$}}_{ c_{n-1}(w+\Delta w),c^k(w)} +  {\text{\large $\mathds{1}$}}_{ c_{n-1}(w),c^k (w+\Delta w)} \right] \\
\nonumber = &  f_{w+W_{Y_1} + \cdots W_{Y_k}}(a_k )  \mathbb{E} \left[  f (a_k+W_{Y_{k+1}} + \cdots W_{Y_n}) {\text{\large $\mathds{1}$}}_{|a_k+W_{Y_{k+1}}  |\ge a_{k+1}, \cdots,|a_k+W_{Y_{k+1}} + \cdots W_{Y_{n-1}} |\ge a_{n-1} }  \right] \\
\nonumber & -  f_{w+W_{Y_1} + \cdots W_{Y_k}}(-a_k ) \mathbb{E} \left[ f (-a_k+W_{Y_{k+1}} + \cdots W_{Y_n})  {\text{\large $\mathds{1}$}}_{ |-a_k+W_{Y_{k+1}}  |\ge a_{k+1}, \cdots,|-a_k+W_{Y_{k+1}} + \cdots W_{Y_{n-1}} |\ge a_{n-1} } \right] \\
\nonumber = & (  f_{w+W_{Y_1} + \cdots W_{Y_k}}(a_k ) + f_{w+W_{Y_1} + \cdots W_{Y_k}}(-a_k ) )  \\ & \times \mathbb{E} \left[    f (a_k+W_{Y_{k+1}} + \cdots W_{Y_n}) {\text{\large $\mathds{1}$}}_{ |a_k+W_{Y_{k+1}}  |\ge a_{k+1}, \cdots,|a_k+W_{Y_{k+1}} + \cdots W_{Y_{n-1}} |\ge a_{n-1} } \right], \label{eq-der2forg}
\end{align}
which is a constant term multiplied by $ f_{w+W_{Y_1} + \cdots W_{Y_k}}(a_k ) + f_{w+W_{Y_1} + \cdots W_{Y_k}}(-a_k )$, where $ f_{w+W_{Y_1} + \cdots W_{Y_k}}(a_k)$ is defined in Definition~\ref{def-5}. The last equation in \eqref{eq-der2forg} holds because $W_Y$ is symmetric and $f(\cdot)$ is odd\footnote{Note that if $f(\cdot)$ is even, then the first term of the last equation in \eqref{eq-der2forg} becomes $ f_{w+W_{Y_1} + \cdots W_{Y_k}}(a_k ) - f_{w+W_{Y_1} + \cdots W_{Y_k}}(-a_k )$.}.
 Recall that $\partial_x g(w,a)$ contains two odd polynomial terms $2 \mathbb{E} \left[ Y \right] w  $ and $2 \text{mse}_{\text{opt}} w - 2/3 w^3 $ related to $w$. Therefore,
the derivative of $ \mathbb{E} \left[ \partial_x g(w+W_{Y_1+\cdots +Y_k},v_{n+1-k}) {\text{\large $\mathds{1}$}}_{ c_{n,k}(w)  } \right]$ that appears in \eqref{eq-eachterm} is expressed as the sum of forms \eqref{eq-der2forgmain}, \eqref{eq-der2forgmain2}, and \eqref{eq-der2forg}. The value of \eqref{eq-der2forgmain2} is $0$. According to Lemma~\ref{lemma-prerequisite}(c), the value of \eqref{eq-der2forg} is a constant multiplied by $f_{w+W_{Y_1} + \cdots W_{Y_k}}(x )$, a continuously differentiable function in $w$ for some parameter $x$. 
For \eqref{eq-der2forgmain}, note that the term of \eqref{eq-der2forgmain} is continuous in $w$. We can take the derivative and apply the previous calculations \eqref{eq-100}---\eqref{eq-der2forg} again\footnote{Despite that $f'(w)$ becomes an even polynomial function with $O(w^2)$, except the minor sign change of the last equality of \eqref{eq-der2forg} as described in the previous footnote, the calculations \eqref{eq-100}---\eqref{eq-der2forg} remain the same.}. Then, the term in \eqref{eq-der2forgmain} is still continuously differentiable. This shows that $ \mathbb{E} \left[ \partial_x g(w+W_{Y_1+\cdots +Y_k},v_{n+1-k}) {\text{\large $\mathds{1}$}}_{ c_{n,k}(w)  } \right]$ is continuously differentiable. Thus, $ \frac{d^2}{dw^2} \mathbb{E} \left[   J_n(w+W_Y) \right],  \frac{d^3}{dw^3} \mathbb{E} \left[   J_n(w+W_Y) \right]$ are both continuous. This ends the proof of Lemma~\ref{derivative1}(a).

\subsection{Proof of Lemma~\ref{derivative1}(b),(c)}

We then use induction to prove Lemma~\ref{derivative1}(b),(c).
Let us denote $\beta=  \text{mse}_{\text{opt}} -  \mathbb{E} \left[ Y \right]$, Note that the free boundary method implies that $v_1=\sqrt{3\beta}$, and $g(w,v_1)$ is continuously differentiable. In addition, 
\begin{align}
\partial_{xx}g(w,v_1) & = \left\{ \begin{array}{lll}
2 \mathbb{E} \left[ Y \right] & \ w > v_1, \vspace{1mm} \\
-2 w^2 + 2 \text{mse}_{\text{opt}}  & \ 0 \le w<v_1.
\end{array}
\right. \\
\partial_{xxx}g(w,v_1) & = \left\{ \begin{array}{lll}
0 & \ w > v_1, \vspace{1mm} \\
-4 w & \ 0 \le w<v_1,
\end{array}
\right.
\end{align}
Then, we have $G''_1(w) = 2 \mathbb{E} \left[ Y \right] $ and for all $w\ge v_1$, 
\begin{align}
G''_1(w) - (-2 w^2 + 2 \text{mse}_{\text{opt}}) = -2 \beta + 2w^2 \ge -2\beta + 6\beta \ge 0.
\end{align}
For all $w \ge 0$, 
\begin{align}
G'''_1(w) - (-4w) = 0+ 4w  \ge 0.
\end{align}
This satisfies the initial condition of Lemma~\ref{derivative1}(b),(c). 
By Lemma~\ref{lemma-jw},
\begin{align}
\nonumber  G''_{n+1}(w) & = 2 \mathbb{E} \left[ Y \right]+  \frac{d^2}{dw^2} \mathbb{E} \left[ J_n (w+W_{Y}) \right] \\
\nonumber & =   \frac{d}{dw} \mathbb{E} \left[  J'_n (w+W_{Y}) \right] \\
 & =  \lim_{\Delta w \rightarrow 0} \frac{1}{\Delta w} \mathbb{E} \left[   J'_n (w+\Delta w+W_{Y}) -  J'_n (w+W_{Y}) \right]. 
\end{align}
Since $J'_n(\cdot)$ is continuous, there exists $0\le \epsilon\le \Delta w$ or $\Delta w \le \epsilon\le 0$, such that
\begin{align}
 \nonumber & \frac{1}{\Delta w} \left(  J'_n (w+\Delta w+W_{Y}) -   J'_n (w+W_{Y}) \right) {\text{\large $\mathds{1}$}}_{|w+\Delta w + W_{Y}|>v_n,|w+W_{Y}|>v_n} \\  \nonumber = &   J''_n (w+\epsilon+W_{Y}) {\text{\large $\mathds{1}$}}_{|w+\Delta w  +W_{Y}|>v_n,|w+W_{Y}|>v_n}, \\
 \nonumber  & \frac{1}{\Delta w}  \left(  J'_n (w+\Delta w+W_{Y}) -  J'_n (w+W_{Y}) \right) {\text{\large $\mathds{1}$}}_{|w+\Delta w + W_{Y}|<v_n,|w+W_{Y}|<v_n} \\ = &  J''_n(w+\epsilon+W_{Y})  {\text{\large $\mathds{1}$}}_{|w+\Delta w + W_{Y}|<v_n,|w+W_{Y}|<v_n}
\end{align}
We have shown that $  \frac{d^2}{dw^2} \mathbb{E} \left[ J_n(w+W_{Y})  \right] $ is continuous, so $ J''_{n+1}(w) $ is continuous for $|w|\ne v_{n+1}$. Applying the same analysis for $ \frac{d^2}{dw^2} \mathbb{E}_{W_{Y_2}} \left[ J_n(w+W_{Y}+ W_{Y_2})  \right] $ (to replace $\frac{d^2}{dw^2} \mathbb{E} \left[ J_n(w+W_{Y})  \right] $) into the proof of Lemma~\ref{derivative1}(a) described in Appendix~\ref{app-lemma4a}, both $ G''_n(w+W_Y+\Delta w)$ and $J''_n(w+W_Y+\Delta w)$ are bounded by a finite random variable. Therefore, using dominated convergence theorem,

\begin{align}
 \nonumber  &  \lim_{\Delta w \rightarrow 0}\mathbb{E} \Big[  \frac{1}{\Delta w} \left(  J'_n (w+\Delta w+W_{Y}) -  J'_n (w+W_{Y}) \right) \times \\ \nonumber & \left( {\text{\large $\mathds{1}$}}_{|w+\Delta w + W_{Y}|>v_n,|w+W_{Y}|>v_n} + {\text{\large $\mathds{1}$}}_{|w+\Delta w + W_{Y}|<v_n,|w+W_{Y}|<v_n} \right) \Big] \\
 = & \mathbb{E} \left[  J''_n (w+W_{Y}) \times \left( {\text{\large $\mathds{1}$}}_{|w+\Delta w  +W_{Y}|>v_n,|w+W_{Y}|>v_n} + {\text{\large $\mathds{1}$}}_{|w+\Delta w + W_{Y}|<v_n,|w+W_{Y}|<v_n} \right) \right]. \label{eq-der2a}
\end{align}
By Lemma~\ref{lemma-prerequisite}, the two remaining events vanishes as $\Delta w \rightarrow 0$. 
Thus, for small $\Delta w$, we have 
\begin{align}
\nonumber & \Big| \ \mathbb{E} \Big{[} \frac{1}{\Delta w} \left(   J'_n (w+\Delta w+W_{Y}) -  J'_n (w+W_{Y}) \right) \\  \nonumber &  \left( {\text{\large $\mathds{1}$}}_{|w+\Delta w + W_{Y}|>v_n,|w+W_{Y}| < v_n} +  {\text{\large $\mathds{1}$}}_{|w+\Delta w + W_{Y}|< v_n,|w+W_{Y}| > v_n} \right)   \Big{]}  \ \Big| \\
\le & \max_{|v_n-x|\le |\Delta w| , x\ne v_n} |J''(x)| \ \mathbb{E}
 \left[ {\text{\large $\mathds{1}$}}_{|w+\Delta w + W_{Y}|>v_n,|w+W_{Y}| < v_n} +  {\text{\large $\mathds{1}$}}_{|w+\Delta w + W_{Y}|< v_n,|w+W_{Y}| > v_n} \right] 
\rightarrow 0. \label{eq-der2b}
\end{align}
By \eqref{eq-der2a},\eqref{eq-der2b}, we have an interesting result: 
\begin{align}
 \frac{d^2}{dw^2}  \mathbb{E} \left[ J_n (w+W_{Y}) \right] & = \mathbb{E} \left[  J''_n (w+W_{Y}) \right], \\ 
  G''_{n+1}(w) & = 2 \mathbb{E} \left[ Y \right]+  \mathbb{E} \left[  J''_n (w+W_{Y}) \right]. 
\end{align}
Then, we consider the third derivative:
\begin{align}
\nonumber G_{n+1}'''(w) & = 0+  \frac{d^3}{dw^3} \mathbb{E} \left[ J_n (w+W_{Y}) \right] 
 =   \frac{d}{dw} \mathbb{E} \left[ J''_n (w+W_{Y}) \right] \\
 & =  \lim_{\Delta w \rightarrow 0} \frac{1}{\Delta w} \mathbb{E} \left[  J''_n (w+\Delta w+W_{Y}) -   J''_n (w+W_{Y}) \right]. 
\end{align}
For this derivation, there exists $0\le \epsilon\le \Delta w$ or $\Delta w \le \epsilon\le 0$, such that 
\begin{align}
\nonumber & \frac{1}{\Delta w} \left(   J''_n (w+\Delta w+W_{Y}) -   J''_n (w+W_{Y}) \right) {\text{\large $\mathds{1}$}}_{|w+\Delta w + W_{Y}|>v_n,|w+W_{Y}|>v_n} \\ \nonumber = &   G'''_n (w+\epsilon+W_{Y}) {\text{\large $\mathds{1}$}}_{|w+\Delta w  +W_{Y}|>v_n,|w+W_{Y}|>v_n}, \\
\nonumber & \frac{1}{\Delta w}  \left(  J''_n (w+\Delta w+W_{Y}) -   J''_n (w+W_{Y}) \right) {\text{\large $\mathds{1}$}}_{|w+\Delta w + W_{Y}|<v_n,|w+W_{Y}|<v_n} \\ = & -4(w+\epsilon+W_Y)  {\text{\large $\mathds{1}$}}_{|w+\Delta w + W_{Y}|<v_n,|w+W_{Y}|<v_n}.
\end{align}
Recall that $G'''_n(\cdot)$ is continuous.
Applying the same analysis for $ \frac{d^3}{dw^3} \mathbb{E}_{W_{Y_2}} \left[ J_n(w+W_{Y}+ W_{Y_2})  \right] $ (to replace $\frac{d^3}{dw^3} \mathbb{E} \left[ J_n(w+W_{Y})  \right] $) into the proof of Lemma~\ref{derivative1}(a) described in Appendix~\ref{app-lemma4a}, both $ G'''_n(w+W_Y+\Delta w)$ and $J'''_n(w+W_Y+\Delta w)$ are bounded by a finite random variable. Therefore, using dominated convergence theorem (similar to \eqref{eq-der2a}),
\begin{align}
\nonumber &  \lim_{\Delta w \rightarrow 0}\mathbb{E} \left[  \frac{1}{\Delta w} \left(   J''_n (w+\Delta w+W_{Y}) -  J''_n (w+W_{Y}) \right) {\text{\large $\mathds{1}$}}_{|w+\Delta w + W_{Y}|>v_n,|w+W_{Y}|>v_n} \right] \\
\nonumber = & \mathbb{E} \left[  \lim_{\Delta w \rightarrow 0} G'''_n (w+\epsilon+W_{Y}) \times {\text{\large $\mathds{1}$}}_{|w+\Delta w  +W_{Y}|>v_n,|w+W_{Y}|>v_n}  \right] \\
\nonumber = &  \mathbb{E} \left[  G'''_n (w+W_{Y}) \times {\text{\large $\mathds{1}$}}_{|w+W_{Y}|>v_n}  \right],
 \\ \nonumber
&  \lim_{\Delta w \rightarrow 0}\mathbb{E} \left[  \frac{1}{\Delta w} \left(  J''_n (w+\Delta w+W_{Y}) -  J''_n (w+W_{Y}) \right) {\text{\large $\mathds{1}$}}_{|w+\Delta w + W_{Y}|<v_n,|w+W_{Y}|<v_n} \right] \\ = & \mathbb{E} \left[  -4(w+W_Y) {\text{\large $\mathds{1}$}}_{|w+W_{Y}|<v_n} \right]. \label{eq-g3a}
\end{align}
We then discuss the two remaining events.
If $\Delta w > 0$,
\begin{align}
\nonumber &  \frac{1}{\Delta w} \left(  J''_n (w+\Delta w+W_{Y}) -   J''_n (w+W_{Y}) \right) {\text{\large $\mathds{1}$}}_{|w+\Delta w + W_{Y}|>v_n,|w+W_{Y}| < v_n} \\ = & \left(   G''_n (v_n) + 2v^2_{n} - 2 \text{mse}_{\text{opt}} + o(\Delta w) \right)  \frac{1}{\Delta w}  {\text{\large $\mathds{1}$}}_{ v_n - w - \Delta w < W_Y < v_n - w }, \\
\nonumber &  \frac{1}{\Delta w} \left(   J''_n (w+\Delta w+W_{Y}) -   J''_n (w+W_{Y}) \right) {\text{\large $\mathds{1}$}}_{|w+\Delta w + W_{Y}|<v_n,|w+W_{Y}| > v_n} \\ = & - \left(  G''_n (-v_n) + 2v^2_{n} - 2 \text{mse}_{\text{opt}} + o(\Delta w) \right)  \frac{1}{\Delta w}  {\text{\large $\mathds{1}$}}_{ -v_n - w- \Delta w < W_Y < -v_n - w },
\end{align}
If $\Delta w < 0$,
\begin{align}
\nonumber &  \frac{1}{\Delta w} \left( J''_n (w+\Delta w+W_{Y}) -  J''_n (w+W_{Y}) \right) {\text{\large $\mathds{1}$}}_{|w+\Delta w + W_{Y}|>v_n,|w+W_{Y}| < v_n} \\ = & \left(   G''_n (- v_n) + 2v^2_{n} - 2 \text{mse}_{\text{opt}} + o(\Delta w) \right)  \frac{1}{\Delta w}  {\text{\large $\mathds{1}$}}_{ -v_n - w  < W_Y < v_n - w - \Delta w }, \\
\nonumber &  \frac{1}{\Delta w} \left(  J''_n (w+\Delta w+W_{Y}) -  J''_n (w+W_{Y}) \right) {\text{\large $\mathds{1}$}}_{|w+\Delta w + W_{Y}|<v_n,|w+W_{Y}| > v_n} \\ = & - \left(   G''_n (v_n) + 2v^2_{n} - 2 \text{mse}_{\text{opt}} + o(\Delta w) \right)  \frac{1}{\Delta w}  {\text{\large $\mathds{1}$}}_{ v_n - w < W_Y < v_n - w - \Delta w }. 
\end{align}
Therefore,
\begin{align}
\nonumber &  \lim_{\Delta w \rightarrow 0} \mathbb{E} \Big[ \frac{1}{\Delta w} \left(   J''_n (w+\Delta w+W_{Y}) -  J''_n (w+W_{Y}) \right) \\ \nonumber & \times \left( {\text{\large $\mathds{1}$}}_{|w+\Delta w + W_{Y}|>v_n,|w+W_{Y}| < v_n} +  {\text{\large $\mathds{1}$}}_{|w+\Delta w + W_{Y}|< v_n,|w+W_{Y}| > v_n} \right)   \Big]  \\
 & = - \left( G''_n (v_n) + 2v^2_{n} - 2 \text{mse}_{\text{opt}} \right) (f_{W_Y}(-v_n - w)  - f_{W_Y}(v_n - w)) \ge 0. \label{eq-remain}
\end{align} The last inequality of \eqref{eq-remain} holds due to the induction hypothesis of $G''_n(\cdot)$ and Lemma~\ref{lemma-prerequisite}.
Combining \eqref{eq-g3a} and \eqref{eq-remain},
for all $w \ge 0$, 
\begin{align}
\nonumber  G'''_{n+1}(w) = & 0+\alpha \mathbb{E} \left[ J'''_n(w+W_{Y}) {\text{\large $\mathds{1}$}}_{|w+W_{Y}|>v_n} +   J'''_n(w+W_{Y}) {\text{\large $\mathds{1}$}}_{|w+W_{Y}|<v_n} \right] \\ 
\nonumber & - \left( G''_n (v_n) + 2v^2_{n} - 2 \text{mse}_{\text{opt}} \right) (f_{W_Y}(-v_n - w)  - f_{W_Y}(v_n - w))\\
\nonumber \ge & \alpha \mathbb{E} \left[ J'''_n(w+W_{Y}) {\text{\large $\mathds{1}$}}_{|w+W_{Y}|>v_n} +   J'''_n(w+W_{Y}) {\text{\large $\mathds{1}$}}_{|w+W_{Y}|<v_n} \right] \\
\nonumber = & \alpha \mathbb{E} \left[   G'''_n(w+W_{Y}) {\text{\large $\mathds{1}$}}_{|w+W_{Y}|>v_n} -4(w+W_{Y}) {\text{\large $\mathds{1}$}}_{|w+W_{Y}|<v_n}  \right] \\
 = & \alpha \mathbb{E} \left[  \left( G'''_n(w+W_{Y})+4(w+W_Y) \right) {\text{\large $\mathds{1}$}}_{|w+W_{Y}|>v_n} -4(w+W_{Y}) {\text{\large $\mathds{1}$}}_{|w+W_{Y}|\ne v_n}  \right]. \label{eq-g3ninequality}
 \end{align}
 Note that $G'''_n(w)+4w$ is an odd function, and by hypothesis, $G'''_n(w)+4w\ge0$ for all $w \ge 0$. By Lemma~\ref{lemma-prerequisite}, $f_{w+W_Y}(x) \ge f_{w+W_Y}(-x)$ for all $w\ge 0$ and $x\ge 0$. Therefore,
 \begin{align}
\nonumber &  \mathbb{E} \left[  \left( G'''_n(w+W_{Y})+4(w+W_Y) \right) {\text{\large $\mathds{1}$}}_{|w+W_{Y}|>v_n} \right] \\ \nonumber  = & ( \int_{x>v_n} + \int_{x<- v_n} ) (G'''_n(x)+4x) f_{w+W_Y}(x) dx \\ \nonumber = & \int_{x>v_n} (G'''_n(x)+4x) f_{w+W_Y}(x) dx + \int_{x>v_n} (G'''_n(-x)-4x) f_{w+W_Y}(-x) dx \\ = &  \int_{x>v_n} (G'''_n(x)+4x) (f_{w+W_Y}(x) - f_{w+W_Y}(-x)) dx \ge 0. \label{eq-g3ninequality2}
 \end{align}
 Finally, \eqref{eq-g3ninequality} and \eqref{eq-g3ninequality2} give
 \begin{align}
\nonumber G'''_n(w) \ge & \alpha  \mathbb{E} \left[ -4(w+W_{Y}) {\text{\large $\mathds{1}$}}_{|w+W_{Y}| \ne v_n} \right] \\
\nonumber  = & \alpha  \mathbb{E} \left[ -4(w+W_{Y})  \right] \\
\nonumber = & -4\alpha w\\
\ge & -4w.
\end{align} The last inequality is strict if $w>0$. This ends the proof of Lemma~\ref{derivative1}(c).

Let us define $ f_{n+1}(w) = G'_{n+1}(w) + 2/3w^3-\text{mse}_{\text{opt}} w$ for simplicity. Then, $f_{n+1}''(w) = G'''_{n+1}(w)+4w\ge 0$ for $w\in[0,\infty)$. This implies that $f_{n+1}(w)$ is convex in $w\in[0,\infty)$ and strictly convex in $w>0$. By Lemma~\ref{derivative1}(a), $G'_{n+1}(w)$ is continuous and odd. Thus, $f_{n+1}(0)=0$. 
By the definition of free boundary method \eqref{eq-fbm3simple}, $J'_{n+1}(v_{n+1}-)=J'_{n+1}(v_{n+1}+)=G'_{n+1}(v_{n+1})$. Thus, $f_{n+1}(v_{n+1})=0$. Therefore, we have $f'(v_{n+1}) = G''_{n+1}(v_{n+1}) + 2v_{n+1}^2-\text{mse}_{\text{opt}}  \ge 0$ and $f'(w) > 0$ for all $w > v_{n+1}$, and $f'(w) < 0$ for $w\in (0,v_{n+1})$. This ends the proof of Lemma~\ref{derivative1}(b).

\subsection{Proof of Lemma~\ref{derivative1}(d)}\label{app-lemmaderivative1d}

Now we show Lemma~\ref{derivative1}(d). 
We now use induction to show that $v_n \le \sqrt{3 \text{mse}_{\text{opt}} } $ for all $n=1,2,\cdots$.   
 Note that $v_1 = \sqrt{3\beta} = \sqrt{3(\text{mse}_{\text{opt}} - \mathbb{E} \left[ Y \right])}\le \sqrt{3\text{mse}_{\text{opt}}}$.  
The second threshold $v_2$ is the root of  
 \begin{align}
 0 = & \frac{2}{3} w^3 - 2\beta w + \alpha  \mathbb{E} \left[  \partial_x g( w+W_{Y_1},v_1  )    \right]. \label{eq-v2}
 \end{align}
Note that $2 \text{mse}_{\text{opt}} w - \frac{2}{3} w^3$ is positive at $0\le w \le \sqrt{3 \text{mse}_{\text{opt}} }$. Therefore, if $v \le \sqrt{3 \text{mse}_{\text{opt}} }$, $\partial_x g(w,v)$ is always positive at $w \ge 0$.
Since $v_1 \le \sqrt{3\text{mse}_{\text{opt}}}$, $\partial_x g(w,v_1) \ge 0$ for all $w \ge 0$. Recall that $\partial_x g(w,v)$ is an odd function on $w$ for any $v \ge 0$. 
Therefore, utilizing the same analysis as \eqref{eq-g3ninequality2}, for all $w \ge 0$,
\begin{align}
  \mathbb{E} \left[  \partial_x g( w+W_{Y_1},v_1  )    \right] \ge 0.
\end{align}
The first term $ \frac{2}{3} w^3 - 2\beta w >0$ for all $w \ge \sqrt{3\text{mse}_{\text{opt}}}$. To keep the equation \eqref{eq-v2} holds, we have $v_2 \le \sqrt{3\text{mse}_{\text{opt}}} $.

Suppose that $v_2,\cdots, v_n \le \sqrt{3\text{mse}_{\text{opt}}}$.
Now, we will show that $v_{n+1} \le \sqrt{3\text{mse}_{\text{opt}}}$. Note that $v_{n+1}$ is the root of  
 \begin{align}
\nonumber
0 = & \frac{2}{3} w^3 - 2\beta w + \alpha  \mathbb{E} \left[  \partial_x g( w+W_{Y_1},v_n  )    \right] \\ & +
\sum_{k=2}^{n} \alpha^k \mathbb{E} \left[  \partial_x g( w+W_{Y_1}+\cdots W_{Y_{k}}, v_{n+1-k} ) {\text{\large $\mathds{1}$}}_{ \left\{ |w+ Y_1| \ge v_n, \cdots, |w+W_{Y_1}+\cdots+W_{Y_{k-1}} | \ge v_{n+2-k}\right\} }    \right].
\end{align}
Since the hypothesis tells that $v_n \le \sqrt{3\text{mse}_{\text{opt}}}$, $ \mathbb{E} \left[  \partial_x g( w+W_{Y_1},v_n  )    \right] \ge 0$.
To show that $v_{n+1} \le \sqrt{3\text{mse}_{\text{opt}}} $, it is sufficient to show that for $k=2,\cdots,n$, 
\begin{align}
& \mathbb{E} \left[  \partial_x g( w+W_{Y_1}+\cdots W_{Y_{k}}, v_{n+1-k} ) {\text{\large $\mathds{1}$}}_{ \left\{ |w+ Y_1| \ge v_n, \cdots, |w+W_{Y_1}+\cdots+W_{Y_{k-1}} | \ge v_{n+2-k}\right\} }    \right] \ge 0. \label{eq-der1single>0}
\end{align}
Since $v_{n+1-k} \le   \sqrt{3\text{mse}_{\text{opt}}}$, $\partial_x g(w,v_{n+1-k})\ge 0$ for all $w \ge 0$. Therefore, the inequality \eqref{eq-der1single>0} is shown by Lemma~\ref{lemma-prerequisite}(c) and that $\partial_x g(w,v) $ is an odd function for any $v\ge0$.

Now, we will jointly show that $G_{n+1}'(w) \ge G_{n}'(w)$, and $v_{n+1} \le v_n$. First, $v_1 = \sqrt{3\beta}, G'_1(w) = 2  \mathbb{E} \left[ Y  \right] w, G'_2 (w) = G'_1(w) + \alpha  \mathbb{E} \left[ \partial_x g(w+W_Y, v_1)   \right]$. Since $v_1 \le \sqrt{3\text{mse}_{\text{opt}}}$, $ \mathbb{E} \left[ \partial_x g(w+W_Y, v_1)   \right] \ge 0$, and we directly have $G'_2(w) \ge G'_1(w)$. For simplicity, let us define $f_n(w)$ as
\begin{align}
f_{n}(w)=G'_{n}(w) -( -\frac{2}{3}w^3+2 \text{mse}_{\text{opt}} w).
\end{align}
If $v_{2} > v_1$, then we have $f_1(v_{2}) \le f_{2}(v_{2})=0$, which contradicts to $f_1(w) >0$ for $w > v_1$. Therefore, $v_{2} \le v_1$. 

Then, suppose that $G'_n(w) \ge G'_{n-1}(w)$ for $w \ge 0$, and $v_n \le v_{n-1}$. We have 
\begin{align}
\nonumber G'_{n+1}(w) - G'_n(w) & = \alpha  \mathbb{E} \left[ J'_n(w+W_Y) - J_{n-1}'(w+W_Y)  \right]\\
\nonumber & =  \alpha  \mathbb{E} \big{[} G'_n(w+W_Y) {\text{\large $\mathds{1}$}}_{|w+W_Y |\ge v_n} - G'_{n-1}(w+W_Y) {\text{\large $\mathds{1}$}}_{|w+W_Y |\ge v_{n-1}} \\ & + J'_n(w+W_Y) {\text{\large $\mathds{1}$}}_{|w+W_Y | < v_n} - J'_{n-1}(w+W_Y) {\text{\large $\mathds{1}$}}_{|w+W_Y | < v_{n-1}} \big]. \label{eq-179}
\end{align} Note that $G'_n (w)$ is odd.
If $|w+W_Y| \ge v_{n-1}$, utilizing the same analysis as \eqref{eq-g3ninequality2}, we have 
\begin{align}
\nonumber  & \mathbb{E} \left[ (J'_n(w+W_Y) - J_{n-1}'(w+W_Y)) {\text{\large $\mathds{1}$}}_{|w+W_Y | \ge v_{n-1}}   \right]  \\ = &
 \mathbb{E} \left[  (G'_n(w+W_Y) - G'_{n-1}(w+W_Y)) {\text{\large $\mathds{1}$}}_{|w+W_Y | \ge v_{n-1}} \right] \ge 0. \label{eq-180-1}
\end{align}
If $|w+W_Y| < v_{n}$, we have $J'_n(w+W_Y) = 2\text{mse}_{\text{opt}} (w+W_Y) - 2/3 (w+W_Y)^3 $, irrelevant to $n$. Thus
\begin{align}
\nonumber & \mathbb{E} \left[ (J'_n(w+W_Y) - J_{n-1}'(w+W_Y)) {\text{\large $\mathds{1}$}}_{|w+W_Y | < v_{n}}   \right] \\ = & \mathbb{E} \left[
J'_n(w+W_Y) {\text{\large $\mathds{1}$}}_{|w+W_Y | < v_n} - J'_{n-1}(w+W_Y) {\text{\large $\mathds{1}$}}_{|w+W_Y | < v_{n}} \right]= 0. \label{eq-180-2}
\end{align}
If $v_n \le |w+W_Y| < v_{n-1}$, 
\begin{align}
\nonumber & \mathbb{E} \left[ (J'_n(w+W_Y) - J_{n-1}'(w+W_Y)) {\text{\large $\mathds{1}$}}_{v_{n} \le |w+W_Y | < v_{n-1}}   \right] \\ \nonumber = &  \mathbb{E} \left[ (G'_n(w+W_Y) - J'_{n-1}(w+W_Y)) {\text{\large $\mathds{1}$}}_{v_{n} \le |w+W_Y | < v_{n-1}}   \right] \\ = &  \mathbb{E} \left[ (G'_n(w+W_Y) -  2\text{mse}_{\text{opt}} (w+W_Y) + \frac{2}{3} (w+W_Y)^3 ) {\text{\large $\mathds{1}$}}_{v_{n} \le |w+W_Y | < v_{n-1}}   \right] \ge 0. \label{eq-180-3}
\end{align}
The last inequality holds because $f_n(x)$ is an odd function and non-negative for $x \ge v_n$. Inserting \eqref{eq-180-1},\eqref{eq-180-2},\eqref{eq-180-3} into \eqref{eq-179}, we finally have $G'_{n+1}(w) \ge G'_n(w)$ for $w\ge 0$. 

Recall that the free boundary method \eqref{eq-fbm3simple} implies that $v_{n+1}$ is the root of $f_{n+1}'(w)=G'_{n+1}(w) -( -\frac{2}{3}w^3+2 \text{mse}_{\text{opt}} w)=0$. Note that $f_{n}'(w) = G'_{n}(w) -( -\frac{2}{3}w^3+2 \text{mse}_{\text{opt}} w) \le f_{n+1}'(w)  $. If $v_{n+1} > v_n$, then we have $f_n'(v_{n+1}) \le f_{n+1}'(v_{n+1})=0$, which contradicts to $f_n'(w) >0$ for $w > v_n$. Therefore, $v_{n+1} \le v_n$, and we have that $\{v_n\}_n$ is decreasing. 

\section{Proof of Lemma~\ref{jsmallg}}\label{appjsmallg}

By \eqref{eq-fbm2}, $ \tilde{J}_n(w,q) = \tilde{G}_n(w,q)$ if $|w|> v_n$. It remains to show that $ \tilde{J}_n(w,q) \le \tilde{G}_n(w,q)$ for $|w| \le v_n$ (by symmetry, we will assume $w\ge0$).

Define $f(w)\triangleq \tilde{G}_n(w,q) - \tilde{J}_n(w,q) = G_n(w) - J_n(w)$. It is easy to see that $f(v_n)=0$, and $f(w)$ is not a function of $y$. 

By Lemma~\ref{derivative1}, we have shown that $J'_n(w) =  -2/3w^3+2 \text{mse}_{\text{opt}} w, J''_n(w) =  -2w^2+2 \text{mse}_{\text{opt}}, J'''_n(w) =  -4w $ if $|w|\le v$.
Therefore, by Lemma~\ref{derivative1}(c), $f'''(w)\ge 0$ for $w\in[0,v]$. This implies that $f'(w)$ is convex in $w\in[0,v_n]$. Since $f'(0)=f'(v_n)=0$, we have $f'(w)\le 0$ for $w\in[0,v_n]$. Note that $f(v_n)=0$. So $f(w)$ is non-increasing in $w\in[0,v_n]$, and thus $f(w) \ge 0$ for $w\in[0,v_n]$. This implies that $\tilde{G}(w,q) - \tilde{J}(w,q) \ge 0$, which completes our proof.

\section{Proof of Lemma~\ref{j-is-excessive}}\label{appj-is-excessive}

$\tilde{J}_n(w,q)$ is continuously differentiable, and twice condinuously differentiable except at $(\pm v_n,q)$. However, since the Lebesgue measure of reaching $(\pm v_n,q)$ is zero, the values $ \frac{\partial^2}{\partial w^2}  \tilde{J}_n(\pm w,q)$ can be chosen in the sequel arbitrary \cite[Section~10]{peskir2006optimal}. 

In Lemma~\ref{derivative1}, it is easy to see that $ \partial_x \tilde{J}_n(w,q) = J' (w)$, not a function of $q$, and $ \partial_x  \tilde{J}_n(w,q) = O(w)$. Therefore, for any given time $t$, 
\begin{align}
\mathbb{E} \left[  \int_0^t \left[ \partial_x  \tilde{J}_n(w+W_r, q+Q_r) \right]^2  dt \right] <\infty.
\end{align}
The integral $\int_0^t W_r^2 dr$ is increasing in $t$. Using It\^o's formula \cite[Theorem~7.14]{morters2010brownian}, almost surely, 
\begin{align}
\nonumber  & \tilde{J}_n(w+W_t,q+Q_t) -  \tilde{J}_n(w,q) \\ = &  \int_0^t  (w+W_r)^2 - \text{mse}_{\text{opt}}+ \frac{1}{2} \partial_{xx} \tilde{J}_{n}(w+W_r, q+Q_r) dr +  \int_0^t  \partial_x  \tilde{J}_n(w+W_r,q+Q_r) dW_r.
\end{align}
By \cite[Theorem~7.11]{morters2010brownian}, the process $ \int_0^t \partial_x  \tilde{J}_n(w+W_r,q+Q_r)  dW_r$ is a martingale and thus
\begin{align}
\mathbb{E} \left[  \int_0^t \partial_x \tilde{J}_n(w+W_r,q+Q_r)  dW_r  \right] = 0.
\end{align}
Therefore, 
\begin{align}
 \mathbb{E} \left[  \tilde{J}_n(w+W_t,q+Q_t)  \right]  -  \tilde{J}_n(w,q)  &  =  \mathbb{E} \left[  \int_0^t  (w+W_r)^2 - \text{mse}_{\text{opt}}+\frac{1}{2}\partial_{xx}  \tilde{J}_{n}(w+W_r,q+Q_r) dr \right]. \label{eq-ito}
\end{align}
If $|w+W_r|<v_n$, according to Lemma~\ref{derivative1}(a) we have $\partial_x  \tilde{J}_{n}(w+W_r, q+Q_r) = -2/3 (w+W_r)^3 + 2\text{mse}_{\text{opt}} (w+W_r)$, and $\partial_{xx}  \tilde{J}_{n}(w+W_r, q+Q_r) = -2 (w+W_r)^2 + 2\text{mse}_{\text{opt}}$ (correspond to the first equation of free boundary method \eqref{eq-fbm1}). Therefore,  
\begin{align}
 (w+W_r)^2 - \text{mse}_{\text{opt}}+\partial_{xx} \tilde{J}_{n}(w+W_r, q+Q_r) =0.
\end{align}
If $|w+W_r| \ge v_n$, according to Lemma~\ref{derivative1}(a),(b), we get 
\begin{align}
\nonumber & 2((w+W_r)^2 - \text{mse}_{\text{opt}})+\partial_{xx}  \tilde{J}_{n}(w+W_r, q+Q_r) \\
= &  2((w+W_r)^2 - \text{mse}_{\text{opt}})+ G_n''(w+W_r) \ge 0.
\end{align} 
Applying to \eqref{eq-ito}, we get $\mathbb{E}^{(w,y)}[\tilde{J}_n(w+W_t,q+Q_t)]\ge \tilde{J}_n(w,q)$. This ends our proof.

\section{Proof of Lemma~\ref{lemma-contraction}}\label{app-lemma-contraction}
Note that for simplicity, we have set $X_\tau = w+ W_\tau $ as a Wiener process that starts from $X_0=w$. When $w^2 < \bar{b}$, 
\begin{align}
\nonumber \frac{v(X_\tau+W_Y)}{v(w)} = & \frac{1}{\bar{b}} \max\{(X_\tau+W_Y)^2,\bar{b}\} = \max\{ \frac{X_\tau^2}{\bar{b}} + \frac{2X_\tau W_Y}{\bar{b}} + \frac{W_Y^2}{\bar{b}},1 \} \\ \nonumber \le & \max\{ 1+\frac{2\sqrt{\bar{b}}}{\bar{b}} |W_Y| + \frac{W_Y^2}{\bar{b}},1\} \\ = & 1+\frac{2 |W_Y|}{\sqrt{\bar{b}}}  + \frac{W_Y^2}{\bar{b}}.
\end{align}
When $w^2 < \bar{b}$, $X_\tau=w$, and we have 
\begin{align}
\nonumber  \frac{v(X_\tau+W_Y)}{v(w)} = & \frac{1}{w^2} \max\{(w+W_Y)^2,\bar{b}\} = \max\{ 1 + \frac{2 W_Y}{w} + \frac{W_Y^2}{w^2},\frac{\bar{b}}{w^2} \} \\ \nonumber \le & \max\{ 1+\frac{2}{\sqrt{\bar{b}}} |W_Y| + \frac{W_Y^2}{\bar{b}},1\} \\ = & 1+\frac{2 |W_Y|}{\sqrt{\bar{b}}}  + \frac{W_Y^2}{\bar{b}}.
\end{align}
This ends the proof of Lemma~\ref{lemma-contraction}(a). 

Note that $g(w; \tau)$ is bounded in $w\in [-\bar{b},\bar{b}]$, and $g(w; \tau) = 2\mathbb{E} \left[ Y \right] w^2 + const$ for $|w|^2 \ge \bar{b}$. Therefore, there exists $k>0$ such that $\|g(w;\tau)\| \le k$. Recall that we denote $\tilde{W}_{n+1}$ as the state value at $n+1$th stage with $\tilde{W}_1 = w$, and $\tilde{W}_{n+1} = \tilde{W}_{n}+W_{\tau'} + W_Y$ for a stopping time $\tau'$. Then,
\begin{align}
\nonumber \mathbb{E} \left[ g(W_{n+1};\tau) \right] = & \mathbb{E} \left[ g(\tilde{W}_{n}+W_{\tau'} + W_Y;\tau) \right]
\le  k \mathbb{E} \left[ v(\tilde{W}_{n}+W_{\tau'} + W_Y) \right] \\
\nonumber \le & k \frac{\rho}{\alpha} \mathbb{E} \left[ v(\tilde{W}_n) \right]  
 \le  k \frac{\rho^2}{\alpha^2}  \mathbb{E} \left[ v(\tilde{W}_{n-1}) \right] \cdots \\
\le & k \frac{\rho^n}{\alpha^n} v(w). 
\end{align}

We have shown that each optimal stopping times for solving $T(T^n 0(w))$ are some hitting times with bounded and convergent thresholds, so each stopping time belongs to the assumption of Lemma~\ref{lemma-contraction}(a). We have 
\begin{align}
\nonumber T^n 0(w) = & \sum_{k=1}^{n} \alpha^{k-1}  \mathbb{E} \left[ g(\tilde{W}_k;\tau_k) \right]
\le   \sum_{j=1}^{n} \rho^k v(w) \le \frac{1}{1-\rho} v(w),\\
T^\infty 0(w) \le &  \sum_{k=1}^{\infty} \rho^k v(w) \le \frac{1}{1-\rho} v(w).
\end{align} 
Thus, both $\|T^n 0\|$ and $\| T^\infty 0\|$ are finite. For any stopping time $\tau$ within the assumption of Lemma~\ref{lemma-contraction}(a), 
\begin{align}
\nonumber T_\tau T^n 0(w) - T_\tau T^{n-1} 0(w) = & \alpha \mathbb{E} \left[ T^n 0(X_\tau + W_Y) - T^{n-1} (X_\tau + W_Y) \right]  \\
\nonumber = &  \mathbb{E} \left[ \frac{ T^n 0(X_\tau + W_Y) - T^{n-1} (X_\tau + W_Y)}{v(X_\tau + W_Y)} v(X_\tau + W_Y) \right] \\
\nonumber \le & \alpha \| T^n0 - T^{n-1} 0 \| \times \mathbb{E} \left[  v(X_\tau + W_Y)  \right] \\
\le & \rho v(w) \| T^n0 - T^{n-1} 0 \|.
\end{align}
This gives that 
\begin{align}
\frac{T_\tau T^n 0(w)}{v(w)} \le \rho \| T^n 0 - T^{n-1} 0 \| + \frac{T_\tau T^{n-1} 0(w)}{v(w)}.\label{eq-contraction1a}
\end{align}
Take the minimum for left and right side of \eqref{eq-contraction1a} over all the stopping times $\tau$ with bounded threshold $\bar{b}$, then
\begin{align}
\frac{T^{n+1} 0(w)}{v(w)} \le \rho \| T^n 0 - T^{n-1} 0 \| + \frac{T^{n} 0(w)}{v(w)}.\label{eq-contraction1}
\end{align} By symmetry, 
\begin{align}
\frac{T^{n} 0(w)}{v(w)} \le \rho \| T^n 0 - T^{n-1} 0 \| + \frac{T^{n+1} 0(w)}{v(w)}.\label{eq-contraction1-5}
\end{align}
Therefore,
\begin{align}
 \|T^{n+1} 0 - T^n 0\| \le \rho \|T^n 0 - T^{n-1} 0\|\cdots \le \rho^n \|T 0\|.\label{eq-contraction2}
 \end{align}
This completes the proof of Lemma~\ref{lemma-contraction}(b).  
Due to \eqref{eq-contraction2}, the sequence $\{ T^n 0(w)\}$ is a Cauchy sequence, and thus $T^n 0(w)$ converges pointwise to $T^\infty 0(w)$, which is also measurable, and we have shown that $\| T^\infty 0(w) \| <\infty$. Therefore, using \cite[pp. 47]{bertsekas1995dynamic2}, $\| T^n 0 - T^\infty 0\| \rightarrow 0$. We replace $T^n 0(w)$ by $T^\infty 0(w)$ in \eqref{eq-contraction1a} and use symmetry, we then find that\footnote{Here, we do not require $T T^\infty 0(w)$ to be measurable.} 
\begin{align}
\| T T^\infty 0 - T^n 0 \| \le \rho \| T^\infty 0 - T^{n-1} 0\| \rightarrow 0. 
\end{align}
Therefore, $J^* = T^\infty 0$ is the solution to the Bellman equation $T J^* = J^*$, and the $\rho-$convergence rate is immediately given. The solution is unique: If there exists any other measurable function $\tilde{J}(w)$ that satisfies the Bellman equation (with $\| \tilde{J}(w) \|<\infty$), we replace $T^n 0(w)$ by $T^\infty 0(w)$ and replace $T^{n-1}0$ by $\tilde{J}(w)$ in \eqref{eq-contraction1a}, and we have
\begin{align}
\| T^\infty 0 - \tilde{J} \| = \| T T^\infty 0 - T \tilde{J} \| \le \rho \| T^\infty 0 - \tilde{J} \|,
\end{align} which implies that $T^\infty 0 = \tilde{J}$. These completes the proof of Lemma~\ref{lemma-contraction}(c).

\section{Proof of Theorem~\ref{thm2}}\label{appprop-thm2}

We denote $\Pi_{j,\text{signal-agnostic}} \subset \Pi_{\text{signal-agnostic}}$ as a collection of sampling times $S_{j,1},S_{j,2},\cdots$ at $j$th epoch such that the inter-sampling times $S_{j,1} - S_{j-1,M_{j-1}},S_{j,2}  - S_{j-1,M_{j-1}},\ldots$ are independent of the history information before $S_{j-1,M_{j-1}}$. Note that the subscripts $(j,1),(j,2),\ldots$ are illustrated in Section~\ref{proof-preliminary}.

Similar to Proposition~\ref{prop-aware}, we have the following result:

\begin{proposition}\label{prop-agnostic}
There exists an optimal policy in $\Pi_{j,\text{signal-agnostic}}$ such that $\{ S_{j,M_{j}} - S_{j-1,M_{j-1}} \}_j$ are i.i.d. Moreover, problem~\eqref{avg-aware} when $\Pi = \Pi_{\text{signal-agnostic}}$ is equivalent to the following problem:
\begin{align}
\text{mse}_{\text{opt}} = & \inf_{\pi \in \Pi_{j,\text{signal-agnostic}}} \frac{ \mathbb{E} \left[  \int_{D_{j-1,M_{j-1}}}^{D_{j,M_j}} \Delta_t dt \right]}{ \mathbb{E} \left[   D_{j,M_j} - D_{j-1,M_{j-1}} \right] }. \label{eq-singlepoch-age}
\end{align} 
\end{proposition}
The proof of Proposition~\ref{prop-agnostic} is a special case of (thus included in) the proof of Proposision~\eqref{prop-aware} and is omitted. Problem~\eqref{eq-singlepoch-age} has a much simpler form to Problem~\eqref{eq-singlepoch-mse} because (i) the sampling times are independent of the Wiener process, and (ii) we replace the square estimation error $(W_t - \hat{W}_t)^2$ by the linear age $\Delta_t$, the time period between $t$ and the sampling time $S_{j-1,M_{j-1}}$. By \eqref{eq-singlepoch-age} and \cite[Section~V.B]{pan2022optimal}, we complete the proof of Theorem~\ref{thm2}.

\end{document}